\documentclass[opre,dblanonrev]{informs4}
\RequirePackage{tgtermes}
\RequirePackage{newtxtext}
\RequirePackage{bm}
\RequirePackage{endnotes}

\OneAndAHalfSpacedXI

\usepackage{natbib}
 \bibpunct[, ]{(}{)}{,}{a}{}{,}%
 \def\bibfont{\small}%
\EquationsNumberedThrough    

\TheoremsNumberedThrough     
\ECRepeatTheorems  %

\usepackage{algorithm}
\usepackage{algpseudocode}
\usepackage{amsfonts,amsmath,amstext,amssymb,amsopn}
\usepackage{dsfont}
\usepackage{multirow}
\usepackage{bbm}
\usepackage{tikz}
\usepackage[margin=1in]{geometry}
\usepackage{scalerel,stackengine}
\usetikzlibrary{arrows.meta,
            decorations.pathreplacing,
                    calligraphy,
                positioning}
\usepackage{booktabs}

\usepackage{amssymb} 

\newenvironment{myproof}[1][Proof]{%
  \par\noindent\textit{#1.}\quad%
}{%
  \par\noindent\mbox{} \hfill$\square$\par
}

\newcommand{\adaptivetwosided}{{\sf Adaptive Two-Sided Assortment with Revenues}}
\newcommand{\atar}{{\sf(ATAR)}}

\def\revmod{\textsc{Rev-Cost}}
\def\subdual{\textsc{Sub-Dual}}

\makeatletter
\def\BState{\State\hskip-\ALG@thistlm}
\makeatother
\def\one{{\mathds{1}}}
\def\RR{{\mathbb R}}

\def\EE{{\mathbb E}}

\def\PP{{\mathbb P}}

\def\C{{\mathcal C}}

\def\S{{\mathcal S}}

\def\h#1{{#1}^1}

\def\n1{\h{\mathscr{N}}}

\usepackage{makecell}

\colorlet{blue}{black}
\begin{document}


\RUNTITLE{Adaptive Two-Sided Assortment Optimization: Revenue Maximization}

\TITLE{Adaptive Two-Sided Assortment Optimization: Revenue Maximization}

\ARTICLEAUTHORS{%
\AUTHOR{Mohammadreza Ahmadnejadsaein, Omar El Housni}
\AFF{School of Operations Research and Information Engineering, Cornell Tech, Cornell University\\
\EMAIL{\{ma932, oe46\}@cornell.edu}}

} 

\ABSTRACT{%
We study adaptive two-sided assortment optimization for revenue maximization in choice-based matching platforms. The platform has two sides of agents: an initiating side and a responding side. The decision-maker sequentially selects agents from the initiating side, shows each an assortment of agents from the responding side, and observes their choices. After processing all initiating agents, the responding agents are shown assortments and make their selections. A match occurs when two agents mutually select each other, generating pair-dependent revenue. Choices follow Multinomial Logit (MNL) models. This setting generalizes prior work focused on maximizing the number of matches under submodular demand assumptions, which do not hold in our revenue-maximization context.
{\color{blue}We show that the revenue maximization problem is APX-hard, so constant-factor loss is unavoidable.} Our main contribution is the design of polynomial-time approximation algorithms with constant-factor guarantees. In particular, for general pairwise revenues, we develop a randomized algorithm that achieves a $(\frac{1}{2}-\epsilon)$-approximation in expectation for any $\epsilon > 0$. The algorithm is static and provides guarantees under various agent arrival settings, including fixed order, simultaneous processing, and adaptive selection. When revenues are uniform across all pairs involving any given responding-side agent, the guarantee improves to $(1 - \frac{1}{e} - \epsilon)$. In structural settings where responding-side agents share a common revenue-based ranking, we design a simpler adaptive deterministic algorithm achieving a $\frac{1}{2}$-approximation. Our approach leverages novel linear programming relaxations, correlation-gap arguments, and structural properties of the revenue functions.
}%

\KEYWORDS{Assortment Optimization, Two-Sided Platforms, Revenue Maximization, Multinomial Logit Model} 

\maketitle


\section{Introduction}\label{sec:introduction}
Two-sided online platforms such as Upwork, Airbnb, Uber, and LinkedIn have profoundly transformed modern marketplaces by facilitating interactions between two distinct user groups, typically customers seeking services and suppliers offering them. These platforms streamline operations across domains including transportation, accommodations, employment, and social connections. A fundamental challenge in two-sided platforms is managing the heterogeneous preferences of agents, which can create imbalanced demand patterns. Specifically, when agents on one side have strong preferences towards a limited set of highly desirable participants from the opposite side, such as freelancers with exceptional qualifications or highly rated listings in accommodation markets, this concentration of demand leads to \emph{choice congestion}. Choice congestion occurs when popular agents receive more requests than they can feasibly handle, negatively impacting both overall system efficiency and the platform's revenue. Consequently, the performance of these markets critically depends on platform design and participant preferences. Such dynamics have been extensively explored in two-sided matching contexts and specific applications (see, e.g.,~\citet{rochet2003,fradkin2015,manshadi2022}). 

{\color{blue}Two-sided online platforms differ in how selective the displayed assortment is. Some show nearly every alternative that meets a search filter. Others select a small subset from a large pool, as on a freelancing platform with millions of active listings, where the assortment decision drives the outcome.} Two-sided assortment optimization has emerged as a powerful tool to improve the efficiency of two-sided platforms by optimizing the assortments displayed to agents on either side of the market. Recent papers have explored this framework in various settings (e.g.,~\citet{ashlagi2022assortment,torrico2023multiagent,aouad2023,housni2024twosided}), highlighting its potential to mitigate inefficiencies such as choice congestion and imbalanced demand.
Most of this literature considers a two-sided market consisting of an initiating side and a responding side. The platform, acting as a central decision maker, presents each initiating-side agent with an assortment of agents from the responding side based on their preferences, and observes their choices. After all initiating-side agents are processed, each responding-side agent is shown an assortment and makes a selection. These works typically assume that each agent selects a single counterpart from the offered set, as in rental markets, although other settings allow multiple selections, such as in dating platforms. 

The design of the platform, specifically how agents interact with the system and how the decision maker controls displayed assortments, gives rise to various optimization settings and policies. For example, platforms differ in whether assortments are static or adaptive, uniform or personalized, and whether agents are processed in a fixed, dynamic, or chosen order. These structural choices directly shape the optimization problem faced by the platform and motivate diverse policy classes.
\citet{housni2024twosided} categorize two-sided assortment policies into several classes based on how the platform sequences agents and adapts assortments. One broad class is that of \emph{one-sided} policies, where the platform processes all agents on one side before serving the other side. A key subclass is \emph{static one-sided} policies, in which the decision maker simultaneously offers assortments to all agents on the initiating side, observes their choices, and then processes the responding side. This class was introduced by~\citet{ashlagi2022assortment} and further developed by~\citet{torrico2023multiagent}. Another important subclass is \emph{dynamic one-sided} policies, where agents on the initiating side are processed sequentially, and assortments are adapted based on revealed choices. Within this subclass, the arrival order may be stochastic, adversarial, or chosen by the platform. For instance, \citet{aouad2023} study an online model with sequential customer arrivals and time-varying supplier availability, both potentially adversarial. \citet{housni2024twosided} analyze the \emph{one-sided adaptive} setting that is a variant of dynamic one-sided policies in which the platform controls the order of agent arrivals and adaptively offers assortments to maximize the expected number of matches.

All the works above focus on maximizing the expected number of successful matches, where a match occurs when two agents from opposite sides select each other. However, for the platform, it is often more desirable to design algorithms that explicitly target revenue maximization because maximizing the number of matches does not necessarily guarantee high revenue under heterogeneous customer-supplier pair revenues.
{\color{blue}Some of these works extend to very special cases of revenues, where the revenue of a customer-supplier pair depends on only one side of the market. We focus on the most general case, where the revenue can depend on both agents in the pair.}

We are aware of two papers that focus on maximizing revenue in static settings with general pairwise revenues. \citet{ahmed2022parameterized} study maximizing expected revenue within the class of {\em fully static} policies. However, their work does not provide constant-factor algorithms, offering only parameterized guarantees. The other is a recent concurrent work by \citet{danny2024}, who study revenue maximization under static one-sided policies, where each customer-supplier pair is associated with a specific revenue. They consider two settings: a \emph{customized model}, where the platform can personalize the set of selecting customers shown to each supplier, and an \emph{inclusive model}, where the supplier must see all customers who selected them.
 In this paper, we study the class of {\bf adaptive} one-sided policies in two-sided platforms with the goal of maximizing expected revenue. {\color{blue}The model we study is presented below.}

\subsection{Model}\label{sec:model}
We consider a two-sided platform represented by a weighted bipartite graph \(G = (\mathcal{C}, \mathcal{S}, E)\), where \(\mathcal{C}\) and \(\mathcal{S}\) refer to the customers and suppliers, with \(|\mathcal{C}| = n\) and \(|\mathcal{S}| = m\), and \(E = \mathcal{C} \times \mathcal{S}\) the set of customer-supplier pairs. Each edge \((i, j) \in E\) has a non-negative weight \(r_{ij}\), representing the revenue accrued by the platform upon a successful match between customer \(i\) and supplier \(j\).

The platform begins from an initiating side, either customers or suppliers, and processes each agent on that side adaptively, then proceeds with adaptive processing of agents on the responding side, where each agent may select at most one agent from the set of offers they receive. We assume without loss of generality that customers are the initiating side and suppliers are the responding side. In each step, the platform selects a customer, offers her an assortment consisting of a subset of suppliers, observes her choice, and determines the next customer to serve. After all agents on the initiating side have been processed, the platform proceeds to process the responding side: it offers each agent a subset of agents from the initiating side, restricted to those who previously selected that agent during initiating steps. {\color{blue}A customer-supplier pair \((i,j)\) forms a \emph{successful match} if the two agents select each other, i.e., customer \(i\) selects supplier \(j\) during the initiating steps and supplier \(j\) then selects customer \(i\) from the assortment of customers it is offered.} Each realized successful match generates revenue corresponding to the weight of the matched customer-supplier pair in the bipartite graph, and the platform's objective is to maximize total expected revenue from these successful matches.

Agents on either side, namely customers and suppliers, may select at most one counterpart from an assortment offered by the platform, consisting of a subset of agents from the opposite side. Their selection behavior is governed by individual preference models. In the literature on assortment optimization, such preferences are typically captured using discrete choice models. Specifically, given an assortment of available alternatives, each agent selects an option probabilistically according to a given choice distribution. In this paper, we focus on settings where  all choice preferences are governed by individual \emph{Multinomial Logit (MNL)} models, and refer to this setting as the {\color{blue}\emph{\adaptivetwosided}}~\atar~problem. For completeness, we present a mathematical formulation of \atar\ using dynamic programming in  Section~\ref{sec:DP_formulation}.

{\color{blue}The model we study adopts the single-selection convention above and, like others in this literature, is a stylized abstraction of platform operation: it captures the two-sided interaction, choice congestion, and substitution effects central to these platforms while simplifying other features of real marketplaces. The primary contribution of this paper is algorithmic. We also evaluate the proposed algorithms numerically on instances calibrated to real marketplace data in Section~\ref{sec:numerics}. Finally, we discuss the modeling assumptions and limitations of our framework in Section~\ref{sec:conclusion}.}

{\color{blue}Although~\atar~lets the platform observe each customer's choice before processing the next, the guarantee we present for our main algorithm does not rely on this feedback: that algorithm fixes every assortment in advance, so the same guarantee holds when the processing order is fixed beforehand, when all customers are processed at once, and when customers arrive in an adversarial or stochastic order. We describe the three processing variants and their relationship in Remark~\ref{remark:variants} and Figure~\ref{fig:arrival-order-comparision}.}

\vspace{2mm}
\noindent
{\bf Challenges in Designing Algorithms for~\atar.}
Most prior work on two-sided assortment optimization focuses on maximizing the number of matches and leverages the submodularity of agents’ demand functions, defined as the probability of selecting an item from a given assortment. To the best of our knowledge, the literature has primarily studied a special case of \atar\ where \( r_{ij} = 1 \) for all \( i \in \mathcal{C} \) and \( j \in \mathcal{S} \). For this special case of~\atar, \citet{housni2024twosided} present a greedy \(\frac{1}{2}\)-approximation algorithm. At each step, their algorithm selects a customer arbitrarily and offers her the assortment that maximizes the marginal increase in suppliers’ demand functions. Crucially, their primal-dual analysis exploits the submodularity and monotonicity of demand functions to prove that the expected number of matches for the greedy algorithm is at least half of an LP relaxation. In online assortment optimization for two-sided matching platforms, \citet{aouad2023}, leveraging the submodularity and monotonicity of each supplier’s demand function, proposed a greedy algorithm that achieves a constant-factor competitive ratio of \(1/2\) in their setting.

In contrast, our setting introduces general pairwise revenues, which break the submodularity of the revenue function associated with each agent and pose a fundamental obstacle, as this precludes the direct application of classical greedy or myopic algorithms that rely on the diminishing returns property. As a result, achieving the same constant-factor guarantee (i.e., $\frac{1}{2}$-approximation) in the~\atar~setting presents significant challenges. The goal of this paper is to address these challenges and design polynomial-time algorithms with constant theoretical guarantees for~\atar.

\subsection{Main Results}\label{sec:main-results}
In this paper, we develop two polynomial-time constant-factor approximation algorithms for~\atar. Unlike most prior work that focuses on maximizing the number of matches, we directly optimize expected revenue, allowing heterogeneous, nonnegative revenues $r_{ij}$ for each customer-supplier pair $(i,j) \in \mathcal{C} \times \mathcal{S}$. Our main contributions are as follows.

\vspace{2mm}
\noindent
{\bf Algorithm for~\atar~with General Revenues.} In the general setting with heterogeneous nonnegative revenues $\{r_{ij}\}_{i\in\mathcal{C}, j\in\mathcal{S}}$, we design a randomized polynomial-time algorithm that achieves a \((\frac{1}{2} - \epsilon)\)-approximation guarantee for~\atar\  for any $\epsilon>0$. Furthermore, we show in the special case where the revenues are uniform on each responding-side agent, meaning that for each supplier \(j\in\mathcal{S}\) there is a nonnegative \(r_j\) such that \(r_{ij}=r_j\) for all \(i\in\mathcal{C}\), our randomized algorithm guarantees a $(1 - \frac{1}{e} - \epsilon)$-approximation for any $\epsilon>0$. Our first main result is given in the following theorem.

\begin{theorem}
\label{thm:general_revenues}
For any $\epsilon > 0$, there exists a randomized static algorithm achieving a \((\frac{1}{2} - \epsilon)\)-approximation for~\atar~in the general case with heterogeneous revenues, and a \((1 - \frac{1}{e} - \epsilon)\)-approximation when revenues are uniform across all pairs involving any given supplier.
\end{theorem}

To prove Theorem~\ref{thm:general_revenues}, we introduce a novel linear programming (LP) relaxation of~\atar, using exponentially many variables. We show that the dual of this formulation admits a fully polynomial-time approximation scheme (FPTAS) for its separation oracle, allowing us to approximately solve the LP relaxation using the Ellipsoid method. We then derive an upper bound on the correlation gap of nonnegative, monotone, and {\em submodular order} functions, a class of functions that appear in the objective of~\atar. By leveraging the approximate LP solution and the correlation-gap bound on the \emph{optimal revenue function} of each supplier, our randomized static algorithm offers a random subset of suppliers to each customer non-adaptively, and we show this gives a \((\frac{1}{2} - \epsilon)\)-approximation guarantee for~\atar. Our analysis and proof of Theorem~\ref{thm:general_revenues} are presented in Section~\ref{sec:1/2-epsilon-approx}.

Our guarantee in Theorem \ref{thm:general_revenues} improves on the previous \(\tfrac13\)-approximation for the customized setting in the one‐sided static problem with pairwise revenues presented by \citet{danny2024}. Furthermore, in the special case of uniform revenues for each supplier, which includes the matching maximization problem in this context, our
\(\bigl(1 - \tfrac1e - \epsilon\bigr)\)-approximation improves the \(\tfrac12\)-approximation guarantee proposed by \citet{housni2024twosided} for the case of MNL choice models.
{\color{blue}Our randomized algorithm also admits a revenue-dependent guarantee. Let $r_{\min}=\min\{r_{ij}: r_{ij}>0\}$ and $r_{\max}=\max_{i,j} r_{ij}$ denote the smallest positive and the largest pairwise revenues. Then, for any $\epsilon>0$, the algorithm achieves a $\bigl(\max\bigl\{\tfrac12,\ \tfrac{r_{\min}}{r_{\max}}\bigl(1-\tfrac1e\bigr)\bigr\}-\epsilon\bigr)$-approximation for~\atar. 
This bound never falls below the general $\bigl(\tfrac12-\epsilon\bigr)$
guarantee and recovers the $\bigl(1-\tfrac1e-\epsilon\bigr)$ guarantee for
uniform revenues.}

\vspace{2mm}
\noindent
{\bf Algorithm for~\atar~with Same-Order Revenues.} We strengthen our result to an exact \(\frac{1}{2}\)-approximation under an additional widely used structural assumption on the revenues. Specifically, we consider a same-order condition on pair revenues which means  that all suppliers rank customers in the same revenue order, i.e., under this condition there exists a permutation $\sigma_{\mathcal{C}}$ of customers $\mathcal{C}$ such that $r_{\sigma_{\mathcal{C}}(1),j}\ge r_{\sigma_{\mathcal{C}}(2),j}\ge \cdots\ge r_{\sigma_{\mathcal{C}}(n),j}$ for every $j\in\mathcal{S}$. For example, any revenue structure of the forms \( r_{ij} = r_i + r_j \), \( r_{ij} = r_i \cdot r_j \), \( r_{ij} = r_i \), \( r_{ij} = r_j \), \( r_{ij} = \max\{r_i, r_j\} \), or \( r_{ij} = \min\{r_i, r_j\}\) for some \( \{r_i\}_{i \in \mathcal{C}} \) and \( \{r_j\}_{j \in \mathcal{S}} \), satisfies the same-order condition. Our second main result is given in the following theorem.

\begin{theorem}
\label{thm:separable_revenues}
Suppose all suppliers rank customers in the same order with respect to the revenue on each edge of the bipartite graph, i.e., there is a permutation $\sigma_{\mathcal{C}}$ of customers $\mathcal{C}$ such that $r_{\sigma_{\mathcal{C}}(1),j}\ge r_{\sigma_{\mathcal{C}}(2),j}\ge \cdots\ge r_{\sigma_{\mathcal{C}}(n),j}$ for every $j\in\mathcal{S}$. Then, there exists an adaptive algorithm that achieves a \(\frac{1}{2}\)-approximation of~\atar.
\end{theorem}

In order to prove Theorem~\ref{thm:separable_revenues}, we show that under the same-order condition the optimal revenue function for each supplier satisfies the submodular order property. Leveraging the submodular order property and a primal-dual analysis, we design an adaptive greedy algorithm that processes customers sequentially and adaptively offers each customer an assortment of suppliers based on the choices of previous customers, maximizing the marginal increase in expected revenue and providing a \(\frac{1}{2}\)-approximation for~\atar. Unlike the general-case algorithm, which is static and requires randomized assortment sampling, the greedy algorithm is deterministic, adaptive, and practically more efficient. It also leverages a completely different analysis from the general case algorithm. Our analysis and proof of Theorem \ref{thm:separable_revenues} are provided in Section \ref{sec:1/2-greedy-approx}. {\color{blue}This $\tfrac12$ guarantee is tight for the greedy algorithm: in Appendix~\ref{appendix:greedy-tight} we construct a family of same-order instances on which the ratio between the expected revenue of the greedy algorithm and the optimum approaches $\tfrac12$. Hence no approximation factor larger than $\tfrac12$ can be guaranteed for this algorithm, even under the same-order condition.}

\vspace{2mm}
\noindent
{\color{blue}{\bf Hardness of~\atar.} We complement these positive results with a hardness result: \atar~is APX-hard (Lemma~\ref{lem:atar-apx-hard}), and therefore admits no polynomial-time approximation scheme unless $\mathsf{P}=\mathsf{NP}$. Some constant-factor loss is thus unavoidable, so constant-factor approximation algorithms such as ours are the appropriate goal for this problem.}

{\color{blue}
\begin{remark}\label{remark:variants}
In \atar, the decision maker chooses the order in which agents are processed. Our randomized algorithm from Theorem~\ref{thm:general_revenues} is static and its performance guarantee holds against an LP upper bound of \atar. Hence, it offers the same guarantee even when the processing order is fixed or agents on each side are processed simultaneously, as illustrated by the three settings in Figure~\ref{fig:arrival-order-comparision}.
We refer to these three settings as follows. \atar\ is the main problem studied in this paper and corresponds to the most adaptive variant, where the platform adaptively selects both the processing order of customers and the assortment offered to each customer based on previously observed choices. In \textsf{FTAR} (Fixed-order Two-Sided Assortment with Revenues), the processing order of customers is exogenously fixed, while the platform may adapt assortments based on previously observed choices. In \textsf{STAR} (Static Two-Sided Assortment with Revenues), all customer assortments are fixed in advance, before any choices are observed, and customers are processed simultaneously. In all three variants, customers are processed before suppliers. The processing order of suppliers does not affect the outcome because, once customer processing is complete, the supplier backlogs, the sets of customers who selected each supplier, are disjoint.

The values of \eqref{eq:relaxation_onesided_customers} and \eqref{eq:problem:primal} provide LP upper bounds that benchmark the performance guarantees of our adaptive greedy algorithm in Theorem~\ref{thm:separable_revenues} and randomized static algorithm in Theorem~\ref{thm:general_revenues}, respectively. Since Algorithm~\ref{alg:Randomized Static 1/2} fixes all customer assortments before observing any choices and its guarantee holds against \eqref{eq:problem:primal}, which upper-bounds all three variants, the same guarantee applies to (\textsf{STAR}), (\textsf{FTAR}), and \atar. More generally, it extends to batched variants in which customers are partitioned into sequentially processed groups and assortments within each group are fixed before observing choices from that group. In contrast, the greedy algorithm of Theorem~\ref{thm:separable_revenues} relies on sequential feedback and therefore applies specifically to \atar.
\end{remark}
}

\begin{figure}[ht]
\centering
\begin{tikzpicture}[>=stealth, every node/.style={font=\small}]
  \draw[->, >=latex, line width=1pt] (0,0) -- (15,0);
  \draw (11,-0.1) -- (11,0.1) node[above,align=center,text width=3.2cm]{$(\textsc{LP--I})$};
  \draw (13,-0.1) -- (13,0.1) node[above,align=center,text width=3.2cm]{$(\textsc{LP--II})$};
\node[align=center] at (14.4,-0.6) {\scriptsize{Optimal Expected Revenue}};
  \foreach \x/\w in {
    2/{{\sf(STAR)}},
    5.5/{{\sf(FTAR)}},
    9/{\atar}
  }{
    \draw (\x,-0.1) -- (\x,0.1) node[above,align=center,text width=3.2cm]{\w};
  }
\end{tikzpicture}
\caption{\textcolor{blue}{Comparison of optimal expected revenues and LP upper bounds across three variants of the two-sided assortment problem. The horizontal axis displays, from left to right in increasing order, the optimal expected revenues of the static (\textsf{STAR}), fixed-order (\textsf{FTAR}), and adaptive (\textsf{ATAR}) processing variants, together with the LP relaxation upper bounds \eqref{eq:relaxation_onesided_customers} and \eqref{eq:problem:primal}.}}
\label{fig:arrival-order-comparision}
\end{figure}

\noindent
{\bf Numerical Experiments.} We numerically evaluate our static and greedy algorithms against the LP upper bounds for~\atar. We show that the empirical performance is better than the worst-case theoretical guarantee. Moreover, our results show that the greedy algorithm consistently outperforms the static algorithm even under general revenues, despite lacking any theoretical guarantee in that case, highlighting the power of adaptivity. {\color{blue}However, this empirical strength does not extend to the worst case: in Appendix~\ref{appendix:hard-instance} we construct a family of general-revenue instances, processed under the same descending-average-revenue heuristic used in our experiments, on which the greedy algorithm obtains only a $\Theta(1/n)$ fraction of the optimum. Beyond the same-order case the greedy algorithm therefore admits no constant-factor approximation, whereas the static algorithm retains its $(\tfrac12-\epsilon)$-guarantee for arbitrary revenues. We also study how the two algorithms respond to the market structure. A sensitivity analysis shows that their relative performance is governed by congestion: the static algorithm improves most as agent preferences diversify and competition for the same matches eases. Finally, we test both algorithms on instances calibrated to real marketplace data from the freelancing platform Upwork. The greedy
algorithm achieves near-optimal revenue in both the pooled marketplace and individual job categories, and both algorithms exhibit stable performance across the five categories.} Our numerical study is presented in Section \ref{sec:numerics}.

\subsection{Related Literature}

The literature on two-sided assortment optimization has primarily focused on algorithms for maximizing the number of successful matches. Revenue maximization is a more recent and structurally distinct objective. A separate line of work explores domain-specific applications with practical design constraints. We review representative contributions across these directions.

\vspace{2mm}
\noindent
{\bf Two-Sided Assortment: Maximizing Number of Matches.}
 \citet{ashlagi2022assortment} appear to be the first to formally model the two-sided assortment optimization problem. They consider a static one-sided setting in which all customers are simultaneously presented assortments of suppliers, make their choices, and suppliers then observe which customers selected them.
A match occurs if both agents select each other, and the goal is to maximize the expected number of successful matches.  Assuming homogeneous MNL choice behavior, they prove the problem is strongly NP-hard and design the first constant-factor approximation algorithm using an LP relaxation and rounding approach. \citet{torrico2023multiagent} extend this static framework to allow for heterogeneous preferences and general choice models{\color{blue}, and to maximize the expected number of matches weighted by responding-side revenues $\{r_j\}_{j \in \mathcal{S}}$}. They improve upon prior results by proposing a $(1 - \tfrac{1}{e})$-approximation algorithm using the continuous greedy method, as well as a simpler greedy algorithm achieving a $\tfrac{1}{2}$-approximation{\color{blue}, both for the static one-sided setting}. In dynamic settings, \citet{aouad2023} study online assortment in labor markets with either adversarial or stochastic arrivals. In their model, customers arrive sequentially and supplier availability can vary over time. Under i.i.d. arrivals with full supplier availability, they provide a $(1 - \tfrac{1}{e})$-approximation guarantee.

\citet{housni2024twosided} unify several of these models under a broader taxonomy of four policy classes: (i) \emph{fully static}, where assortments for all agents are determined simultaneously; (ii) \emph{one-sided static}, where agents on one side are processed simultaneously and the other side responds; (iii) \emph{one-sided adaptive}, where one side is processed sequentially with adaptive assortments, followed by the other side; and (iv) \emph{fully adaptive}, where agents from both sides can be processed in any order with adaptive decisions. They analyze the performance gaps between these policies, referred to as \emph{adaptivity gaps}, showing that the gap between one-sided adaptive and one-sided static policies is exactly \((1 - \tfrac{1}{e})\), and the gap between fully adaptive and one-sided adaptive is \(\tfrac{1}{2}\). In contrast, fully static policies can perform arbitrarily poorly. For their one-sided adaptive setting, they propose a $\tfrac{1}{2}$-approximation greedy algorithm that selects customers adaptively and offers assortments maximizing the myopic marginal gain in expected matches at each step.

\vspace{2mm}
\noindent
{\bf Two-Sided Assortment: Maximizing Expected Revenue.}  
A recent, smaller body of work investigates expected revenue maximization with general pairwise revenues. In particular, \citet{ahmed2022parameterized} consider {fully static} policies and show that revenue maximization is NP-hard, providing parameterized complexity results but no constant-factor approximation guarantees. More recently, in a concurrent work to ours, \citet{danny2024} study static one-sided policies with pair-specific revenues, where customers are processed simultaneously,  under two (response) models: a {customized model}, where the platform can personalize the set of selected customers shown to each supplier, and an {inclusive model}, where each supplier must see all customers who selected them. For the customized model, they propose a randomized LP-based algorithm achieving a $\tfrac{1}{3}$-approximation. For the inclusive model, they provide a $\left(\tfrac{10}{539} - \epsilon\right)$-approximation via a case-based analysis.
Our work tackles the general adaptive setting, in contrast to the static policies studied by \citet{danny2024}. As a byproduct of our analysis, one of our algorithms for the adaptive problem is static, and therefore achieves a $(\tfrac{1}{2} - \epsilon)$-approximation guarantee for the customized model of \citet{danny2024} as well, thereby improving the best known guarantee for that setting. Our techniques are also fundamentally different in nature.

\vspace{2mm}
\noindent
{\bf Two-Sided Assortment: Domain-Specific Applications.}
Several studies examine two-sided assortment for specific applications. For example, \citet{rios2023} and \citet{rios2024} model dating platforms with fairness and visibility constraints, proposing personalized dynamic menus. \citet{aouad2023} also investigate labor markets, considering employer preferences and online worker presentation. \citet{shi2025optimal} introduces endogenous pricing into two-sided platforms, characterizing equilibrium-based outcomes. These works incorporate important real-world platform constraints but generally remain focused on match count rather than revenue maximization. In contrast, our work is motivated by revenue-centric goals.

{\em Outline.} The rest of this paper is organized as follows. In Section~\ref{sec:preliminaries}, we provide some preliminaries on agent choice models and revenue functions. In Section~\ref{sec:1/2-epsilon-approx}, we establish Theorem~\ref{thm:general_revenues}. 
In Section~\ref{sec:1/2-greedy-approx}, we establish Theorem~\ref{thm:separable_revenues}. 
 {\color{blue} Section~\ref{sec:numerics} presents numerical experiments, and Section~\ref{sec:conclusion} concludes with limitations and future directions.}
\section{Preliminaries and Notation}\label{sec:preliminaries}
In this section, we provide  preliminaries and notation related to {\color{blue}the model we study}. We begin with MNL choice models and then present an exact dynamic programming formulation of~\atar.

\subsection{The MNL Choice Models} 
We formalize customer and supplier preferences through {agent‑specific} MNL choice models.

\vspace{2mm}
\noindent
{\em Customer Choice Model.}
Each customer \( i \in \mathcal{C} \) has preferences over suppliers \( \mathcal{S} \cup \{0\} \), where \(\{0\}\) denotes the outside option, i.e., selecting no option from the offered assortment. The choice model is defined as a function \(\phi_i \colon \mathcal{S} \cup \{0\} \times 2^{\mathcal{S}} \to [0,1],\) where \(\phi_i(j, S)\) represents the probability that customer \( i \) selects supplier \( j \in S \cup \{0\} \) when presented with an assortment \( S \subseteq \mathcal{S} \). Under the MNL model, these probabilities are parameterized by preference weights \(\{u_{ij}\}_{j \in \mathcal{S} \cup \{0\}}\). Without loss of generality, we normalize the outside-option preference weight so that \(u_{i0} = 1\) for all $i\in\mathcal{C}$. For any set \( S \subseteq \mathcal{S} \), the selection probabilities for customer $i \in \mathcal{C}$ are:
\[
\phi_i(j, S) = 
\begin{cases} 
\frac{u_{ij}}{1 + \sum_{\ell \in S} u_{i\ell}}, & \text{if } j \in S, \\
\frac{1}{1 + \sum_{\ell \in S} u_{i\ell}}, & \text{if } j = 0, \\
0, & \text{otherwise}.
\end{cases}
\]

\vspace{2mm}
\noindent
{\em Supplier Choice Model.}
Symmetrically, each supplier \( j \in \mathcal{S} \) has preferences over customers \( \mathcal{C} \cup \{0\} \), modeled by \(\phi_j \colon \mathcal{C} \cup \{0\} \times 2^{\mathcal{C}} \to [0,1],\) with preference weights \(\{w_{ji}\}_{i \in \mathcal{C}\cup\{0\}}\) and \( w_{j0} = 1 \). For any set \( C \subseteq \mathcal{C} \), the selection probabilities for supplier $j \in \S$ are:
\[
\phi_j(i, C) = 
\begin{cases} 
\frac{w_{ji}}{1 + \sum_{k \in C} w_{jk}}, & \text{if } i \in C, \\
\frac{1}{1 + \sum_{k \in C} w_{jk}}, & \text{if } i = 0, \\
0, & \text{otherwise}.
\end{cases}
\]

The platform's expected revenue from offering any set of customers \( C \subseteq \mathcal{C} \), who initially chose supplier \( j \), to supplier \( j \) is given by
$$R_j(C)=\sum_{i\in C} r_{ij} \phi_j(i, C)=\frac{\sum_{i\in C} r_{ij}w_{ji}}{1+\sum_{i\in C} w_{ji}}.$$

For supplier $j\in \mathcal{S}$ and set of customers $C\subseteq \mathcal{C}$, we denote the \emph{optimal revenue function} by \( R_j^*(C) \) which  represents the maximum expected revenue attainable by any subset of $C$, i.e.,
$$R_j^*(C)=\underset{C'\subseteq C}{\max}\;  R_j(C').$$
Note that, in the special case where $r_{ij}=1$ for all $i\in\mathcal{C}$ and $j\in\mathcal{S}$, we have \(R_j(C)=R_j^*(C)
\ \text{for all}\,j\in\mathcal{S},\ C\subseteq\mathcal{C},\) which coincides with the demand function of supplier $j$, that is, the probability that supplier $j$ chooses any customer from the offered set  $C$.

\subsection{Dynamic Programming Formulation of~\atar}
\label{sec:DP_formulation}
We present a dynamic programming formulation of~\atar~that computes the two‐sided platform’s expected revenue under the optimal policy.

For all $t\in \{0, \dots, n-1\}$, let \( \mathcal{C}^{t} \subseteq \mathcal{C} \) denote the set of customers remaining to be processed at the beginning of step \( t+1\), where  \( n \) is the total number of customers. {\color{blue}Each customer is processed exactly once: at each of the \( n \) steps the platform selects a single customer from \( \mathcal{C}^{t} \), processes it, and removes it from the set of remaining customers.} For each supplier \( j \in \mathcal{S} \), define \( C_j^t \subseteq \mathcal{C} \) as the \emph{backlog} of supplier \( j \) at the beginning of step \( t+1 \), representing the set of customers who have selected supplier \( j \) thus far. {\color{blue}We use the term \emph{backlog} solely to denote this set of already-selected customers; it carries no connotation of delay or penalty.} Let \( \{C_j^t\}_{j \in \mathcal{S}} \) denote all supplier backlogs at time \( t \). Define the value-to-go function \( V_t(\mathcal{C}^t, \{C_j^t\}_{j \in \mathcal{S}}) \) as the maximum expected revenue achievable from step \( t+1 \) onward under adaptive policies. 
In particular, \( V_0(\mathcal{C}, \{\emptyset\}_{j \in \mathcal{S}}) \) represents the maximum expected revenue at the beginning.
The dynamic programming of~\atar~is described as follows.

For all  $t\in \{0, \dots, n-1\},\ \mathcal{C}^t\subseteq \mathcal{C},\ j\in\mathcal{S},\ C_j^t\subseteq\mathcal{C}$, the Bellman equation is given by
\begin{align}
V_t(\mathcal{C}^t, \{C_j^t\}_{j \in \mathcal{S}}) = \max_{i \in \mathcal{C}^t,\, S_i \subseteq \mathcal{S}} \left\{ \sum_{\ell \in S_i \cup \{0\}} \phi_i(\ell, S_i) \, V_{t+1}\left(\mathcal{C}^t \setminus \{i\}, \{C_j^{t+1}\}_{j \in \mathcal{S}} \right) \right\},\notag
\end{align}
where for all $j \in \mathcal{S}$, the updated backlogs \( \{C_j^{t+1}\}_{j \in \mathcal{S}} \) in this equation satisfy $C_j^{t+1} = C_j^t \cup \{i\}$ if $j = \ell$, and $C_j^{t+1} = C_j^t$ otherwise.
Also, the boundary condition for every disjoint final collection of backlogs $\{C_j^n\}_{j\in \mathcal{S}}$ is
\begin{equation}
\begin{aligned}
&\quad\quad V_n(\emptyset, \{C_j^n\}_{j \in \mathcal{S}}) = \sum_{j \in \mathcal{S}} R_j^*( C_j^n ).\notag
\end{aligned}
\end{equation}

Finally, let \( \mathsf{OPT} \) denote the optimal expected revenue for~\atar, which is the value-to-go function evaluated at the initial state. We have
\begin{align}
\mathsf{OPT} = V_0(\mathcal{C}, {\emptyset}_{j \in \mathcal{S}}). \label{eq:C-OA_DP}
\tag{\sf{ATAR}}
\end{align}

The Bellman equation in the dynamic program~\eqref{eq:C-OA_DP} corresponds to each adaptive step \( t \in \{0, \dots, n-1\} \), where the platform selects a customer \( i \in \mathcal{C}^t \) to process and offers an assortment \( S_i \subseteq \mathcal{S} \) of suppliers. The expression inside the maximization captures expected future revenue: if customer \( i \) chooses supplier \( \ell \in S_i \), the backlog \( C_\ell^t \) is updated; if the outside option \( 0 \) is chosen, no backlog updates. The boundary equation characterizes the final stage, where all customers have been processed adaptively, and the platform processes suppliers adaptively. The platform accumulates revenue based on supplier backlogs, which form disjoint subsets of customers since each customer can appear in at most one supplier's backlog. Thus, \( V_n(\emptyset, \{C_j^n\}_{j \in \mathcal{S}}) = \sum_{j \in \mathcal{S}} R_j^*(C_j^n) \). Solving~\eqref{eq:C-OA_DP} from \( \mathcal{C}^0 = \mathcal{C} \) and \( C_j^0 = \emptyset \) for all \( j \), yields the optimal expected revenue \( \mathsf{OPT} = V_0(\mathcal{C}, \{\emptyset\}_{j \in \mathcal{S}}) \).

While this dynamic program provides an exact formulation of~\atar, solving it is intractable due to the high-dimensional state space and exponential size decision tree. A natural approach to address this challenge is designing approximation algorithms. In our setting with general pairwise revenues, the optimal revenue function \( R_j^* \), for all \( j \in \mathcal{S} \), is monotone but not submodular (see Appendix~\ref{Appendix:proof:non-submodularity-g_j}), which presents new foundational challenges, as it limits the use of existing algorithmic frameworks that rely on the submodularity of objective functions.

{\color{blue}This intractability is not limited to solving~\atar\ exactly: the problem is also hard to approximate. In the following lemma, we show that \atar\ is APX-hard, so it admits no polynomial-time approximation scheme, and any polynomial-time approximation algorithm must incur a constant-factor loss.

\begin{lemma}[APX-hardness]\label{lem:atar-apx-hard}
Computing $\mathsf{OPT}$ for~\atar\ is APX-hard. Consequently, unless $\mathsf{P}=\mathsf{NP}$,~\atar\ admits no PTAS.
\end{lemma}
The proof uses a gap-preserving reduction from Maximum 3-Dimensional Matching restricted to instances in which each element appears in at most three triples (Max-3DM-3), and is deferred to Appendix~\ref{appendix:apx-hardness}.}

\section{\((\frac{1}{2}-\epsilon)\)-Approximation Algorithm for \atar}
\label{sec:1/2-epsilon-approx}
In this section, we present a randomized algorithm that gives a \((\frac{1}{2} - \epsilon)\)-approximation to \atar{} for any $\epsilon>0$. Furthermore, in the special case when the revenues associated with edges connected to the same supplier are all equal, our algorithm gives a \(\bigl(1 - \tfrac1e - \epsilon\bigr)\)-approximation to \atar, thereby establishing Theorem~\ref{thm:general_revenues}.

\vspace{2mm}
\noindent\textcolor{blue}{
\textbf{Road Map.}
To prove Theorem~\ref{thm:general_revenues}, we proceed as follows.
Our approach relies on an LP relaxation of~\atar, whose solution is used both to construct the algorithm and as the benchmark in its analysis.
In Section~\ref{sec:1/2-epsilon-approx:LP-relxation}, we show how to obtain this approximate LP solution. Since the LP relaxation has exponentially many variables, solving it exactly is intractable. We therefore turn to its dual and seek an efficient separation oracle. A key structural observation is that the separation oracle for this dual is closely connected to \emph{assortment optimization with fixed costs} under the MNL model, for which an FPTAS exists. Applying this FPTAS within the Ellipsoid method, we obtain a $(1-\delta)$-approximate solution to the LP relaxation in polynomial time for any $\delta > 0$.
Given this LP solution, we construct our randomized algorithm: for each customer, the algorithm samples an assortment from a probability distribution derived from the LP solution and offers it, without requiring any adaptive feedback.
The remaining task is to analyze the algorithm's performance. In Section~\ref{sec:1/2-epsilon-approx:correlation-gap}, 
we state and prove a bound on the \emph{correlation gap} of the optimal revenue function, which quantifies the loss from replacing a potentially correlated distribution over supplier backlogs by an independent distribution with the same marginal probabilities.
In Section~\ref{sec:1/2-epsilon-approx:algorithm}, we combine this bound with the LP solution to construct the algorithm and complete the proof of Theorem~\ref{thm:general_revenues}.}

In this work, we slightly abuse notation and use the same terms to denote both mathematical programs and their corresponding optimal values throughout the remainder of the paper.

\subsection{LP Relaxations for \atar}\label{sec:1/2-epsilon-approx:LP-relxation}
\textcolor{blue}{In this section, we develop two LP relaxations of~\atar\ and show that each provides a valid upper bound on $\mathsf{OPT}$. The first, \eqref{eq:relaxation_onesided_customers}, is valid under any discrete choice model; the second, \eqref{eq:problem:primal}, is a relaxation of \eqref{eq:relaxation_onesided_customers} under MNL that exploits the structure of MNL choice probabilities and admits an efficient approximate solution.} We begin by considering the following LP relaxation for~\atar.
\begin{align}
\begin{aligned}
\max \quad & \sum_{j\in \S}\sum_{C\subseteq \mathcal{C}}R_j^*(C)\cdot \lambda_{j,C}&\\
\text{s.t.}\quad & \sum_{C\subseteq \mathcal{C}}\lambda_{j,C} = 1,  &\forall j\in \mathcal{S},\notag\\
& \sum_{C\subseteq\mathcal{C}: C\ni i} \lambda_{j,C} = \sum_{S\subseteq \mathcal{S}: S\ni j}\tau_{i,S}\cdot\phi_i(j,S), &\forall i\in \mathcal{C}, \ j\in \mathcal{S},\notag\\
& \sum_{S\subseteq \mathcal{S}} \tau_{i,S} =1, &\forall i\in \mathcal{C},\notag \\
& \lambda_{j,C}, \ \tau_{i,S}\geq 0,  &\forall j\in \mathcal{S}, \ C\subseteq \mathcal{C}, \ i\in \mathcal{C}, \ S\subseteq \mathcal{S}. \notag
\end{aligned}\label{eq:relaxation_onesided_customers}\tag{\textsc{LP--I}}
\end{align}
This formulation is an adaptation of the LP relaxation introduced by \citet{housni2024twosided} for their \emph{one-sided adaptive} problem, where the goal is to maximize the expected number of matches in that setting. In their LP relaxation, for any supplier $j\in \mathcal{S}$, the objective function of the LP involves the demand function \( f_j \) defined for any set of customers \( C \subseteq \mathcal{C} \) as $f_j(C) = \sum_{i \in C} \phi_j(i, C)$, which represents the probability that supplier \( j \) selects any customer from the offered set \( C \subseteq \mathcal{C} \). In our setting, for all $j\in \mathcal{S}$, we replace the demand function \( f_j\) with the optimal revenue function \( R_j^*\), where for  $C\subseteq \mathcal{C}$, the function $R_j^*(C)$ represents the maximum expected revenue of offering a subset of customers in $C$ to supplier~$j$. Notably, solving this LP in polynomial-time is challenging due to its exponentially many variables. 

\textcolor{blue}{We interpret the variables and constraints of~\eqref{eq:relaxation_onesided_customers} as follows. For each $j \in \mathcal{S}$ and $C \subseteq \mathcal{C}$, $\lambda_{j,C}$ is the probability that the backlog of supplier $j$ (i.e., the set of customers who selected $j$ during the customer-processing phase) equals exactly $C$. For each $i \in \mathcal{C}$ and $S \subseteq \mathcal{S}$, $\tau_{i,S}$ is the probability that customer $i$ is offered assortment $S$. The first and third families of constraints ensure that $\{\lambda_{j,C}\}_{C \subseteq \mathcal{C}}$ and $\{\tau_{i,S}\}_{S \subseteq \mathcal{S}}$ form valid probability distributions for each $j \in \mathcal{S}$ and $i \in \mathcal{C}$, respectively. The second family links the two sides of the market: for each $(i,j) \in \mathcal{C} \times \mathcal{S}$, the probability that customer $i$ selects supplier $j$, computed from $\{\tau_{i,S}\}$, coincides with the probability that $i$ appears in the backlog of $j$, computed from $\{\lambda_{j,C}\}$.} The following lemma demonstrates that this formulation is indeed a relaxation of~\atar.

\begin{lemma}\label{lemma:LPrelaxation-1}
\eqref{eq:relaxation_onesided_customers} is a relaxation of~\atar.
\end{lemma}
The proof of Lemma \ref{lemma:LPrelaxation-1} is similar to that of Lemma~2 in \citet{housni2024twosided}; for completeness, we include it in Appendix~\ref{proof:sec:1/2-epsilon-approx:LP-relxation}. Note that~\eqref{eq:relaxation_onesided_customers} is a relaxation of \atar\ that is valid under a broad class of discrete choice models. Under the MNL model, we leverage its structural properties to further relax \eqref{eq:relaxation_onesided_customers} and develop a new novel LP relaxation of \atar\,  which plays a central role in our analysis.

\begin{align}
\begin{aligned}
\max  \quad & \sum_{j\in \mathcal{S}}\sum_{C\subseteq \mathcal{C}}R_j(C)\cdot \lambda_{j,C}&\\
\text{s.t.}\quad & \sum_{C\subseteq \mathcal{C}}\lambda_{j,C} = 1,  &\forall j\in \mathcal{S},\notag\\
& \sum_{C\subseteq\mathcal{C}: C\ni i} \lambda_{j,C} =  x_{ij}, &\forall i\in \mathcal{C}, \ j\in \mathcal{S},\notag\\
& \frac{x_{ij}}{u_{ij}} + \sum_{\ell \in \mathcal{S}} x_{i\ell} \leq 1, &\forall i\in \mathcal{C}, \ j\in \mathcal{S},\notag\\
& \lambda_{j,C}, \ x_{ij}\geq 0,  &\forall j\in \mathcal{S}, \ C\subseteq \mathcal{C}, \ i\in \mathcal{C}. \notag
\end{aligned}\label{eq:problem:primal}\tag{\textsc{LP--II}}
\end{align}

{\color{blue}In~\eqref{eq:problem:primal}, the variable $x_{ij}$ represents the probability that customer $i$ selects supplier $j$, equivalently the marginal probability that $i$ appears in the backlog of $j$, as enforced by the second family of constraints. 
Without loss of generality, we assume that $u_{ij}>0$ for all $(i,j)$. Indeed, if $u_{ij}=0$ for some pair $(i,j)$, then customer $i$ never selects supplier $j$, so we may fix $x_{ij}=0$ and omit the corresponding constraint from the third family.}
The LP relaxations~\eqref{eq:relaxation_onesided_customers} and~\eqref{eq:problem:primal} differ in two key aspects that will be useful in our analysis. One difference is that the objective function of \eqref{eq:problem:primal} uses the expected revenue function \( R_j \) \textcolor{blue}{(recall that $R_j(C) = \frac{\sum_{i\in C}r_{ij}w_{ji}}{1+\sum_{i\in C}w_{ji}}$ is the expected revenue from offering the set $C$ to supplier $j$, as defined in Section~\ref{sec:preliminaries})} instead of the optimal revenue function \( R_j^* \) for all \( j \in \mathcal{S} \). Another difference lies in the second and third families of constraints, which differ between the two LPs because we leverage structural properties of the MNL model to design the constraints in \eqref{eq:problem:primal}. Although the two formulations are not equivalent, we show that \eqref{eq:problem:primal} is a relaxation of \eqref{eq:relaxation_onesided_customers} {\color{blue}under the MNL choice model}, and therefore \eqref{eq:problem:primal} is also a relaxation of \atar. In particular, we have the following lemma.

\begin{lemma}
\label{lemma:equivalent-LPs}
\eqref{eq:problem:primal} is a relaxation of~\atar.

\end{lemma}

{\color{blue}The relaxation~\eqref{eq:problem:primal} also differs from the LP relaxation of \citet{danny2024} for the static one-sided setting they study. Theirs is parameterized by marginal probabilities alone. We give a formal comparison of the two in Appendix~\ref{appendix:nissim-comparison}.}

\subsubsection{Proof of Lemma~\ref{lemma:equivalent-LPs}~\label{appendix:proof:lemma:equivalent-LPs}}
{\color{blue}To prove Lemma~\ref{lemma:equivalent-LPs}, we first introduce a few auxiliary definitions, then state two claims that together establish the lemma. The proofs of the two claims are given at the end of this section.}

Denote by $(\textsc{LP--III})$ the linear program with the same feasible set as \eqref{eq:problem:primal}, but with the objective obtained by replacing each function \(R_j\) with \(R_j^*\) for all \(j \in \mathcal{S}\). In other words, the objective of $(\textsc{LP--III})$ is $\sum_{j \in \mathcal{S}} \sum_{C \subseteq \mathcal{C}} R_j^*(C) \cdot \lambda_{j,C}$.

\textcolor{blue}{To compare \eqref{eq:problem:primal} with $(\textsc{LP--III})$, we show that they attain the same optimal value. Since the two programs share the same feasible set and differ only in their objectives ($R_j^*$ versus $R_j$), the key is to show that there exists an optimal solution to $(\textsc{LP--III})$ on which $R_j^*$ and $R_j$ coincide. This motivates the following definition. For any supplier $j \in \mathcal{S}$, a set $C \subseteq \mathcal{C}$ is called an \emph{active set} for supplier $j$ if $R_j^*(C) = R_j(C)$, i.e., offering the full set $C$ to supplier $j$ is  revenue-optimal. Let $\mathcal{P}_j$ denote the collection of all active sets for supplier $j$:
\[
\mathcal{P}_j = \{ C \subseteq \mathcal{C} : R_j^*(C) = R_j(C) \}.
\]}
Next, take a supplier \( j \in \mathcal{S} \), and consider any subset \( C \subseteq \mathcal{C} \). Since \( R_j^*(C) = \max_{C' \subseteq C} R_j(C') \), there exists at least one subset \( C^* \subseteq C \) such that \( R_j^*(C) = R_j(C^*) \). Although this maximizer may not be unique, we fix a lexicographic order over all subsets of \( \mathcal{C} \), and define a function $\mathcal{M}_j : 2^{\mathcal{C}} \to \mathcal{P}_j$ that assigns to each set \( C \subseteq \mathcal{C} \) the unique subset \( \mathcal{M}_j(C) \in \arg\max_{C' \subseteq C} R_j(C') \) with the smallest lexicographic rank. By construction, \( \mathcal{M}_j(C) \subseteq C \) and \( R_j^*(C) = R_j(\mathcal{M}_j(C)) = R_j^*(\mathcal{M}_j(C)) \), so \( \mathcal{M}_j(C) \in \mathcal{P}_j \).
We also define the inverse mapping \( \mathcal{M}_j^{-1}: \mathcal{P}_j \to 2^{\mathcal{C}} \) for each $j\in\mathcal{S}$ by
\begin{align}
    \mathcal{M}_j^{-1}(C) = \{ C' \subseteq \mathcal{C} : \mathcal{M}_j(C') = C \}, \quad \text{for any } C \in \mathcal{P}_j.\notag
\end{align}
Since \( \mathcal{M}_j \) is a function, the collection \( \{ \mathcal{M}_j^{-1}(C) \}_{C \in \mathcal{P}_j} \) forms a partition of \( 2^{\mathcal{C}} \) for each \( j \in \mathcal{S} \). Moreover, for any \( C \in \mathcal{P}_j \) and any \( C' \in \mathcal{M}_j^{-1}(C) \), we have \( R_j^*(C') = R_j^*(C) \).

{\color{blue}We present the following two claims, whose proofs are deferred to the end of this section.}
\begin{claim}\label{claim:relaxation_OA_C_with_x,g}
\textup{$(\textsc{LP--III})$} is a relaxation of~\eqref{eq:relaxation_onesided_customers}.
\end{claim}
\begin{claim}
\label{claim:1-LP-optimal-solution}
There exists an optimal solution \( (\hat{\lambda}, \hat{x}) \) to~\textup{$(\textsc{LP--III})$} satisfying \(\hat{\lambda}_{j,C} = 0\) for all \( C \notin \mathcal{P}_j \); that is, \(\hat{\lambda}_{j,C} = 0\) whenever \(C\) is not an active set for supplier \(j \in \mathcal{S}\).
\end{claim}

{\color{blue}We now prove the lemma using these two claims.} By Claim~\ref{claim:relaxation_OA_C_with_x,g},~$(\textsc{LP--III})$ is a relaxation of~\eqref{eq:relaxation_onesided_customers}. To complete the proof of Lemma~\ref{lemma:equivalent-LPs}, it suffices to show that~\eqref{eq:problem:primal} and~$(\textsc{LP--III})$ are equivalent.
These two problems share the same feasible set; the only difference lies in their objective functions:~$(\textsc{LP--III})$ uses \(R_j^*(C)\), while~\eqref{eq:problem:primal} uses \(R_j(C)\). Since \(R_j^*(C) \geq R_j(C)\) for any \(C \subseteq \mathcal{C}\), it follows that the optimal value of~\eqref{eq:problem:primal} is less than or equal to that of~$(\textsc{LP--III})$.
Conversely, Claim~\ref{claim:1-LP-optimal-solution} guarantees the existence of an optimal solution to~$(\textsc{LP--III})$ in which \(\lambda_{jC} = 0\) for all non-active sets \(C\), and we know that for active sets \(C\), \(R_j^*(C) = R_j(C)\). Therefore, the optimal value of~$(\textsc{LP--III})$ is also less than or equal to that of~\eqref{eq:problem:primal}.
These two bounds imply that the two problems are equivalent. {\color{blue}Since $(\textsc{LP--III})$ is a relaxation of~\eqref{eq:relaxation_onesided_customers}, which is itself a relaxation of~\atar\ by Lemma~\ref{lemma:LPrelaxation-1}, and~\eqref{eq:problem:primal} is equivalent to $(\textsc{LP--III})$, it follows that~\eqref{eq:problem:primal} is a relaxation of~\atar. This completes the proof of the lemma.}
{\color{blue}It remains to prove the two claims.}
\begin{myproof}[Proof of Claim~\ref{claim:relaxation_OA_C_with_x,g}]
Let $(\lambda^*, \tau^*)$ be an optimal solution to \eqref{eq:relaxation_onesided_customers}. For all  $  (i,j)\in \mathcal{C}\times\mathcal{S}$ denote
    \begin{align*}  x_{ij}^*= \sum_{S\subseteq \mathcal{S}: S\ni j} \tau_{i,S}^*\cdot\phi_i(j,S),
\end{align*}
 where $\phi_i(j,S)=\frac{\mathds{1}_{\{j\in S\}}u_{ij}}{1+\sum_{\ell\in S} u_{i\ell}}$. First, we show that $(\lambda^*, x^*)$ constitutes a feasible solution to $(\textsc{LP--III})$ by verifying that it satisfies all constraints. The first two families of constraints are satisfied due to the feasibility of $\lambda^*$ with respect to~\eqref{eq:relaxation_onesided_customers}, and the construction of $x$. We now show that the third constraint is also satisfied: for all $i \in \mathcal{C}$ and $j \in \mathcal{S}$, we have $x_{ij}^* \geq 0$,
and 
\begin{align*}
\frac{ x_{ij}^*}{u_{ij}}+\sum_{\ell\in\mathcal{S}} x_{i\ell}^*=&\frac{ x_{ij}^*}{u_{ij}}+\sum_{\ell\in\mathcal{S}}\sum_{S\subseteq\mathcal{S}: S\ni \ell}\tau_{i,S}^*\cdot\phi_i(\ell, S)=\frac{x_{ij}^*}{u_{ij}}+\sum_{\ell\in\mathcal{S}}\sum_{S\subseteq\mathcal{S}}\tau_{i,S}^*\cdot\phi_i(\ell ,S)\\
    =&\frac{x_{ij}^*}{u_{ij}}+\sum_{S\subseteq\mathcal{S}}\sum_{\ell\in\mathcal{S}}\tau_{i,S}^*\cdot\phi_i(\ell,S)=\frac{ x_{ij}^*}{u_{ij}}+\sum_{S\subseteq\mathcal{S}}\tau_{i,S}^*\sum_{\ell\in\mathcal{S}}\phi_i(\ell,S)\\
    =&\frac{ x_{ij}^*}{u_{ij}}+\sum_{S\subseteq\mathcal{S}}\tau_{i,S}^*\left(1-\phi_i(0,S)\right)=\sum_{S\subseteq\mathcal{S}: S\ni j}\tau_{i,S}^*\cdot\frac{\phi_i(j,S)}{u_{ij}}+\sum_{S\subseteq\mathcal{S}}\tau_{i,S}^*\left(1-\phi_i(0,S)\right)\\
=&\sum_{S\subseteq\mathcal{S}}\tau_{i,S}^*\left(\mathds{1}_{\{j\in S\}}\phi_i(0,S)\right)+\sum_{S\subseteq\mathcal{S}}\tau_{i,S}^*\left(1-\phi_i(0,S)\right)\\
=&\sum_{S\subseteq\mathcal{S}}\tau_{i,S}^*-\sum_{S\subseteq\mathcal{S}}\tau_{i,S}^*\cdot\phi_i(0,S)\left(1-\mathds{1}_{\{j\in S\}}\right)\leq 1,
\end{align*}
where the last inequality holds since $\tau^*$ is a feasible solution to \eqref{eq:relaxation_onesided_customers}. Thus, any optimal solution $( \lambda^*, \tau^*)$ to \eqref{eq:relaxation_onesided_customers} forms a feasible solution $( \lambda^*,  x^*)$ to $(\textsc{LP--III})$ achieving the same objective value for both linear programs. Therefore, $(\textsc{LP--III})$ is a relaxation of~\eqref{eq:relaxation_onesided_customers}. 
\end{myproof}

\begin{myproof}[Proof of Claim~\ref{claim:1-LP-optimal-solution}]
Let \( (\lambda^*, x^*) \) be an optimal solution to~$(\textsc{LP--III})$, we consider the following solution to~$(\textsc{LP--III})$, such that for any $j\in\mathcal{S},\ i\in\mathcal{C}$, and $ C \subseteq \mathcal{C}$:
\begin{align*}
\hat{\lambda}_{j,C} =
\begin{cases}
\sum_{C' \in \mathcal{M}_j^{-1}(C)} \lambda^*_{j,C'} &\text{if } C \in \mathcal{P}_j, \\
0 & \text{otherwise},
\end{cases}  \qquad\quad  \text{and} \qquad\quad
\hat{x}_{ij} = \sum_{C \subseteq \mathcal{C}:C \ni i} \hat{\lambda}_{j,C}.
\end{align*}
We show that \((\hat{\lambda}, \hat{x})\) is feasible to~$(\textsc{LP--III})$ and achieves the optimal objective value of this LP. Let us start with evaluating the objective value of the LP attainable by  \((\hat{\lambda}, \hat{x})\):
\begin{align*}
\sum_{j \in \mathcal{S}} \sum_{C \subseteq \mathcal{C}} R_j^*(C) \, \hat{\lambda}_{j,C}
&= \sum_{j \in \mathcal{S}} \sum_{C \in \mathcal{P}_j} R_j^*(C) \, \hat{\lambda}_{j,C} \\
&= \sum_{j \in \mathcal{S}} \sum_{C \in \mathcal{P}_j} R_j^*(C) \, \sum_{C' \in \mathcal{M}_j^{-1}(C)} \lambda^*_{j,C'} \\
&= \sum_{j \in \mathcal{S}} \sum_{C \in \mathcal{P}_j} \sum_{C' \in \mathcal{M}_j^{-1}(C)} R_j^*(C') \, \lambda^*_{j,C'} = \sum_{j \in \mathcal{S}} \sum_{C' \subseteq \mathcal{C}} R_j^*(C') \, \lambda^*_{j,C'},
\end{align*}
where the third equality holds because 
for  \( C \in \mathcal{P}_j \) and  \( C' \in \mathcal{M}_j^{-1}(C) \), we have \( R_j^*(C') = R_j^*(C) \), and the fourth equality holds because the collection \( \{ \mathcal{M}_j^{-1}(C) \}_{C \in \mathcal{P}_j} \) forms a partition of \( 2^{\mathcal{C}} \) for each \( j \in \mathcal{S} \). 
Now we verify the feasibility of the solution. For the first family of constraints, we have  for any $j\in\mathcal{S}$:
\begin{align}
\sum_{C \subseteq \mathcal{C}} \hat{\lambda}_{j,C}= \sum_{C \in \mathcal{P}_j} \hat{\lambda}_{j,C} = \sum_{C \in \mathcal{P}_j} \sum_{C' \in \mathcal{M}_j^{-1}(C)} \lambda^*_{j,C'} = \sum_{C' \subseteq \mathcal{C}} \lambda^*_{j,C'} = 1.\notag
\end{align}
The second family of constraints is verified by construction of our solution. Now, let us verify the third family of constraints. We have for all $i\in\mathcal{C}$ and $j\in\mathcal{S}$:
\begin{align*}
\hat x_{ij}=\sum_{C \subseteq \mathcal{C}:C\ni i} \hat{\lambda}_{j,C}
= \sum_{C \in \mathcal{P}_j:C \ni i} \hat{\lambda}_{j,C}
&= \sum_{C \in \mathcal{P}_j} \sum_{C' \in \mathcal{M}_j^{-1}(C)} \mathds{1}_{\{i \in C\}} \lambda^*_{j,C'} \\
&\leq \sum_{C \in \mathcal{P}_j} \sum_{C' \in \mathcal{M}_j^{-1}(C)} \mathds{1}_{\{i \in C'\}} \lambda^*_{j,C'} = \sum_{C' \subseteq \mathcal{C}:C' \ni i} \lambda^*_{j,C'} = x^*_{ij},
\end{align*}
where the inequality holds since for all $C' \in \mathcal{M}_j^{-1}(C)$ we have $C'\supseteq C$, so $\mathds{1}_{\{i \in C\}}\leq \mathds{1}_{\{i \in C'\}}$. 
Therefore, we have established that \( 0 \leq \hat{x}_{ij} \leq x^*_{ij} \) for all \( i \in \mathcal{C} \), \( j \in \mathcal{S} \). This implies that for all $i\in \mathcal{C}, \ j\in \mathcal{S},$
$$ \frac{\hat{x}_{ij}}{u_{ij}} + \sum_{\ell \in \S} \hat{x}_{i\ell} \leq \frac{x^*_{ij}}{u_{ij}} + \sum_{\ell \in \S} x^*_{i\ell} \leq 1.$$
This concludes our proof of the claim.
\end{myproof}

\subsubsection{Solving~\textsc{(LP--II)}} \label{sec:algLP2}
Now we focus on how to solve~\eqref{eq:problem:primal}, although solving linear programs with an exponential number of variables is typically challenging. For example, \citet{housni2024twosided} did not directly solve their proposed LP relaxation, that is~\eqref{eq:relaxation_onesided_customers} with demand functions instead of our optimal revenue functions. Instead, they used their LP relaxation for an existence proof in the analysis of adaptivity gaps and a benchmark in the analysis of their algorithm. Our approach is different. Although we cannot solve~\eqref{eq:problem:primal} exactly, we will show how to obtain $(1-\delta)$-approximate solutions in polynomial-time for any $\delta>0$.
To begin solving~\eqref{eq:problem:primal} we present its dual program that is
\begin{align}
\begin{aligned}
    \min \quad & \sum_{j \in \mathcal{S}} \beta_j + \sum_{(i,j) \in \mathcal{C} \times \mathcal{S}} \alpha_{ij}&\notag\\
    \text{s.t.}\quad&  \beta_j \geq \underset{C\subseteq\mathcal{C}}{\max}\{R_j(C)-\sum_{i\in C}\gamma_{ij}\},  &\forall  j\in \mathcal{S},\notag\\
     &  \frac{\alpha_{ij}}{u_{ij}} + \sum_{\ell \in \mathcal{S}} \alpha_{i\ell} \geq \gamma_{ij}, &\forall i\in \mathcal{C}, \ j\in \mathcal{S},\notag\\
    &   \alpha_{ij} \geq 0,\ \beta_j, \gamma_{ij}\in \mathbb{R}, & \forall i\in \mathcal{C}, \ j\in \mathcal{S}.\notag
\end{aligned}\label{eq:problem:dual}\tag{\textsc{Dual--II}}
\end{align}
For each $j\in \mathcal{S},\ C\subseteq \mathcal{C}$, and $\boldsymbol{\gamma}\in \mathbb{R}^{n\times m}$, we define the following function
\begin{align*}
    \revmod_j(C, \boldsymbol{\gamma}) = R_j(C)-\sum_{i\in C}\gamma_{ij},
\end{align*}
and refer to this function as \emph{revenue with cost} to describe a setting in which, in expectation we receive revenue \(R_j(C)\), and for each customer \(i\) included in the set \(C\) offered to supplier \(j\), we incur a fixed cost \(\gamma_{ij}\).

For any $j\in \mathcal{S}$, given any $\boldsymbol{\gamma}\in \mathbb{R}^{n\times m}$, the following maximization problem denotes the \emph{sub-dual} problem as 
\begin{align*}\label{eq:subdual}\tag{$\subdual_j$}
   \subdual_j(\boldsymbol{\gamma})= \underset{C\subseteq\mathcal{C}}{\max} \ \revmod_j(C, \boldsymbol{\gamma}).
\end{align*}
Finally, for any $j\in \mathcal{S}$, and $\boldsymbol{\gamma}\in \mathbb{R}^{n\times m}$, the first family of constraints in \eqref{eq:problem:dual} is referred to as
\begin{align*}\label{constraint:dual-AC} \tag{$\textsc{Dual-AC}_j$}
\beta_j \geq \subdual_j(\boldsymbol{\gamma}).
\end{align*}
Here, $\textsc{Dual-AC}_j$ stands for \emph{assortment with cost} constraint in the dual program for supplier~$j$.

Later in this section, we show that we can approximately solve~\eqref{eq:problem:dual} by approximately solving the sub-dual problem using the Ellipsoid method. We now focus on the sub-dual problem, known in the literature as \emph{assortment optimization with fixed costs}.

\vspace{2mm}
\noindent
\textbf{FPTAS for ($\textsc{Sub-Dual}_j$)}. As explained earlier, for any \(j \in \mathcal{S},\ C \subseteq \mathcal{C}\), and \(\boldsymbol{\gamma} \in \mathbb{R}^{n \times m}\), Problem~\eqref{eq:subdual} corresponds to selecting a subset of customers to maximize the expected revenue \(R_j(C)\) minus the total fixed cost \(\sum_{i \in C} \gamma_{ij}\). This is an instance of assortment optimization with fixed costs under the MNL model.
The literature on this problem is fairly developed. The initial work of~\citet{kunnumkal2010} considers the case of positive costs, proposing both a \(\tfrac{1}{2}\)-approximation algorithm and a PTAS. \citet{lu2023fairassortments}, building on the framework of~\citet{desir2022capacitated}, develop an FPTAS for the case of negative fixed costs. However, their approach does not extend to arbitrary cost settings.
In our case, the presence of the free variable \(\boldsymbol{\gamma}\) in~\eqref{eq:problem:dual} necessitates handling arbitrary (positive or negative) fixed costs. The only applicable algorithm is that of~\citet{chen2025fairassortment}, who recently proposed an FPTAS for assortment optimization with arbitrary costs and an arbitrary cardinality constraint. In particular, it also applies to the unconstrained setting by setting the cardinality bound equal to {\color{blue}the size of the ground set, that is, the number of customers \( n \)}. Therefore, we use their FPTAS to get a \((1 - \delta)\)-approximate solution to~\eqref{eq:subdual} in polynomial time for any \(\delta > 0\).

\vspace{2mm}
\noindent
{\bf Near-Optimal Algorithm for Solving~\textsc{(LP--II)}.} We continue this section by presenting an FPTAS for solving~\eqref{eq:problem:primal}. To design an algorithm for approximately solving this LP, we adapt the classical Ellipsoid method. This adaptation, similar to the approaches of \citet{chen2025fairassortment} and \citet{jansen2003approximate}, leverages approximate oracles to efficiently solve LPs with exponentially many variables, as is the case in our setting. To this end, the following lemma shows that given any $\delta >0$, if we can find a $(1-\delta)$-approximate solution to Problem~\eqref{eq:subdual} for any $j\in\mathcal{S}$, then there exists an Ellipsoid-based algorithm that  finds a $(1-\delta)$-approximate solution to~\eqref{eq:problem:primal}, and runs in polynomial-time. We defer the proof of Lemma~\ref{lemma:approx} to Appendix \ref{appendix:proof:lemma:approx}.

\begin{lemma}[Approximation Algorithm for Solving~\eqref{eq:problem:primal}]
\label{lemma:approx} 
Fix any \(\delta > 0\). For any \(j \in \mathcal{S}\), suppose there exists a polynomial-time \((1 - \delta)\)-approximation algorithm for~\eqref{eq:subdual}. Then, there exists a polynomial-time Ellipsoid-based algorithm that computes a \((1 - \delta)\)-approximate solution to~\eqref{eq:problem:primal}.
\end{lemma}

\subsection{Correlation Gap for Optimal Expected Revenue Functions under MNL}\label{sec:1/2-epsilon-approx:correlation-gap}

\textcolor{blue}{In this section, we derive a bound on the correlation gap of the optimal revenue function $R_j^*$ for each supplier $j \in \mathcal{S}$. This analysis is self-contained and independent of the LP developed in Section~\ref{sec:1/2-epsilon-approx:LP-relxation}; it will be used in Section~\ref{sec:1/2-epsilon-approx:algorithm} to bound the loss incurred when replacing the correlated distribution over supplier backlogs given by the LP solution with the independent distribution induced by our randomized algorithm. The resulting correlation-gap bound is also of independent interest.}

We begin by defining the notion of correlation gap. Let \(\mathcal{N}\) be a finite ground set, and let \( f : 2^{\mathcal{N}} \to \mathbb{R}_+ \) be a nonnegative set function. Given a probability distribution \(D\) over subsets of \(\mathcal{N}\), let \(D^{\mathrm{ind}}\) denote the independent distribution with the same marginals as \(D\). In particular, \(D^{\mathrm{ind}}\) is a distribution in which each element in \(\mathcal{N}\) is selected independently of the others, and the probability of each element being selected equals its marginal probability of appearing in a set drawn from the distribution \(D\). The \emph{correlation gap} of \(f\) is defined as the worst-case ratio 
$$\sup_{D} \frac{\mathbb{E}_{A \sim D}[f(A)]}{\mathbb{E}_{A \sim D^{\mathrm{ind}}}[f(A)]}.$$

In our setting for~\atar, correlation gap will measure the maximum possible loss when replacing an adaptive (correlated) assortment policy with a static (independent) assortment policy, and plays a central role in bounding the performance gap between adaptive and static assortment policies.

We are interested in analyzing the correlation gap of the optimal revenue function \( R_j^* \) for each supplier \( j \in \mathcal{S} \), which appears in the objective of~\atar. To this end, we simplify notation by dropping the supplier index and focus on the expected revenue function of the optimal assortment in the following setting. In particular, let $\mathcal{N}$ be a ground set of items with nonnegative revenues $\{r_i\}_{i \in \mathcal{N}}$\textcolor{blue}{, where \(r_i \in [r_{\min}, r_{\max}]\) for all \(i \in \mathcal{N}\). An item with zero revenue is never selected, so replacing it by one with zero preference weight and an arbitrary revenue leaves \(R^*\) unchanged; we may therefore take all revenues to be positive, with \(r_{\min}>0\)}. Under the MNL choice model, let $\phi:2^{\mathcal{N}}\to[0,1]$ denote the choice probabilities, and define the optimal revenue function $R^*(A)=\underset{A'\subseteq A}{\max}\ \sum_{i\in A'} r_i\,\phi(i,A')$ for all $ A\subseteq\mathcal{N}$.  
In Lemma~\ref{lemma:MNL 1/2 coerrelation gap}, we derive the correlation gap of function $R^*$.

\begin{lemma}[Correlation Gap for Optimal Revenue Function under MNL] \label{lemma:MNL 1/2 coerrelation gap}
For any distribution \(D\) over subsets of \(\mathcal{N}\), let \(D^{\mathrm{ind}}\) be an independent distribution having the same marginals as \(D\).  Then,
\[\frac{\mathbb{E}_{A\sim D}\bigl[R^*(A)\bigr]}{\mathbb{E}_{A\sim D^{\mathrm{ind}}}\bigl[R^*(A)\bigr]}\leq {\sf CorrGap},
\]
where {\color{blue}${\sf CorrGap}:=\min\bigl\{2,\ \tfrac{r_{\max}}{r_{\min}}\bigl(\tfrac{e}{e-1}\bigr)\bigr\}$. When all revenues are equal, this bound reduces to $\tfrac{e}{e-1}$.}

\end{lemma}

{\color{blue}For general revenues, the correlation gap of \(R^*\) attains the value \(2\) in the limit, while for equal revenues it attains \(\tfrac{e}{e-1}\); both tightness results are discussed in Appendix~\ref{appendix:lp-tightness}.}

{\color{blue}We will prove Lemma \ref{lemma:MNL 1/2 coerrelation gap} after establishing the necessary preliminaries on cost-sharing schemes.}

\subsubsection{Preliminaries}
\textcolor{blue}{We introduce the notion of a ``\emph{cost-sharing scheme}'' from game theory, which plays a central role in establishing Lemma~\ref{lemma:MNL 1/2 coerrelation gap}. Intuitively, a cost-sharing scheme distributes the total value of a set function among its individual elements. If we can find a way to share this value such that adding a new element never increases the cost share of existing elements (\emph{cross-monotonicity}) and the sum of shares equals the total value (\emph{budget-balance}), we can guarantee a strong correlation-gap bound.}
More formally, let $\mathcal{N} = \{1, \ldots, n\}$ denote a finite ground set. Given a total cost function $f : 2^{\mathcal{N}} \to \mathbb{R}_+$, a cost-sharing scheme is defined as a collection of cost allocation functions $\{\chi_i\}_{i\in\mathcal{N}}$, where each $\chi_i : 2^{\mathcal{N}} \rightarrow \mathbb{R}_+$ specifies the share of cost assigned to element $i$ in any subset $A \subseteq \mathcal{N}$, denoted by $\chi_i(A)$. For a given $\beta>0$,
a cost-sharing scheme $\{\chi_i\}_{i\in\mathcal{N}}$ of a function $f$ is \emph{$\beta$-budget-balanced} if:
 $$\frac{f(A)}{\beta}\leq \sum_{i\in \mathcal{N}} \chi_i(A)\leq f(A),\qquad  \text{for all }~~A\subseteq \mathcal{N},$$
 and it is \emph{cross-monotonic}, if for all $i\in \mathcal{N}$:
 $$\chi_i(A\cup\{i\})\geq \chi_i(B\cup\{i\}),\qquad \text{for all}~~A\subseteq B\subseteq \mathcal{N}.$$
 A cost-sharing scheme is a \emph{$\beta$-cost-sharing scheme} if it is cross-monotonic and $\beta$-budget-balanced.

\citet{agrawal2012price} show that if a set function \( f \) is monotone and submodular, then its correlation gap is bounded by \( \frac{e}{e - 1} \). More generally, if \( f \) is monotone (non-decreasing) but not necessarily submodular, and it admits a \(\beta\)-cost sharing scheme, then its correlation gap is bounded by \( 2\beta \). We summarize their result in the following lemma.

\begin{lemma}[Theorem 4 in \citet{agrawal2012price}]\label{lemma:Agrawal-correlation-gap}
Consider a nonnegative set function \( f \). The correlation gap of \( f \) is upper bounded by:
\begin{itemize}
    \item[i)] \(\dfrac{e}{e-1}\), if \( f \) is monotone and submodular.
    \item[ii)] \(2\beta\), if \( f \) is monotone (non-decreasing) and admits a \(\beta\)-cost sharing scheme.
\end{itemize}
\end{lemma}

Let \( \pi \) be a total order over the ground set \( \mathcal{N} \). A set function \( f: 2^\mathcal{N} \to \mathbb{R}_+ \) is said to be \emph{submodular order} with respect to \( \pi \) if, for all sets \( A \subseteq B \subseteq \mathcal{N} \) and for any set \( C \subseteq \mathcal{N} \) that \emph{succeeds} \( B \) under \( \pi \), the marginal contributions satisfy $f(C \mid B) \leq f(C \mid A)$, where \( f(C \mid A) := f(C \cup A) - f(A) \).
Here, a set \( C \) is said to \emph{succeed} \( B \) under \( \pi \) if all elements of \( C  \) come after the elements in \( B \) in the total order \( \pi \). This condition generalizes the diminishing returns property of submodular functions.
As shown by \citet{udwani2024submodular}, our optimal expected revenue function \( R^* \) satisfies the submodular order property. In particular, we have the following lemma.

\begin{lemma}[Lemma 9 in~\citet{udwani2024submodular}]\label{lemma:udwani-MNL-SO}
Given a ground set of items \( \mathcal{N} \) with corresponding nonnegative revenues \( \{r_i\}_{i \in \mathcal{N}} \) and the MNL choice model \( \phi \), the optimal revenue function, defined as
$
R^*(A) = \max_{A' \subseteq A} \sum_{i \in A'} r_i \,\phi(i, A'),
$
for all \( A \subseteq \mathcal{N} \), is monotone sub-additive and submodular order with respect to the descending order of item revenues (with ties broken arbitrarily).
\end{lemma}

This structural property enables us to construct a valid cost-sharing scheme for the function \( R^* \), and thereby derive a correlation-gap bound using Lemma \ref{lemma:Agrawal-correlation-gap}. This is used in the next section to establish Lemma~\ref{lemma:MNL 1/2 coerrelation gap}.

\subsubsection{Proof of Lemma~\ref{lemma:MNL 1/2 coerrelation gap}}

{\color{blue}Let $f(A):=\sum_{i\in A}\phi(i,A)$ denote the expected demand of the set $A$, i.e., the probability that the agent selects some element of $A$. Under the MNL model, and more generally under most random utility models, $f$ is monotone and submodular, so by part~(i) of Lemma~\ref{lemma:Agrawal-correlation-gap} it has correlation gap $\tfrac{e}{e-1}$. Since $r_i\in[r_{\min},r_{\max}]$ for all $i$ and $R^*(A)$ maximizes $\sum_{i\in A'}r_i\,\phi(i,A')$ over subsets $A'\subseteq A$, the monotonicity of $f$ yields $r_{\min}\,f(A)\le R(A)\le R^*(A)\le r_{\max}\,f(A)$ for every $A\subseteq\mathcal{N}$. Hence, for any distribution $D$ and its same-marginal independent counterpart $D^{\mathrm{ind}}$,
\[
\frac{\mathbb{E}_{D}[R^*]}{\mathbb{E}_{D^{\mathrm{ind}}}[R^*]}\ \le\ \frac{r_{\max}\,\mathbb{E}_{D}[f]}{r_{\min}\,\mathbb{E}_{D^{\mathrm{ind}}}[f]}\ =\ \frac{r_{\max}}{r_{\min}}\cdot\frac{\mathbb{E}_{D}[f]}{\mathbb{E}_{D^{\mathrm{ind}}}[f]}\ \le\ \frac{r_{\max}}{r_{\min}}\cdot\frac{e}{e-1},
\]
and the correlation gap of $R^*$ is therefore at most $\tfrac{r_{\max}}{r_{\min}}(\tfrac{e}{e-1})$, which reduces to $\tfrac{e}{e-1}$ when revenues are equal.}

{\color{blue}To bound the correlation gap by $2$, we proceed as follows.}  Without loss of generality, assume the revenues of items are ordered such that \( r_1 \geq  \dots \geq r_n \). This ordering is used to define the sets as follows: for any set \( A \subseteq \mathcal{N} \) and \( i \in \mathcal{N} \), let \( A_{[1:i)} \) be the subset of \( A \) containing the items of $A$ that rank above item $i$ in this order, i.e., $A_{[1:i)}=\{1,\ldots,i-1\} \cap A$. Let \( A_{[1:i]} \) be the subset of \( A \) containing the items of $A$ that rank at or above item $i$ in this order, i.e., $A_{[1:i]}=\{1,\ldots,i\} \cap A$.
We define the following cost-sharing scheme for all \( A \subseteq \mathcal{N} \) and \( i \in \mathcal{N} \):
\[
\chi_i(A) = \mathds{1}_{\{i \in A\}} \left( \; R^*(A_{[1:i]}) - R^*(A_{[1:i)}) \; \right).
\]
To verify the cross-monotonicity property, for any \( A \subseteq \mathcal{N} \) and \( i \in \mathcal{N} \), we have
\begin{align}
    \chi_i(A \cup \{i\}) = R^*((A \cup \{i\})_{[1:i]}) - R^*((A \cup \{i\})_{[1:i)}) = R^*(\{i\} \mid A_{[1:i)}).\notag
\end{align}
Thus, for \( A \subseteq B \subseteq \mathcal{N} \), we obtain:
\begin{align}
    \chi_i(B \cup \{i\}) = R^*(\{i\} \mid B_{[1:i)}) \leq R^*(\{i\} \mid A_{[1:i)}) = \chi_i(A \cup \{i\}),\notag
\end{align}
where the inequality follows since $R^*$ is submodular order, \( A_{[1:i)} \subseteq B_{[1:i)} \) and \( i \) has a revenue smaller than all elements in \( B_{[1:i)} \), i.e., $i$ succeeds  \( B_{[1:i)} \) in the descending order  with respect to revenues.

To check budget balance, we have the following telescopic sum over all \( i \in \mathcal{N} \):
\begin{align*}
\sum_{i \in \mathcal{N}} \chi_i(A)
&{\color{blue}= \sum_{i \in \mathcal{N}} \mathds{1}_{\{i \in A\}} \left( R^*(A_{[1:i]}) - R^*(A_{[1:i)}) \right)}\\
&= \sum_{i \in A} \left( R^*(A_{[1:i]}) - R^*(A_{[1:i)}) \right) = R^*(A) - R^*(\emptyset) = R^*(A),
\end{align*}
{\color{blue}where the second equality restricts the sum to $A$, and the third follows from the telescopic sum.}
Thus, the cost-sharing scheme is \( \beta \)-budget-balanced with \( \beta = 1 \) for the function \( R^* \), completing the proof. Note as a byproduct, this proof establishes that the correlation gap of any non-negative, monotone (non-decreasing), and submodular order function is bounded by 2.

\subsection{Our Approximation Algorithm for~\atar}\label{sec:1/2-epsilon-approx:algorithm}

In this section, we design our approximation algorithm for~\atar\ and prove Theorem~\ref{thm:general_revenues}, which establishes its theoretical performance guarantee. \textcolor{blue}{The algorithm is constructed entirely from the LP solution of Section~\ref{sec:1/2-epsilon-approx:LP-relxation}; the correlation-gap bound of Section~\ref{sec:1/2-epsilon-approx:correlation-gap} is used in the proof of Theorem~\ref{thm:general_revenues} but does not affect the algorithm's construction.}

Our approach begins with a near-optimal solution to the LP relaxation of~\atar, given in~\eqref{eq:problem:primal}. As shown in Lemma~\ref{lemma:approx}, we can efficiently compute a \((1 - \delta)\)-approximate solution to~\eqref{eq:problem:primal} for any $\delta > 0$. This solution yields the marginal probabilities $x_{ij}$ that customer $i$ selects supplier $j$.
\textcolor{blue}{For each customer $i$, we rank the suppliers with $u_{ij}>0$ in decreasing order of the ratio $x_{ij}/u_{ij}$ and assign positive probability only to assortments that form prefixes of this ranking. As established in Lemma~\ref{lemma:MNL-polyhedron-distribution} below, this construction yields a valid probability distribution over assortments that recovers the marginal selection probabilities $x_{ij}$.}

This construction is standard and relies on a well-known property of the MNL choice model: any point inside the MNL polyhedron can be represented as a convex combination of assortments (see, e.g., \citet{Topaloglu2013}). The idea is formalized as follows. Let $\mathcal{S}=\{1,\dots,m\}$. Given a feasible solution \((\lambda, x)\) to~\eqref{eq:problem:primal}, for each customer \(i \in \mathcal{C}\), define a permutation \(\sigma_i: \mathcal{S} \to \mathcal{S}\) that orders the suppliers such that
\[
\frac{x_{i\sigma_i(1)}}{u_{i\sigma_i(1)}} \;\ge\;
\frac{x_{i\sigma_i(2)}}{u_{i\sigma_i(2)}} \;\ge\; \dots \;\ge\;
\frac{x_{i\sigma_i(m)}}{u_{i\sigma_i(m)}}.
\]
For any \(\ell \in \{1, \dots, m\}\), let
\(S_i(\ell)=\{\sigma_i(1), \sigma_i(2), \dots, \sigma_i(\ell)\}\) denote the top-\(\ell\) suppliers in this order, and define \(S_i(0)=\emptyset\).
For each $i\in\mathcal{C}$, we construct distribution \(\mathbb{P}_i:2^\mathcal{S} \to [0,1]\) as follows:
\begin{itemize}
  \item[i)]   $\mathbb{P}_i[S] = 0$, for any $S\in 2^\mathcal{S}\setminus\{S_i(0), S_i(1), \dots, S_i(m)\}$,
  \item[ii)]  \(\mathbb{P}_i[S_i(0)] = 1 - \left(\dfrac{x_{i\sigma_i(1)}}{u_{i\sigma_i(1)}} + \sum_{\ell=1}^{m} x_{i\sigma_i(\ell)} \right),\)
  \item[iii)] 
  $
  \mathbb{P}_i[S_i(\ell)] =
        \left(\frac{x_{i\sigma_i(\ell)}}{u_{i\sigma_i(\ell)}} -
              \frac{x_{i\sigma_i(\ell+1)}}{u_{i\sigma_i(\ell+1)}}\right)
        \left(1 + \sum_{j=1}^{\ell} u_{i\sigma_i(j)}\right),
  $
  for \(\ell \in \{1, \ldots, m-1\},\)
  \item[iv)] \(\mathbb{P}_i[S_i(m)] =
        \left(\dfrac{x_{i\sigma_i(m)}}{u_{i\sigma_i(m)}}\right)
        \left(1 + \sum_{j=1}^{m} u_{i\sigma_i(j)}\right).\)
\end{itemize}
\begin{lemma}\label{lemma:MNL-polyhedron-distribution}
For each $i \in \mathcal{C}$, $\mathbb{P}_i$  forms a probability distribution over $2^\mathcal{S}$. Moreover, sampling a set of suppliers from \(\mathbb{P}_i\), and offering it to customer \(i\) guarantees that for every \(j \in \mathcal{S}\) the marginal probability of customer $i$ choosing supplier \(j\) equals \(x_{ij}\).
\end{lemma}

\begin{myproof}
We begin by verifying that the given distribution is a valid probability measure. First, for any \( i \in \mathcal{C} \), since \((\lambda, x)\) is feasible to~\eqref{eq:problem:primal}, we have \( \mathbb{P}_i[S_i(0)] \geq 0 \). Moreover, by construction, \( \mathbb{P}_i[S_i(\ell)] \geq 0 \) for all \( \ell \in \{1, \dots, m\} \). Thus, for any $S\subseteq\mathcal{S}$ we have $\mathbb{P}_i[S]\geq 0$.
In order to show that the total probability sums to 1, we verify the following:
\begin{align*}
   \sum_{\ell=0}^m \mathbb{P}_i[S_i(\ell)] &= 
   \mathbb{P}_i[S_i(0)] + \sum_{\ell=1}^{m-1} \left( \frac{x_{i\sigma_i(\ell)}}{u_{i\sigma_i(\ell)}} - \frac{x_{i\sigma_i(\ell+1)}}{u_{i\sigma_i(\ell+1)}} \right) \left(1 + \sum_{j=1}^\ell u_{i\sigma_i(j)} \right)  +  \frac{x_{i\sigma_i(m)}}{u_{i\sigma_i(m)}}  \left( 1 + \sum_{j=1}^m u_{i\sigma_i(j)} \right) \\
   &= \mathbb{P}_i[S_i(0)] + \left(\frac{x_{i\sigma_i(1)}}{u_{i\sigma_i(1)}}\right) \left(1 + u_{i\sigma_i(1)}\right) + \sum_{\ell=2}^{m} \left(\frac{x_{i\sigma_i(\ell)}}{u_{i\sigma_i(\ell)}}\right) u_{i\sigma_i(\ell)} \\
   &= \mathbb{P}_i[S_i(0)] + \frac{x_{i\sigma_i(1)}}{u_{i\sigma_i(1)}} + \sum_{\ell=1}^{m} x_{i\sigma_i(\ell)} = 1.
\end{align*}
Next, we show that under this construction, the marginal probability that customer \(i\) selects supplier \(\sigma_i(j)\) is exactly \(x_{i\sigma_i(j)}\) for all $j\in \{1,\dots, m\}$:
\begin{align*}
   \mathbb{P}[i \text{ chooses } \sigma_i(j)] 
   &= \sum_{\ell=1}^{m} \mathbb{P}_i[S_i(\ell)] \cdot \phi_i(\sigma_i(j), S_i(\ell))
   = \sum_{\ell=j}^{m} \mathbb{P}_i[S_i(\ell)] \cdot \phi_i(\sigma_i(j), S_i(\ell)) \\
   &= \sum_{\ell=j}^{m-1} \left(\frac{x_{i\sigma_i(\ell)}}{u_{i\sigma_i(\ell)}} - \frac{x_{i\sigma_i(\ell+1)}}{u_{i\sigma_i(\ell+1)}} \right) \left( 1 + \sum_{j'=1}^\ell u_{i\sigma_i(j')} \right) \left( \frac{u_{i\sigma_i(j)}}{1 + \sum_{j'=1}^\ell u_{i\sigma_i(j')}}\right) \\
   &\quad + \left( \frac{x_{i\sigma_i(m)}}{u_{i\sigma_i(m)}} \right) \left(1 + \sum_{j'=1}^m u_{i\sigma_i(j')} \right) \left( \frac{u_{i\sigma_i(j)}}{1 + \sum_{j'=1}^m u_{i\sigma_i(j')}}\right) \\
   &= \sum_{\ell=j}^{m-1} \left( \frac{x_{i\sigma_i(\ell)}}{u_{i\sigma_i(\ell)}} - \frac{x_{i\sigma_i(\ell+1)}}{u_{i\sigma_i(\ell+1)}} \right) u_{i\sigma_i(j)} + \left(\frac{x_{i\sigma_i(m)}}{u_{i\sigma_i(m)}} \right) u_{i\sigma_i(j)} \\
   &= x_{i\sigma_i(j)},
\end{align*}
where the second equality holds since $\phi_i(\sigma_i(j), S_i(\ell))=0$ as $\sigma_i(j)\not \in S_i(\ell)$ for all $\ell\in \{1,\dots,j-1\}$, and the third equality follows by $\phi_i(\sigma_i(j), S_i(\ell)) =( \frac{u_{i\sigma_i(j)}}{1 + \sum_{j'=1}^\ell u_{i\sigma_i(j')}})$ for all $\ell\in \{j,\dots, m\}$.
Therefore, the given distribution recovers the desired marginal choice probabilities under the MNL choice model.
\end{myproof}

We now present Algorithm~\ref{alg:Randomized Static 1/2}, a simple randomized static policy that, for any \(\delta > 0\), takes a \((1 - \delta)\)-approximate solution \((\lambda,x)\) to~\eqref{eq:problem:primal}, and samples an assortment for each customer \(i\) from the distribution \(\mathbb{P}_i\) constructed for Lemma~\ref{lemma:MNL-polyhedron-distribution}. We conclude this section by showing that this algorithm achieves  a \(\frac{1}{2}\times(1 - \delta)\)-approximation for general revenues and a \((1 - \tfrac{1}{e})\times(1 - \delta)\)-approximation when revenues are uniform per supplier, thereby proving Theorem~\ref{thm:general_revenues}.

\begin{algorithm}[h]
 \caption{Randomized Static Assortment for~\atar}\label{alg:Randomized Static 1/2}
 \begin{algorithmic}[1]
\State Compute a $(1-\delta)$-approximate solution $(\lambda,x)$ to~\eqref{eq:problem:primal} as in Section \ref{sec:algLP2}.
\State Use the solution $x$ to construct the probability distribution  $\mathbb{P}_i$ over subsets of suppliers for all $i \in \C$ as in Lemma~\ref{lemma:MNL-polyhedron-distribution}.
    \State For each $i\in \mathcal{C}$, sample an assortment \(S_i^*\subseteq\mathcal{S}\)
           according to distribution $\mathbb{P}_i$,  and offer \(S_i^*\) to customer \(i\).
 \State Observe the choices made by all customers. For all $j\in \mathcal{S}$, consider the final backlog $C_j$ of supplier $j$, i.e., the set of all customers who chose supplier $j$. Offer supplier $j$ an optimal assortment $C_j^*\subseteq C_j$ that maximizes its expected revenue, i.e., $C_j^*\in\underset{C\subseteq C_j} {\arg\max} \ R_j(C) $.
 \end{algorithmic}
 \end{algorithm}

{\color{blue}\noindent\textbf{Running time.} Algorithm~\ref{alg:Randomized Static 1/2} runs in polynomial time. Step~1 computes a $(1{-}\delta)$-approximate solution to~\eqref{eq:problem:primal} in polynomial time by Lemma~\ref{lemma:approx}. Step~2 constructs the distribution $\mathbb{P}_i$ for each customer in polynomial time by Lemma~\ref{lemma:MNL-polyhedron-distribution}. In Step~4, for each supplier $j\in\mathcal{S}$, we solve an assortment optimization problem under MNL; the optimal assortment is revenue-ordered \citep{talluri2004revenue}, so it can be found in polytime by sorting customers in decreasing order of revenue and evaluating each prefix. Thus, the overall running time is polynomial.}

\vspace{2mm}
\begin{myproof}[{\bf Proof of Theorem~\ref{thm:general_revenues}}]
For $\delta > 0$, let \((\lambda, x)\) be a \((1 - \delta)\)-approximate solution to~\eqref{eq:problem:primal}, and run Algorithm~\ref{alg:Randomized Static 1/2} with this input. By Lemma~\ref{lemma:MNL-polyhedron-distribution}, the probability that supplier \(j\) observes customer \(i\) in its backlog is exactly \(x_{ij}\). Moreover, for every supplier \(j \in \mathcal{S}\), the events that supplier \(j\) observes each customer \(i\) are mutually independent across customers. 
Thus, for any subset \(C \subseteq \mathcal{C}\), the probability that supplier \(j\) ends up observing exactly the set \(C\) of customers in its backlog is given by
\[
\lambda^{{\sf ind}}_{j,C}
  \;=\;
  \prod_{i \in C} x_{ij}
  \prod_{i \in \mathcal{C} \setminus C} (1 - x_{ij}).
\]
This defines a product distribution over customer subsets, where for each customer \(i\), inclusion in supplier \(j\)'s backlog is modeled as an independent Bernoulli random variable with parameter \(x_{ij}\). We refer to this distribution as the \emph{independent} distribution for supplier \(j\). It is easy to verify that \((\lambda^{{\sf ind}}, x)\) forms a feasible solution to~\eqref{eq:problem:primal}.

Let \({\sf ALG}_{\sf static}\) be the expected revenue obtained by Algorithm~\ref{alg:Randomized Static 1/2}. For any subset $C$ of customers, recall $R_j^*(C)$ is the optimal expected revenue obtained from supplier $j$. Hence, we have
\begin{equation}\label{eq:objective_c_onesidedstatic}
{\sf ALG}_{\sf static}=\sum_{j\in \mathcal{S}} \sum_{C\subseteq \mathcal{C}}
  R_j^*(C)\, \cdot \lambda^{{\sf ind}}_{j,C}.
\end{equation}
Since $(\lambda^{\sf ind}, x)$ and $(\lambda, x)$ are both feasible to \eqref{eq:problem:primal}, \(\lambda^{\sf ind}_{j,\cdot}\) and \(\lambda_{j,\cdot}\) share the same marginal values for every \(j\in\mathcal{S}\). {\color{blue}Therefore, applying Lemma~\ref{lemma:MNL 1/2 coerrelation gap} to each supplier $j\in\mathcal{S}$ with the revenues $\{r_{ij}\}_{i\in\mathcal{C}}$ on its ground set, where $r_{j,\min}=\min_{i\in\mathcal{C}}\{r_{ij}:r_{ij}>0\}$ and $r_{j,\max}=\max_{i\in\mathcal{C}}\{r_{ij}:r_{ij}>0\}$, yields, for all $j \in \mathcal{S}$,
\[
\frac{\sum_{C \subseteq \mathcal{C}} R^*_j(C) \cdot\lambda^{\mathrm{ind}}_{j,C}}{\sum_{C \subseteq \mathcal{C}} R^*_j(C) \cdot\lambda_{j,C}} \geq \frac{1}{{\sf CorrGap}_j},\qquad {\sf CorrGap}_j:=\min\Big\{2,\ \tfrac{r_{j,\max}}{r_{j,\min}}(\tfrac{e}{e-1})\Big\}.
\]
Aggregating over all $j \in \mathcal{S}$ and using that a ratio of sums of nonnegative terms is at least the smallest of the per-supplier ratios, we obtain}
\begin{equation}\label{eq:correlation-gap}
    \frac{\sum_{j \in \mathcal{S}} \sum_{C \subseteq \mathcal{C}} R^*_j(C) \cdot \lambda^{\mathrm{ind}}_{j,C}}{\sum_{j \in \mathcal{S}} \sum_{C \subseteq \mathcal{C}} R^*_j(C) \cdot \lambda_{j,C}} \geq {\color{blue}\min_{j\in\mathcal{S}}\frac{1}{{\sf CorrGap}_j}}.
\end{equation}

Finally, recall from Lemma \ref{lemma:equivalent-LPs} that \eqref{eq:problem:primal} is a relaxation of \atar. Let \((\lambda^*,x^*)\) be an optimal solution to \eqref{eq:problem:primal}, and \(\mathsf{OPT}\) be the optimal value of~\atar. Therefore, 

\begin{align}\label{eq:final-ratio}
\frac{{\sf ALG}_{\sf static}}{\mathsf{OPT}} \geq \frac{{\sf ALG}_{\sf static}}{\eqref{eq:problem:primal}} 
=&\frac{\sum_{j\in\mathcal{S}}\sum_{C\subseteq\mathcal{C}}R_j^*(C)\cdot\lambda_{j,C}^{\sf ind}}{\sum_{j\in\mathcal{S}}\sum_{C\subseteq\mathcal{C}}R_j(C)\cdot\lambda_{j,C}^{\star}}\\
=&\frac{\sum_{j\in\mathcal{S}}\sum_{C\subseteq\mathcal{C}}R_j^*(C)\cdot\lambda_{j,C}^{\sf ind}}{\sum_{j\in\mathcal{S}}\sum_{C\subseteq\mathcal{C}}R_j^*(C)\cdot\lambda_{j,C}}\times \frac{\sum_{j\in\mathcal{S}}\sum_{C\subseteq\mathcal{C}}R_j^*(C)\cdot\lambda_{j,C}}{\sum_{j\in\mathcal{S}}\sum_{C\subseteq\mathcal{C}}R_j(C)\cdot\lambda_{j,C}^{\star}}
\geq{\color{blue}\min_{j\in\mathcal{S}}\frac{1}{{\sf CorrGap}_j}}\times (1-\delta),\notag
\end{align}
where the first component of the last inequality follows from Equation~\eqref{eq:correlation-gap}, while the second term follows from \((\lambda,x)\) being a \((1 - \delta)\)-approximate solution to~\eqref{eq:problem:primal}, and the fact that, for all \( j \in \mathcal{S} \), the function \( R_j^* \) dominates \( R_j \). These together imply that \( \sum_{j \in \mathcal{S}} \sum_{C \subseteq \mathcal{C}} R_j^*(C) \cdot \lambda_{j,C} \geq \sum_{j \in \mathcal{S}} \sum_{C \subseteq \mathcal{C}} R_j(C) \cdot \lambda_{j,C}\geq (1-\delta)\sum_{j \in \mathcal{S}} \sum_{C \subseteq \mathcal{C}} R_j(C) \cdot \lambda_{j,C}^{\star} \).

{\color{blue}Substituting the per-supplier bounds $\frac{1}{{\sf CorrGap}_j}=\max\{\frac{1}{2},\,\tfrac{r_{j,\min}}{r_{j,\max}}(1-\tfrac{1}{e})\}$ from Lemma \ref{lemma:MNL 1/2 coerrelation gap} into Equation \eqref{eq:final-ratio} and choosing $\delta$ appropriately small yields, for any $\epsilon>0$, 
\[
\frac{{\sf ALG}_{\sf static}}{\mathsf{OPT}}\ \ge\ \max\Big\{\tfrac12,\ \min_{j\in\mathcal{S}}\tfrac{r_{j,\min}}{r_{j,\max}}\big(1-\tfrac1e\big)\Big\}-\epsilon\ \ge\ \max\Big\{\tfrac12,\ \tfrac{r_{\min}}{r_{\max}}\big(1-\tfrac1e\big)\Big\}-\epsilon,
\]
where $r_{\min}=\min\{r_{ij}: r_{ij}>0\}$ and $r_{\max}=\max_{i,j}r_{ij}$, and the last inequality uses $\min_{j\in\mathcal{S}}\tfrac{r_{j,\min}}{r_{j,\max}}\ge\tfrac{r_{\min}}{r_{\max}}$. In particular, Algorithm~\ref{alg:Randomized Static 1/2} achieves a $(\tfrac12-\epsilon)$-approximation for~\atar~under general heterogeneous revenues, and a $(1-\tfrac1e-\epsilon)$-approximation when revenues are uniform for each supplier.}
\end{myproof}

{\color{blue}\begin{remark}[Tightness of the analysis and hardness of \atar]\label{rem:lp-tightness}
The guarantees of Theorem~\ref{thm:general_revenues} follow from the correlation-gap bounds of Lemma~\ref{lemma:MNL 1/2 coerrelation gap}. In Appendix~\ref{appendix:lp-tightness} we construct, for general and for uniform revenues, families of \atar\ instances on which the ratio $\mathsf{ALG}_{\sf static}/\eqref{eq:problem:primal}$ converges to $\tfrac12$ and to $1-\tfrac1e$, respectively. Hence our analysis is tight against the relaxation~\eqref{eq:problem:primal}: no argument bounding $\mathsf{ALG}_{\sf static}$ against~\eqref{eq:problem:primal} can give a better guarantee in either case. This tightness is with respect to the relaxation rather than $\mathsf{OPT}$ itself; improving the guarantee for $\mathsf{ALG}_{\sf static}/\mathsf{OPT}$ would therefore require comparing $\mathsf{ALG}_{\sf static}$ directly with $\mathsf{OPT}$ or using a tighter upper bound on $\mathsf{OPT}$ than~\eqref{eq:problem:primal}. In light of the APX-hardness of \atar\ (Lemma~\ref{lem:atar-apx-hard}), whether any polynomial-time algorithm can guarantee a factor greater than $\tfrac12$, and whether the guarantee of our static algorithm is tight against $\mathsf{OPT}$, remain open.
\end{remark}}
\section{$\frac{1}{2}$-Approximation Algorithm for~\atar~with Same‐Order Revenues}\label{sec:1/2-greedy-approx}

In this section, we establish Theorem~\ref{thm:separable_revenues} by presenting an adaptive greedy algorithm for~\atar~under a widely used structural assumption on pairwise revenues. Under this assumption, we improve the approximation guarantee from \((\tfrac{1}{2} - \epsilon)\) for any $\epsilon > 0$ to an exact \(\tfrac{1}{2}\). While the improvement in approximation is modest, the key advantage lies in significantly better practical efficiency. In contrast, the randomized algorithm from Theorem~\ref{thm:general_revenues} involves solving~\eqref{eq:problem:primal} using the Ellipsoid method, which requires repeated calls to an FPTAS-based separation oracle, making it computationally demanding. \textcolor{blue}{We begin by defining the same-order revenue structure and presenting the greedy algorithm (Section~\ref{sec:greedy-alg}), then prove the $\tfrac{1}{2}$-approximation guarantee via a primal-dual argument (Section~\ref{sec:greedy-proof}).}

\vspace{2mm}
\noindent {\bf Same‐Order Revenue Structure.} Let \(\{r_{ij}\}_{i \in \mathcal{C},\, j \in \mathcal{S}}\) represent pair revenues in a two-sided platform. We say that the platform admits a \emph{same-order} revenue structure across suppliers if there exists a permutation over customers \(\sigma_{\mathcal{C}}: \mathcal{C} \to \mathcal{C}\) such that for all \(j \in \mathcal{S}\), we have \(r_{\sigma_{\mathcal{C}}(1)j} \geq r_{\sigma_{\mathcal{C}}(2)j} \geq \dots \geq r_{\sigma_{\mathcal{C}}(n)j}.\) This condition implies that all suppliers rank customers in a common descending order with respect to the pair revenues in the bipartite graph.  

Note that the same-order condition holds for a variety of two-sided platforms. For example, consider a platform where each customer \(i \in \mathcal{C}\) submits a bid \(r_i\) indicating her willingness to pay, and each supplier \(j \in \mathcal{S}\) submits a bid \(r_j\) representing the maximum payment they offer upon a successful match. Assume that the platform's revenue from matching a customer \(i\) with a supplier \(j\) is given by a function \(\mathcal{R}:\RR^2_+\to \RR_+\) of the two bids, i.e., $r_{ij}=\mathcal{R}(r_i, r_j)$. If this function is monotone in \(r_i\) with a fixed direction (non-increasing or non-decreasing) for all \(r_j\), then the revenue structure \(\{r_{ij}\}_{i \in \mathcal{C},\, j \in \mathcal{S}}\) satisfies the same-order condition, meaning that all suppliers rank customers in the same order.

For instance, the following revenue structures satisfy this condition, when for each \(i \in \mathcal{C},\ j \in \mathcal{S}\) we have 
\(r_{ij} = r_i\), where the platform's revenue depends solely on the price paid by the customer; 
\(r_{ij} = r_j\), where it depends solely on the supplier's payment to the platform; 
\(r_{ij} = r_i + r_j\), where it is the additive contribution of both customer and supplier payments; 
\(r_{ij} = r_i \cdot r_j\), where, though less common, it is the product of payments from both sides; 
\(r_{ij} = \min\{r_i, r_j\}\), where it is the minimum of bid payments from each side; and 
\(r_{ij} = \max\{r_i, r_j\}\), where it is the maximum of bid payments from each side.

\subsection{Greedy Algorithm}\label{sec:greedy-alg}

When revenues admit the same‐order property, for all $j\in\mathcal{S}$, the optimal revenue function $R_j^*$ is submodular order with respect to a common order over customers. Exploiting this property, we develop a deterministic adaptive greedy algorithm that effectively incorporates observed agent choices at each decision step. Unlike Algorithm~\ref{alg:Randomized Static 1/2}, which necessitates randomized policies, our adaptive greedy algorithm directly leverages sequentially obtained information, achieving a precise \(\frac{1}{2}\)-approximation factor.

In a setting with a same-order revenue structure, our algorithm proceeds as follows. The platform takes a common descending order of customers, based on pairwise revenues, for all suppliers, and processes customers in this order. At each step, the platform offers the current customer \(i\in\mathcal{C}\) a set of suppliers \(S_i\subseteq\mathcal{S}\) that maximizes the marginal increase in expected revenue, observes the customer’s choice, and moves to the next customer. After all customers are processed, for each supplier \(j\in\mathcal{S}\) let \(C_j^n\) denote the backlog of customers who chose \(j\). For all $j\in\mathcal{S}$, the platform then selects a subset \(C_j^{n*}\subseteq C_j^n\) that maximizes the expected revenue \(R_j(C_j^{n*})\) and offers it to supplier \(j\).

\begin{algorithm}[h]
 \caption{Same-Order Greedy for~\atar}\label{alg:Submodular_order_greedy}
 \begin{algorithmic}[1]
 \State Consider a common descending order of  $i_1,\ldots,i_n$ over the set of customers $\mathcal{C}$ with respect to pairwise revenues for each supplier.
 \State For each $j\in \mathcal{S}$, set $C^0_j=\emptyset$.
\For{$t=1,\ldots,n$}
\State Offer assortment $S_{i_t} \in {\arg\max}_{S \subseteq \mathcal{S}} \left\{ \sum_{j\in S} R_{j}^*(i_t|C^{t-1}_j) \cdot \phi_{i_t}(j,S)\right\}$ to customer $i_t$.
\State Observe the choice $\ell\in S_{i_t}\cup\{0\}$ of customer $i_t$ which happens w.p. $\phi_{i_t}(\ell,S_{i_t})$.
\State If $\ell\neq 0$, update $C_{\ell}^t \leftarrow C_{\ell}^{t-1}\cup\{i_t\}$ and $C_j^{t} \leftarrow C_j^{t-1}$ for all $j\in\mathcal{S}\setminus\{\ell\}$. Otherwise, $C_j^{t} \leftarrow C_j^{t-1}$ for all $j\in\mathcal{S}$.
\EndFor
\State For all $j\in \mathcal{S}$, consider the final backlog $C_j$ of supplier $j$, i.e., the set of all customers who chose supplier $j$. Offer supplier $j$ an optimal assortment $C_j^*\subseteq C_j$ that maximizes its expected revenue, i.e., $C_j^*\in\underset{C\subseteq C_j} {\arg\max} \ R_j(C) $.
 \end{algorithmic}
 \end{algorithm}

 {\color{blue}\noindent\textbf{Running time.} Algorithm~\ref{alg:Submodular_order_greedy} runs in polynomial time. Steps~4 and~8 each require
solving a classical assortment optimization problem under MNL. Since the optimal
MNL assortment is revenue-ordered \citep{talluri2004revenue}, each such problem can be solved in polytime by sorting in
decreasing order of revenue and evaluating each prefix}

 \subsection{Proof of Theorem~\ref{thm:separable_revenues}}\label{sec:greedy-proof}
By the same-order condition, all suppliers rank customers in a common descending order with respect to pairwise revenues in the bipartite graph. Under this ordering, for all \( j \in \mathcal{S} \), the optimal revenue function \( R_j^* \) satisfies the submodular order property, implying that processing customers in this order allows us to leverage the diminishing returns property for each newly processed customer. 

In our analysis, we adapt the randomized primal-dual argument from the analysis of the greedy \(\frac{1}{2}\)-approximation algorithm proposed by \citet{housni2024twosided} for their \emph{one-sided adaptive} matching problem, which is a special case of~\atar\ with equal revenues across all pairs. At each step, their algorithm chooses a customer arbitrarily, and offers her the assortment that maximizes the marginal increase of suppliers’ demand functions, i.e., the probability of selecting an item from a given assortment. 
Their analysis does not directly apply to~\atar~as it exploits the submodularity of the demand functions, an assumption that does not hold for the optimal revenue functions. However, we are still able to derive the desired result by leveraging the existence of a common descending revenue order of customers and the submodular order property of the optimal revenue functions.

To this end, fix a descending order of customers \(i_1,\dots,i_n\) for which each supplier’s optimal revenue function \(R_j^*\) satisfies the submodular order property, such order exists as ensured by our same-order revenue assumption. Treat this ordering as the customer processing order, and for any \(i\in\mathcal{C}\), \(j\in\mathcal{S}\) and \(C\subseteq\mathcal{C}\), define the marginal contribution $R_j^*(i \mid C) := R_j^*\bigl(C\cup\{i\}\bigr) - R_j^*(C).$ These marginals serve as the basis for both the greedy selections in the algorithm and the construction of the dual solution in our analysis.

Our guarantee will hold against the LP relaxation \eqref{eq:relaxation_onesided_customers}. The dual program for \eqref{eq:relaxation_onesided_customers} is:
\begin{align}
 \begin{aligned}
 \min \quad & \sum_{j\in \mathcal{S}}\beta_j + \sum_{i\in \mathcal{C}} \alpha_i &\\
 \text{s.t.} \quad &  \beta_j +\sum_{i\in C}\gamma_{i,j} \geq R_j^*(C),  &\forall  C\subseteq\mathcal{C}, \ j\in \mathcal{S}, \notag\\
 & \alpha_i - \sum_{j\in S}\gamma_{i,j}\cdot \phi_{i}(j,S) \geq 0, &\forall S\subseteq\mathcal{S}, \ i\in \mathcal{C},\notag\\
 & \alpha_i, \ \beta_j, \ \gamma_{i,j}\in\RR, &\forall i\in \mathcal{C}, \ j\in \mathcal{S}.\notag\\
\end{aligned}\label{eq:dual_relaxation_onesided_customers}\tag{\textsc{Dual--I}}
\end{align}
Let \({\sf ALG}_{\sf greedy}\) denote the expected revenue obtained by Algorithm~\ref{alg:Submodular_order_greedy}. We construct random variables \(\tilde{\alpha}_i\), \(\tilde{\gamma}_{i,j}\), and \(\tilde{\beta}_j\) based on the decisions made by Algorithm~\ref{alg:Submodular_order_greedy} along the customer processing sequence. We then demonstrate that their expectations form a feasible dual solution whose objective value is exactly $2$ times \({\sf ALG}_{\sf greedy}\). Thus, by strong duality, \({\sf ALG}_{\sf greedy}\) is at least half of \eqref{eq:relaxation_onesided_customers}.

\paragraph{Construction of the Variables.} At step \(t\), as customer \(i_t\) is processed and offered assortment \(S_{i_t}\), we define the random variable \(\tilde{\alpha}_{i_t}\) by
\[
\tilde{\alpha}_{i_t}
= R_{\xi_{i_t}}^*\bigl(\{i_t\}\mid C^{t-1}_{\xi_{i_t}}\bigr)
= R_{\xi_{i_t}}^*\bigl(C^{t-1}_{\xi_{i_t}}\cup\{i_t\}\bigr) - R_{\xi_{i_t}}^*\bigl(C^{t-1}_{\xi_{i_t}}\bigr),
\]
where \(\xi_{i_t}\) denotes the supplier chosen by \(i_t\) (if \(i_t\) selects the outside option, then \(\tilde{\alpha}_{i_t}=0\)).  For each \(j\in S_{i_t}\cup\{0\}\), the probability that \(\xi_{i_t}=j\) is \(\phi_{i_t}(j,S_{i_t})=\frac{u_{ij}}{1+\sum_{\ell\in S_{i_t}} u_{i\ell}}\), and is zero otherwise.  Thus \(\tilde{\alpha}_{i_t}\) depends on both the history of previous choices and the current selection.  Similarly, for every supplier \(j\in\mathcal{S}\), we set $\tilde{\gamma}_{i_t,j} = R_{j}^*\bigl(\{i_t\}\mid C^{t-1}_j\bigr)$. Observe that \(\tilde{\gamma}_{i_t,j}\) remains random, as it is determined by the random backlog \(C^{t-1}_j\).  At the final step of the algorithm, we define $\tilde{\beta}_j = R_j^*\bigl(C^n_j\bigr)$, representing the optimal expected revenue realized from supplier \(j\) after all customers have been processed.

 \paragraph{Feasibility and Objective Value.} We now verify that the expectations of these random variables yield a feasible solution to~\eqref{eq:dual_relaxation_onesided_customers} whose objective value is within a constant factor of the algorithm’s expected revenue. Clearly,
\begin{align*}
\sum_{i_t\in\C}\tilde{\alpha}_{i_t} &= \sum_{i_t\in\C}\sum_{j\in\S} R_j^*(i_t|C^{t-1}_j)\cdot \mathds{1}_{\{\xi_{i_t}=j\}} \\
&=\sum_{j\in\S}\sum_{i_t\in C_j^n} R_j^*(i_t|C^{t-1}_j)\\
&=\sum_{j\in\S}\left(R_j^*(C^n_j) - R_j^*(\emptyset)\right) = \sum_{j\in\S}R_j^*(C^n_j).
 \end{align*}

In the second equality, we exploit the fact that for each supplier \(j\in\mathcal{S}\) and time \(t\in\{1,\dots,n\}\), the backlog updates according to
\[
C_j^t = 
\begin{cases}
C_j^{t-1}, & \text{if customer } i_t \text{ does not choose } j,\\
C_j^{t-1}\cup\{i_t\}, & \text{if customer } i_t \text{ chooses } j.
\end{cases}
\]
Moreover, the final backlogs \(C_1^n,\dots,C_m^n\) are mutually disjoint since each customer selects at most one supplier.  Hence, by taking expectations we have
\[
\sum_{i\in\mathcal{C}} \EE\bigl[\tilde{\alpha}_i\bigr] + \sum_{j\in\mathcal{S}}\EE\bigl[\tilde{\beta}_j\bigr]
\;=\;
2 \EE\bigl[\sum_{j\in\mathcal{S}} R_j^*\bigl(C_j^n\bigr)\bigr]=2{\sf ALG}_{\sf greedy},
\]
where the second equality holds because the algorithm’s expected revenue for any realization of disjoint backlogs \(\{C_j^n\}_{j \in \mathcal{S}}\) is \(\sum_{j \in \mathcal{S}} R_j^*\bigl(C_j^n\bigr)\). Thus, \({\sf ALG}_{\sf greedy}\) is exactly half of the expected dual objective, thereby establishing, by strong duality, the \(\tfrac{1}{2}\)-approximation guarantee.

 Next, we verify that the expectations of our random variables satisfy the first dual constraint.  Specifically, for every supplier \(j\in\mathcal{S}\) and subset \(C\subseteq\mathcal{C}\), we require
\[
\EE\bigl[\tilde{\beta}_j\bigr] \;+\; \sum_{i\in C}\EE\bigl[\tilde{\gamma}_{i,j}\bigr]
\;\ge\;
R_j^*(C),
\]
where the expectation is taken over the random customer choices. We now show that this constraint is implied by the \emph{interleaved partition} bound, which applies to submodular order functions.
\textcolor{blue}{Intuitively, while classical submodularity allows us to bound the value of a union of sets by summing marginal contributions in any order, the ``interleaved partition'' argument provides a structured way to bound these marginals when the elements are processed according to an order with respect to which the function is submodular-order.}

Given any set \( C \) and a total order \( \pi \) on its elements, \citet{udwani2024submodular} define an \emph{interleaved partition} of \( C \) as a sequence of disjoint sets \( \{O_1, E_1, \ldots, O_k, E_k\} \) partitioning \( C \), where the sets in \( \{O_\ell\}_{\ell=1}^k \) and \( \{E_\ell\}_{\ell=1}^k \) alternate and do not cross in order \( \pi \). They show that for a monotone, sub-additive function \( R^* \) with submodular order \( \pi \), and any set \( A \) with interleaved partition \( \{O_t, E_t\}_{t=1}^k \), we have
\begin{align}
R^*(A) \leq R^* \left( \cup_{t=1}^k E_t \right) + \sum_{t =1}^k R^* \left( O_t \big | \cup_{\ell=1}^{t - 1} E_\ell \right).\label{eq:interleaved-inequality} 
\end{align}

Now, fix any subset \( C \subseteq \mathcal{C} \) and supplier \( j \in \mathcal{S} \). Define the sequence of sets \( \{O_t, E_t\}_{t=1}^n \) such that for each \( t \in \{1,\ldots,n\} \):
\[
E_t \;=\;
\begin{cases}
\{i_t\} & \text{if } i_t \in C^{n}_j,\\[4pt]
\emptyset & \text{otherwise},
\end{cases}
\qquad\text{and}\qquad
O_t \;=\;
\begin{cases}
\{i_t\} & \text{if } i_t \in C \setminus C^{n}_j,\\[4pt]
\emptyset & \text{otherwise}.
\end{cases}
\]

Note that \(\{O_t, E_t\}_{t=1}^n\) partitions \(C \cup C_j^n\), since for each \(t \in \{1, \dots, n\}\), if \(i_t \in C \cup C_j^n\) then \(i_t\) belongs to exactly one set in \(\{O_t, E_t\}_{t=1}^n\), and if \(i_t \notin C \cup C_j^n\) it belongs to neither. Moreover, since all customers \(i_1, \dots, i_n\) are processed in a common submodular order for all suppliers including \(j\), sets \(\{O_t\}_{t=1}^n\) and \(\{E_t\}_{t=1}^n\) alternate and never cross in this order. Thus, the collections \(\{E_t\}_{t=1}^n\) and \(\{O_t\}_{t=1}^n\) form an interleaved partition of \(C \cup C_j^n\).

By Lemma~\ref{lemma:udwani-MNL-SO}, for all $j\in\mathcal{S}$, the function $R_j^*$ is monotone, sub-additive, and submodular order with respect to the common descending order of revenues. Hence, applying Equation~\eqref{eq:interleaved-inequality} to the set $C\cup C^n_j$ yields:
\begin{align*}
R_j^*(C)\leq R_j^*(C\cup C^n_j)&\leq R_j^* (\cup_{t=1}^n E_t)+ \sum_{t =1}^n R_j^* ( O_t \mid \cup_{\ell=1}^{t-1} E_\ell )\\
&= R_j^*(C^n_j)+  \sum_{t =1}^n R_j^*(O_t\mid C^{t-1}_j)   = R_j^*(C^n_j)+   \sum_{i_t \in C\setminus C^n_j} R_j^*(i_t\mid C^{t-1}_j) \\
  &\leq R_j^*(C^n_j)+ \sum_{i_t \in C} R_j^*(i_t\mid C^{t-1}_j)  = \tilde \beta_j+ \sum_{i_t \in C} \tilde\gamma_{i_t,j}.
 \end{align*}
Here, the first inequality holds due to the monotonicity of \(R_j^*\); the second inequality follows from \(\{O_t, E_t\}_{t=1}^n\) forming an interleaved partition of \(C\cup C_j^n\) and Equation~\eqref{eq:interleaved-inequality}. The first equality holds because \(\bigcup_{\ell=1}^{t-1}E_\ell = C_j^{t-1}\) for all \(t\) by definition of \(E_\ell\)'s. The second equality follows by the definition of \(O_\ell\)'s and the fact that \(R_j^*(\emptyset\mid C_j^{t-1})=0\). Finally, the last inequality again uses the monotonicity of \(R_j^*\), which implies \(R_j^*(i_t\mid C_j^{t-1})\ge0\).

For any \( j \in \mathcal{S} \) and \( C \subseteq \mathcal{C} \), we showed that \( R_j^*(C) \le \tilde{\beta}_j + \sum_{i \in C} \tilde{\gamma}_{i,j} \); taking expectation over all trajectories then gives \( R_j^*(C) \le \EE[\tilde{\beta}_j] + \sum_{i \in C} \EE[\tilde{\gamma}_{i,j}] \). Therefore, the first family of constraints in~\eqref{eq:dual_relaxation_onesided_customers} is satisfied. 
Next, we verify that the expected values of the dual variables satisfy the second family of constraints. For every \( i \in \mathcal{C} \) and subset \( S \subseteq \mathcal{S} \), the dual feasibility condition requires:
\begin{equation}
     \mathbb{E}[\tilde{\alpha}_i] - \sum_{j \in S} \mathbb{E}[\tilde{\gamma}_{i,j}] \cdot \phi_i(j, S) \geq 0. \label{eq:aux_primal_dual2}
\end{equation}
To establish this inequality, we condition on the state of the supplier backlogs \( C_j^{t-1} \) at time step \( t \), when customer \( i_t \) is processed. Given these sets, the variables \( \tilde{\gamma}_{i_t,j} \) become deterministic. Thus, the expectation in Equation \eqref{eq:aux_primal_dual2} is solely over the random supplier choice \( \xi_{i_t} \) of customer \( i_t \), which we denote by \( \mathbb{E}_{\xi_{i_t}} \). This gives:
\begin{equation*} \label{eq:aux_primal_dual1}
\mathbb{E}_{\xi_{i_t}}\bigl[\tilde{\alpha}_{i_t} \mid \{C_j^{t-1}\}_{j\in\mathcal{S}}\bigr] = \mathbb{E}_{\xi_{i_t}} \left[ R_{\xi_{i_t}}^*(i_t \mid C_{\xi_{i_t}}^{t-1}) \mid \{C_j^{t-1}\}_{j\in\mathcal{S}} \right] = \sum_{j \in S_{i_t}} R_j^*(i_t \mid C_j^{t-1}) \cdot \phi_{i_t}(j, S_{i_t}).
\end{equation*}
Finally, note that
$$
     \sum_{j\in S}\tilde{\gamma}_{i_t,j}\cdot\phi_{i_t}(j,S) = \sum_{j\in S}R_j^*(i_t|C^{t-1}_j)\cdot\phi_{i_t}(j,S), 
$$ 
where we used that $\phi_{i_t}(j,S)=0$ for any $j\notin S$. Due to step 4 of Algorithm~\ref{alg:Submodular_order_greedy}, conditional on $\{C_j^{t-1}\}_{j\in\mathcal{S}}$ we have:
$$\sum_{j\in S_{i_t}}R_j^*(i_t|C^{t-1}_j) \cdot\phi_{i_t}(j,S_{i_t})\geq \sum_{j\in S}R_j^*(i_t|C^{t-1}_j)\cdot\phi_{i_t}(j,S).$$
Equation~\eqref{eq:aux_primal_dual2} then follows by taking an expectation over $\{C_j^{t-1}\}_{j\in\mathcal{S}}$. 

Since the expected value of our dual variables forms a feasible solution to \eqref{eq:dual_relaxation_onesided_customers}, we  conclude that the expected value obtained by running Algorithm~\ref{alg:Submodular_order_greedy} is at least $1/2$ of~\eqref{eq:relaxation_onesided_customers},  and consequently of the original problem. Therefore, the proof of the $1/2$-approximation of Algorithm~\ref{alg:Submodular_order_greedy} follows.

{\color{blue}\begin{remark}[\textbf{Greedy has no guarantee under general revenues}]\label{rem:greedy-hard-instance}
In the absence of the same-order revenue condition, Algorithm~\ref{alg:Submodular_order_greedy} lacks a theoretical guarantee, and its performance can degrade arbitrarily. In Appendix~\ref{appendix:hard-instance} we construct a family of general-revenue instances $\{\mathcal{I}_n\}_{n\ge2}$, growing in the number of customers, on which the greedy algorithm obtains only a $\Theta(1/n)$ fraction of the optimum; it therefore admits no constant-factor approximation guarantee. This lack of worst-case robustness underscores the necessity of Algorithm~\ref{alg:Randomized Static 1/2}, which guarantees a $(\tfrac12-\epsilon)$-approximation for arbitrary pairwise revenues.
\end{remark}}
\section{Numerical Experiments}\label{sec:numerics}

{\color{blue}
In this section, we evaluate the revenue performance, runtime,
scalability, and robustness of the randomized static and greedy
algorithms. Our main performance measure is the ratio of each
algorithm's expected revenue to an LP upper bound on the optimal
expected revenue of \atar. We compare these empirical ratios with
the worst-case approximation guarantees established in
Theorems~\ref{thm:general_revenues} and~\ref{thm:separable_revenues}.
Section~\ref{sec:exp-setup} describes the algorithm implementations,
performance benchmarks, and evaluation protocol.
Section~\ref{sec:results} presents synthetic experiments on
performance and scalability, the effects of preference and revenue
heterogeneity, and the algorithms' behavior on adversarial instances.
Each experiment is presented together with its instance-generation
procedure. Finally, Section~\ref{sec:numerics-upwork} evaluates both
algorithms on instances calibrated to data from the freelancing
platform Upwork, considering the pooled marketplace and individual
job categories.
}

\subsection{{\color{blue}Evaluation Methodology}}
\label{sec:exp-setup}

{\color{blue}
\noindent\textbf{Algorithms and implementation.}
We evaluate our randomized static
Algorithm~\ref{alg:Randomized Static 1/2} and greedy
Algorithm~\ref{alg:Submodular_order_greedy}. The theoretical
guarantee for the static algorithm holds for general revenue
structures, whereas the guarantee for the greedy algorithm applies
only to same-order revenues. Under the same-order condition, the
greedy algorithm processes customers in their common descending
revenue order. For general revenues, we use it as a heuristic,
processing customers in descending order of their average pairwise
revenue across all suppliers.

Implementing the static algorithm requires solving
\eqref{eq:problem:primal}. Although our theoretical analysis uses
the ellipsoid method, we use column generation in the experiments
for computational efficiency. We consider two implementations.
The first employs the FPTAS of~\citet{chen2025fairassortment} as an
approximate oracle for the pricing problem~\eqref{eq:subdual}.
The second uses a MIP solver to solve the pricing problem exactly,
allowing us to compute an exact LP benchmark and evaluate larger
instances. The column-generation procedures are described in
Appendix~\ref{appendix:exact-lp}.

\vspace{3mm}
\noindent\textbf{Performance benchmarks.}
The guarantees in Theorems~\ref{thm:general_revenues}
and~\ref{thm:separable_revenues} are established relative to the LP
relaxations~\eqref{eq:problem:primal}
and~\eqref{eq:relaxation_onesided_customers}, respectively.
The proof of Lemma~\ref{lemma:equivalent-LPs} shows that
\eqref{eq:problem:primal} upper-bounds
\eqref{eq:relaxation_onesided_customers}, so we use
\eqref{eq:problem:primal} to benchmark both algorithms.

When the LP is solved approximately, Lemma~\ref{lemma:approx}
provides a feasible solution  whose objective value
is at least a $(1-\delta)$ fraction of the LP optimum, for
$0<\delta<1$. We define ${\sf LP}_{\sf approx}$ as
$\tfrac{1}{1-\delta}$ times the objective value of this solution.
This scaling makes ${\sf LP}_{\sf approx}$ an upper bound on the
optimal value of~\eqref{eq:problem:primal}, and therefore also on
\eqref{eq:relaxation_onesided_customers} and the optimal expected
revenue of \atar. Let ${\sf ALG}_{\sf static}$ and ${\sf ALG}_{\sf greedy}$ denote
the expected revenues of the two algorithms. We report the performance ratios
${\sf ALG}_{\sf static}/{\sf LP}_{\sf approx}$ and
${\sf ALG}_{\sf greedy}/{\sf LP}_{\sf approx}$, which provide
pessimistic estimates of the algorithms' performance relative to
the optimum. When~\eqref{eq:problem:primal} is solved exactly, we
denote its optimal value by ${\sf LP}$ and report the corresponding
ratios ${\sf ALG}_{\sf static}/{\sf LP}$ and
${\sf ALG}_{\sf greedy}/{\sf LP}$.

\vspace{3mm}
\noindent\textbf{Evaluation protocol.}
We compute the LP benchmark for each instance and estimate both expected revenues by simulation, unless stated otherwise. Since
Equation~\eqref{eq:objective_c_onesidedstatic} sums over all subsets
of customers, we use Monte Carlo simulation to estimate
${\sf ALG}_{\sf static}$. To estimate ${\sf ALG}_{\sf greedy}$,
we average the revenue over $1000$ simulated adaptive sample paths,
drawing each customer's choice from her MNL model. The LP benchmark
is computed directly, whereas the revenue estimates are subject
to sampling error. Consequently, a few reported ratios slightly
exceed one, even though the LP benchmark upper-bounds the
algorithms' true expected revenues.

Unless stated otherwise, for each pair $(n,m)$ we generate $50$
independent instances and impose a three-hour time limit per
instance. We report the mean, minimum, and maximum performance
ratios across the instances solved within this limit. The
instance-generation procedures and parameter choices specific
to each experiment are described alongside the corresponding
results.

We also report the average runtime of the static algorithm, whose
main computational cost is solving~\eqref{eq:problem:primal}.
The remaining steps take less than a second in our experiments.
In particular, the optimal assortment offered to each supplier
can be computed efficiently because it is a revenue-ordered
set~\citep{talluri2004revenue}. The greedy algorithm also runs
in less than a second on the smaller instances, so we omit its
runtime from Table~\ref{table:revenue-comparison-approx}.
In Table~\ref{table:revenue-comparison-exact}, we additionally
report the greedy simulation time to show how it grows with
the number of customers.

All experiments are run on a machine with an Intel i7 processor,
$20$ cores clocked at $1.30$GHz, and $32$GB of RAM.
}

\subsection{{\color{blue}Experiments on Synthetic Instances}}
\label{sec:results}

{\color{blue}
\noindent\textbf{Performance and scalability.}
We first compare the algorithms on synthetic instances with
same-order and general revenues. For these baseline experiments,
we sample the customer and supplier preference weights as
$u_{ij},w_{ji}\sim\mathrm{Uniform}[0.1,5]$ for all
$i\in\mathcal{C}$ and $j\in\mathcal{S}$. To generate same-order
revenues, we independently sample
$r_i^{\mathcal{C}},r_j^{\mathcal{S}}\sim\mathrm{Uniform}[0.01,1]$
and set $r_{ij}=r_i^{\mathcal{C}}+r_j^{\mathcal{S}}$.
For general revenues, we sample
$r_{ij}\sim\mathrm{Uniform}[0.01,1]$ directly.

Table~\ref{table:revenue-comparison-approx} reports the results
using ${\sf LP}_{\sf approx}$. We set $\delta=0.01$ for
$n=m\in\{2,3,5,8,10\}$ and $\delta=0.1$ for
$n=m\in\{12,15\}$, marked with $^*$. The larger tolerance is
needed because the tighter tolerance exceeds the three-hour time
limit on the larger instances. Increasing $\delta$ can also
inflate the benchmark, making the reported performance ratios
more pessimistic. For reference, the static algorithm's guaranteed
ratio to ${\sf LP}_{\sf approx}$ is $(1-\delta)/2$, namely
$0.495$ for $\delta=0.01$ and $0.45$ for $\delta=0.1$.
The greedy algorithm has a $1/2$-approximation guarantee under
same-order revenues. The reported ratios lie comfortably above
these values.

The greedy algorithm achieves at least as high a mean performance
ratio as the static algorithm in every reported configuration.
Under same-order revenues, its average performance ratio is
approximately $88\%$, compared with $76\%$ for the static
algorithm, with a similar advantage under general revenues.
These results illustrate the practical value of adaptivity,
including in settings where the greedy algorithm is used as a
heuristic.

\begin{table}[H]
\centering
\caption{Performance against ${\sf LP}_{\sf approx}$ under same-order and general revenues. Times report LP computation in seconds; starred instances use $\delta=0.1$.}
\begin{minipage}[t]{0.475\textwidth}
\centering
\resizebox{\linewidth}{!}{
\begin{tabular}{c|ccc|ccc|c}
    \toprule
    & \multicolumn{3}{c|}{${\sf ALG}_{\sf static}/{\sf LP}_{\sf approx}$} & \multicolumn{3}{c|}{${\sf ALG}_{\sf greedy}/{\sf LP}_{\sf approx}$} & ${\sf ALG}_{\sf static}$   \\
    $(n,m)$ & min & mean & max & min & mean & max & Time (s) \\
    \hline
    (2,2) & 0.79 & 0.88 & 0.99 & 0.73 & 0.92 & 0.99 & 0.05 \\
    (3,3) & 0.76 & 0.83 & 0.89 & 0.88 & 0.92 & 0.96 & 1.36 \\
    (5,5) & 0.75 & 0.79 & 0.83 & 0.84 & 0.90 & 0.94 & 87 \\
    (8,8) & 0.74 & 0.76 & 0.79 & 0.88 & 0.93 & 0.97 & 2,463 \\
    (10,10) & 0.73 & 0.75 & 0.77 & 0.87 & 0.89 & 0.91 & 10,726 \\
    $(12,12)^*$ & 0.65 & 0.67 & 0.70 & 0.78 & 0.82 & 0.84 & 542 \\
    $(15,15)^*$ & 0.65 & 0.67 & 0.69 & 0.79 & 0.82 & 0.83 & 2,372 \\
    \bottomrule
\end{tabular}}
\vspace{3mm}

{\textbf{(a)}~Same-Order Revenue Structure}
\end{minipage}%
\hfill%
\begin{minipage}[t]{0.475\textwidth}
\centering
\resizebox{\linewidth}{!}{
\begin{tabular}{c|ccc|ccc|c}
    \toprule
    & \multicolumn{3}{c|}{${\sf ALG}_{\sf static}/{\sf LP}_{\sf approx}$} & \multicolumn{3}{c|}{${\sf ALG}_{\sf greedy}/{\sf LP}_{\sf approx}$} & ${\sf ALG}_{\sf static}$ \\
    $(n,m)$ & min & mean & max & min & mean & max & Time (s) \\
    \hline
    (2,2) & 0.81 & 0.93 & 0.99 & 0.70 & 0.93 & 1.01 & 0.05 \\
    (3,3) & 0.79 & 0.88 & 0.99 & 0.79 & 0.92 & 1.00 & 1.00 \\
    (5,5) & 0.78 & 0.84 & 0.94 & 0.84 & 0.90 & 0.94 & 60 \\
    (8,8) & 0.77 & 0.81 & 0.88 & 0.85 & 0.89 & 0.93 & 2,203 \\
    (10,10) & 0.78 & 0.81 & 0.85 & 0.85 & 0.89 & 0.91 & 6,312 \\
    $(12,12)^*$ & 0.69 & 0.72 & 0.75 & 0.78 & 0.81 & 0.83 & 337 \\
    $(15,15)^*$ & 0.69 & 0.71 & 0.72 & 0.78 & 0.81 & 0.82 & 1,748 \\
    \bottomrule
\end{tabular}}
\vspace{3mm}

{\textbf{(b)}~General Revenue Structure}
\end{minipage}

\label{table:revenue-comparison-approx}
\end{table}

Table~\ref{table:revenue-comparison-exact} reports the results
using the exact benchmark ${\sf LP}$ and extends the comparison
to much larger instances. For market sizes appearing in both
tables, the mean performance ratios are similar or slightly
higher with the exact benchmark. Using ${\sf LP}$ removes the
additional scaling applied in ${\sf LP}_{\sf approx}$.
The computational improvement is substantial: the FPTAS approach
approaches the three-hour limit already at $n=m=10$ with
$\delta=0.01$, whereas the MIP approach solves the reported
instances with $n=m\le100$ in under $30$ seconds on average.
At $n=m=500$, the average runtimes are approximately $178$
seconds under same-order revenues and $522$ seconds under general
revenues. The greedy simulation takes approximately one second
at this largest market size.

\begin{table}[t]
\centering
\caption{Performance and runtime against the exact LP benchmark under same-order and general revenues. Times are in seconds.}
\vspace{0.5em}

\begin{minipage}[t]{0.92\textwidth}
\centering
\scalebox{0.95}{
{\color{blue}
\begin{tabular}{c|ccc|c|ccc|c}
    \toprule
    & \multicolumn{3}{c|}{${\sf ALG}_{\sf static}/{\sf LP}$} &  ${\sf ALG}_{\sf static}$& \multicolumn{3}{c|}{${\sf ALG}_{\sf greedy}/{\sf LP}$} & ${\sf ALG}_{\sf greedy}$\\
    $(n,m)$ & min & mean & max  & Time (s)& min & mean & max & Time (s)\\
    \hline
    (2,2) & 0.80& 0.89& 1.00& 0.03& 0.82 & 0.94& 1.02& 0.00\\
    (3,3) & 0.78& 0.84& 0.94& 0.04& 0.78 & 0.92 & 0.98& 0.00\\
    (5,5) & 0.76& 0.80& 0.86& 0.09& 0.84 & 0.90 & 0.95& 0.00\\
    (10,10) & 0.74& 0.76& 0.79& 0.31& 0.87 & 0.90& 0.93& 0.00\\
    (20,20)     & 0.72& 0.73& 0.75& 0.74   & 0.88& 0.91& 0.93& 0.00\\
    (50,50)     & 0.70& 0.71& 0.72& 3.82   & 0.91& 0.92& 0.94& 0.02\\
    (100,100)   & 0.70& 0.70& 0.71& 14.97  & 0.93& 0.94& 0.94& 0.07\\
    $(250,250)$ & 0.69& 0.70& 0.70& 76.45  & 0.95& 0.95& 0.95& 0.41\\
    $(500,500)$ & 0.69& 0.70& 0.70& 177.93& 0.96& 0.96& 0.96& 1.05\\

    \bottomrule
\end{tabular}}}
\vspace{3mm}

{\textbf{(a)}~Same-Order Revenue Structure}
\end{minipage}

\vspace{3mm} 

\begin{minipage}[t]{0.92\textwidth}
\centering
\scalebox{0.95}{
{\color{blue}
\begin{tabular}{c|ccc|c|ccc|c}
    \toprule
    & \multicolumn{3}{c|}{${\sf ALG}_{\sf static}/{\sf LP}$} &  ${\sf ALG}_{\sf static}$& \multicolumn{3}{c|}{${\sf ALG}_{\sf greedy}/{\sf LP}$} & ${\sf ALG}_{\sf greedy}$\\
    $(n,m)$ & min & mean & max  & Time (s)& min & mean & max & Time (s)\\
    \hline
    (2,2) & 0.82& 0.93& 1.00& 0.03& 0.78& 0.95& 1.01& 0.00\\
    (3,3) & 0.79& 0.88& 1.00& 0.05& 0.79& 0.92 & 0.99& 0.00\\
    (5,5) & 0.77& 0.85& 0.94& 0.10& 0.85& 0.91& 0.97& 0.00\\
    (10,10) & 0.77& 0.81& 0.85& 0.35& 0.87 & 0.90& 0.94& 0.00\\
    (20,20)     & 0.75& 0.77& 0.80& 1.20&   0.89& 0.91& 0.93& 0.00\\
    (50,50)     & 0.73& 0.74& 0.75& 5.80&   0.91& 0.92& 0.93& 0.01\\
    (100,100)   & 0.71& 0.72& 0.72& 28.83&  0.93& 0.93& 0.94& 0.06\\
    $(250,250)$ & 0.70& 0.70& 0.70& 144.57& 0.95& 0.95& 0.95& 0.31\\
    $(500,500)$ & 0.69& 0.69& 0.69& 521.70& 0.96& 0.96& 0.96& 1.04\\

    \bottomrule
\end{tabular}}}
\vspace{3mm}

{\textbf{(b)}~General Revenue Structure}
\end{minipage}

\label{table:revenue-comparison-exact}
\end{table}

\vspace{3mm}
\noindent\textbf{Preference heterogeneity.}
We next examine how the diversity of agent preferences affects
performance. We vary the number of base MNL preference types on
each side of the market, denoted by $q\in\{1,\ldots,10\}$.
When $q=1$, all agents on a given side share the same base model;
larger values of $q$ allow more diverse preferences.

For each $q$, we independently generate $q$ base models on each
side. Customer-side base models have preference weights drawn
from $\mathrm{Uniform}[0,1]$. Each customer is assigned uniformly
at random to one of these models, and her weights are scaled by
an individual factor $a_i\in\{0.5,1,2\}$. Supplier-side models
are generated analogously, using $\mathrm{Exp}(1)$ base
distributions. To induce sparsity, we independently set each
entry of every base model to zero with probability
$\sqrt{2}/2\approx0.71$, leaving about $29\%$ of the entries
positive on average. Revenues are drawn from
$\mathrm{Uniform}[0.01,1]$. We consider
$n=m\in\{50,100,250\}$ and generate $50$ independent instances
for each pair $(n,q)$.

Figure~\ref{fig:heterogeneity} shows that the static algorithm's
performance improves substantially with $q$, while the greedy
algorithm's performance remains high and roughly flat. This
pattern holds across all three market sizes and is consistent
with reduced congestion as preferences diversify.
With few sparse preference types, demand is concentrated on a
small set of suppliers. As preferences become more diverse,
demand spreads across more suppliers and competition for the
same matches decreases. The static algorithm benefits
particularly from this reduction because it fixes assortments
in advance, whereas the greedy algorithm can adapt assortments
to previously observed customer choices. The greedy algorithm
outperforms the static algorithm throughout these experiments.

\begin{figure}[t]
\centering
\begin{minipage}{0.31\textwidth}
    \centering
    \includegraphics[width=\linewidth]{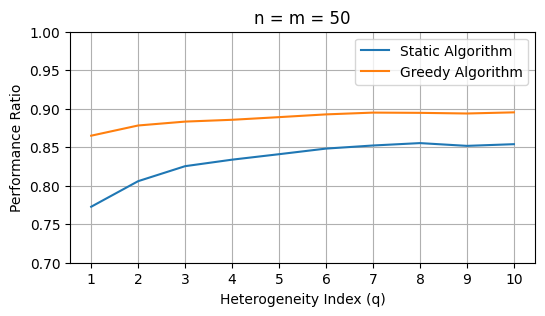}\\[2pt]
    {\small (a) $n = m = 50$}
\end{minipage}
\hfill
\begin{minipage}{0.31\textwidth}
    \centering
    \includegraphics[width=\linewidth]{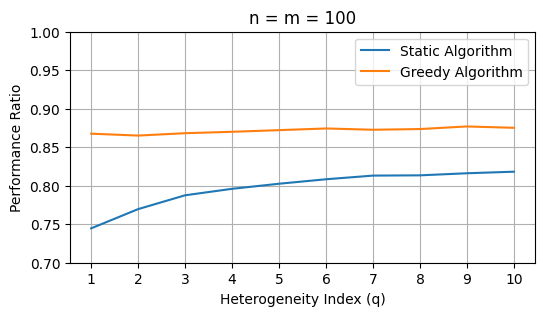}\\[2pt]
    {\small (b) $n = m = 100$}
\end{minipage}
\hfill
\begin{minipage}{0.31\textwidth}
    \centering
    \includegraphics[width=\linewidth]{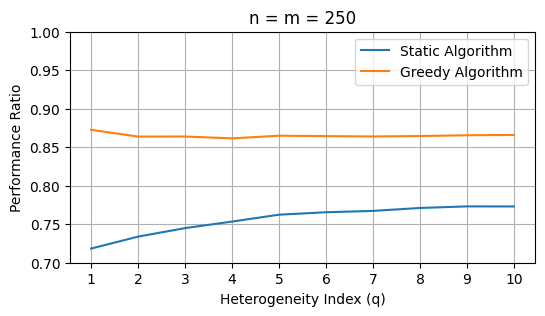}\\[2pt]
    {\small (c) $n = m = 250$}
\end{minipage}
\vspace{3mm}
\caption{Mean performance ratios as a function of preference heterogeneity $q$ for three market sizes.}
\label{fig:heterogeneity}
\end{figure}

\vspace{3mm}
\noindent\textbf{Revenue heterogeneity.}
We also examine how performance responds to the revenue
structure. Theorem~\ref{thm:general_revenues} gives the static
algorithm a $(1-\tfrac1e-\epsilon)$-approximation guarantee when
revenues are uniform across the edges incident to each supplier,
and a $(\tfrac12-\epsilon)$-approximation guarantee under general
revenues. Let $r_{\min}$ and $r_{\max}$ denote the smallest and
largest positive pairwise revenues. The proof also gives the
revenue-dependent guarantee
\[
\max\left\{\frac12,\,
\frac{r_{\min}}{r_{\max}}\left(1-\frac1e\right)\right\}
-\epsilon,
\]
which is non-decreasing in $r_{\min}/r_{\max}$ and reaches
$1-\tfrac1e-\epsilon$ when all revenues are equal.

For this experiment, we fix $n=m=50$ and sample revenues
uniformly from an interval, varying the revenue-ratio parameter
$r_{\min}/r_{\max}$ from $0$ (highly heterogeneous revenues)
to $1$ (uniform revenues). We average the performance ratios over $50$ instances
for each setting. Figure~\ref{fig:revenue-spread} shows that the
algorithms respond differently. As revenues become more uniform,
the greedy algorithm's ratio increases from $0.92$ to $0.96$,
while the static algorithm's ratio decreases from $0.74$ to
$0.69$, moving closer to the bound
$1-\tfrac1e\approx0.63$ for an exact LP solution.

Thus, although the static algorithm's worst-case guarantee
improves as revenues become more uniform, its empirical
performance decreases in these experiments. The heterogeneous
instances generated here differ from the hard family in
Appendix~\ref{appendix:lp-tightness}, which makes the static
algorithm's ratio to the LP benchmark approach $1/2$.
The observed trend therefore describes this particular instance
class and does not imply that revenue heterogeneity generally
benefits the static algorithm. For the greedy algorithm,
empirical performance and the available guarantees improve
in the same direction: it has no constant-factor guarantee
under general revenues, but retains its $1/2$ guarantee when
revenues are uniform. These results show that a stronger
worst-case guarantee need not translate into higher empirical
performance for the static algorithm.

\begin{figure}[t]
\centering
\includegraphics[width=0.62\textwidth]{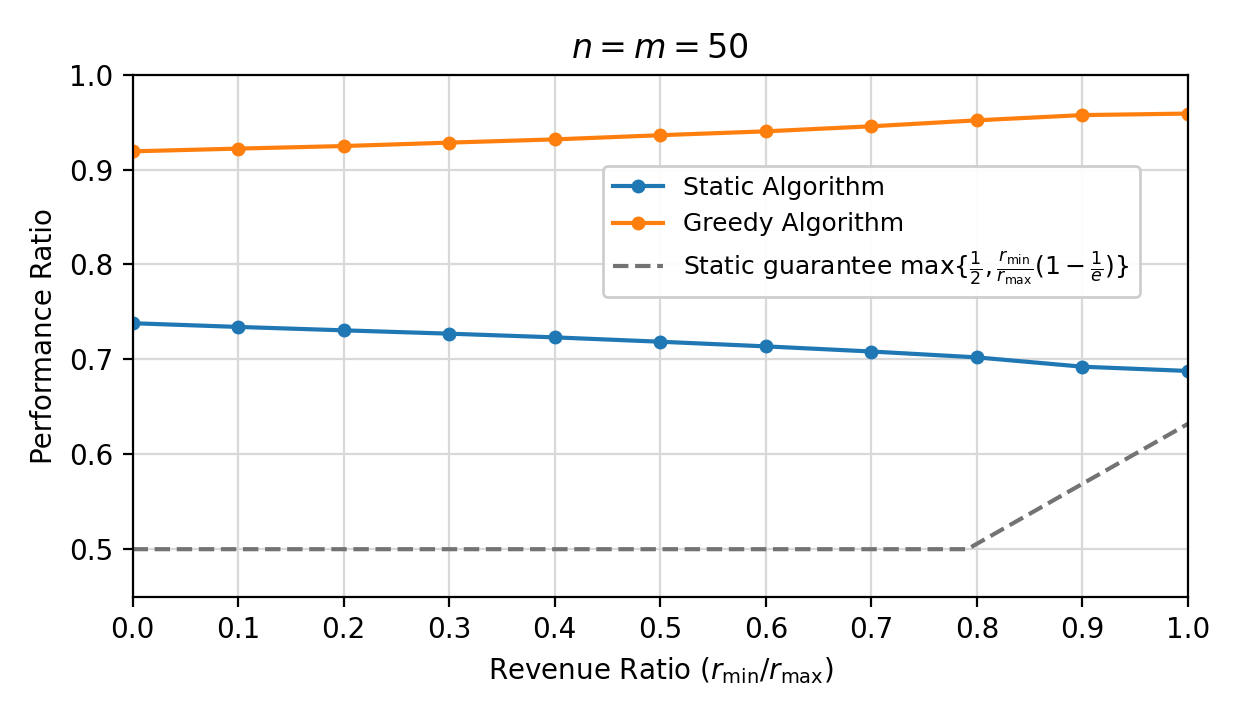}
\caption{Mean performance ratios as a function of $r_{\min}/r_{\max}$ for $n=m=50$. The dashed curve shows the theoretical bound for an exact LP solution, $\max\{\tfrac12,\tfrac{r_{\min}}{r_{\max}}(1-\tfrac1e)\}$.}
\label{fig:revenue-spread}
\end{figure}

\vspace{3mm}
\noindent\textbf{Adversarial instances.}
The greedy algorithm matches or exceeds the static algorithm's
mean performance throughout the preceding experiments.
However, this empirical advantage does not extend to all
instances. Appendix~\ref{appendix:hard-instance} constructs a
family $\{\mathcal{I}_n\}_{n\ge2}$ with general pairwise revenues
on which the greedy algorithm, run under our heuristic processing
order, obtains only a $\Theta(1/n)$ fraction of the optimum
in the asymptotic regime established there. The static algorithm
retains its $(\tfrac12-\epsilon)$ guarantee on this family.
The construction places a congesting customer early in the
processing order. Once this customer selects the hub, the
marginal value of routing each subsequent light customer to
that supplier becomes very small, causing the greedy algorithm
to miss the revenue available from serving the light customers
together.

Table~\ref{table:hard-instance-numeric} reports results for
$\rho=1.2$, $a=1$, $b=\tfrac12$, $V=1$, $p=50$,
$\kappa=20$, and $\theta=10^{10}$. We compute
${\sf ALG}_{\sf greedy}$ exactly under the heuristic order and
evaluate ${\sf ALG}_{\sf static}$ from the static algorithm's
LP solution. For this large value of $\theta$, we use the
asymptotic value $(n-1)V+A$, where
$A=\tfrac{ab}{1+b}=\tfrac13$, to approximate ${\sf OPT}$;
the difference does not affect the reported precision.
The greedy ratios decrease with $n$, consistently with the
asymptotic $\Theta(1/n)$ bound, while the static ratios round
to $1.00$ on every reported instance. These results demonstrate
that the greedy algorithm's strong performance on the preceding
synthetic instances can coexist with poor performance on
adversarial instances.

\begin{table}[H]
\centering
\caption{Performance on the adversarial family $\mathcal{I}_n$ under the heuristic processing order. The reported ratios approximate ${\sf OPT}$ by $(n-1)V+A$.}
\label{table:hard-instance-numeric}
\vspace{0.5em}
{\color{blue}\begin{tabular}{l|cccccc}
    \toprule
    $n$ & $3$ & $5$ & $10$ & $20$ & $50$ & $100$\\
    \midrule
    ${\sf ALG}_{\sf greedy}/{\sf OPT}$ & $0.515$ & $0.277$ & $0.129$ & $0.063$ & $0.025$ & $0.013$\\
    ${\sf ALG}_{\sf static}/{\sf OPT}$ & $1.00$ & $1.00$ & $1.00$ & $1.00$ & $1.00$ & $1.00$\\
    \bottomrule
\end{tabular}}
\end{table}
}

\subsection{{\color{blue}Experiments on Marketplace-Calibrated Data}}
\label{sec:numerics-upwork}

{\color{blue}
\noindent\textbf{Data and calibration.}
We complement the synthetic experiments with instances calibrated
to data from the freelancing platform Upwork. We combine two public
datasets: a corpus of roughly $50{,}000$ Upwork job postings and a
dataset of freelancer attributes. Freelancers represent customers,
and jobs, or equivalently the clients who posted them, represent
suppliers.

We obtain the two-sided preference weights from skill profiles
generated by a large language model (LLM). A single query per agent
produces a low-dimensional skill profile. The preference weights
$u_{ij}$ and $w_{ji}$ are then computed from the cosine similarity
between the freelancer's and the job's profiles and from the
freelancer's earnings on the job as a fraction of the job's budget.
Appendix~G provides the full data-processing and calibration
procedure, including how job budgets, job durations, and freelancer
rates determine the revenues $r_{ij}$.

The platform's percentage-fee revenue model provides a useful
connection to the same-order condition. If freelancer $i$ charges
an hourly rate $p_i$ and job $j$ requires $L_j$ hours, the platform's
revenue is
\[
r_{ij}=(\mathrm{fee}\cdot p_i)L_j.
\]
Before infeasible pairs are removed, this has the multiplicative
form $r_{ij}=r_i^{\mathcal{C}}r_j^{\mathcal{S}}$, which satisfies
the same-order condition. We set the parameters of unaffordable
or cross-category pairs to zero. This feasibility truncation can
break the common revenue order, so the guarantee of
Theorem~\ref{thm:separable_revenues} need not apply to the resulting
instances.

We consider both the pooled marketplace and each of its five job
categories. For each setting, we draw $100$ independent
sub-instances with $n=m=50$. We use the exact benchmark
${\sf LP}$ and evaluate both algorithms following the methodology
in Section~\ref{sec:exp-setup}.

\vspace{3mm}
\noindent\textbf{Pooled marketplace.}
We first draw sub-instances from the full marketplace, pooling
all job categories. Table~4 reports the results. The greedy
algorithm achieves a mean performance ratio of $0.94$ and runs
in well under a second, while the static algorithm achieves
$0.78$ at a higher computational cost. The static algorithm
comfortably exceeds its worst-case guarantee. The greedy
algorithm also performs well despite the possible violations
of the same-order condition introduced by feasibility
truncation.

{\color{blue}
\begin{table}[t]
\centering
\caption{{\color{blue}Performance of the static and greedy algorithms against the exact benchmark ${\sf LP}$ on Upwork-calibrated instances, averaged over $100$ independent sub-instances with $n=m=50$.}}
\label{table:upwork}
\vspace{0.5em}
{\color{blue}\scalebox{0.95}{
\begin{tabular}{ccc|c|ccc|c}
    \toprule
    \multicolumn{3}{c|}{${\sf ALG}_{\sf static}/{\sf LP}$} & ${\sf ALG}_{\sf static}$ & \multicolumn{3}{c|}{${\sf ALG}_{\sf greedy}/{\sf LP}$} & ${\sf ALG}_{\sf greedy}$\\
    min & mean & max & Time (s) & min & mean & max & Time (s)\\
    \hline
    0.74&  0.78&  0.84&  5.70&  0.90&  0.94&  0.95& 0.02\\
    \bottomrule
\end{tabular}}}
\end{table}}

\vspace{3mm}
\noindent\textbf{Individual job categories.}
Because the marketplace organizes listings by category, we also
evaluate the algorithms separately within each of the five
categories. Every pair within a category is category-compatible,
so feasibility depends only on affordability. These markets are
relatively dense, with between $76\%$ and $88\%$ of
customer--supplier pairs feasible.

Figure~4 reports the results. The greedy algorithm achieves mean
performance ratios between $0.94$ and $0.95$ in every category,
while the static algorithm achieves ratios between $0.70$ and
$0.73$, well above its worst-case guarantee. Both algorithms'
performance is stable across the five categories.

{\color{blue}
\begin{figure}[t]
\centering
\includegraphics[width=0.72\textwidth]{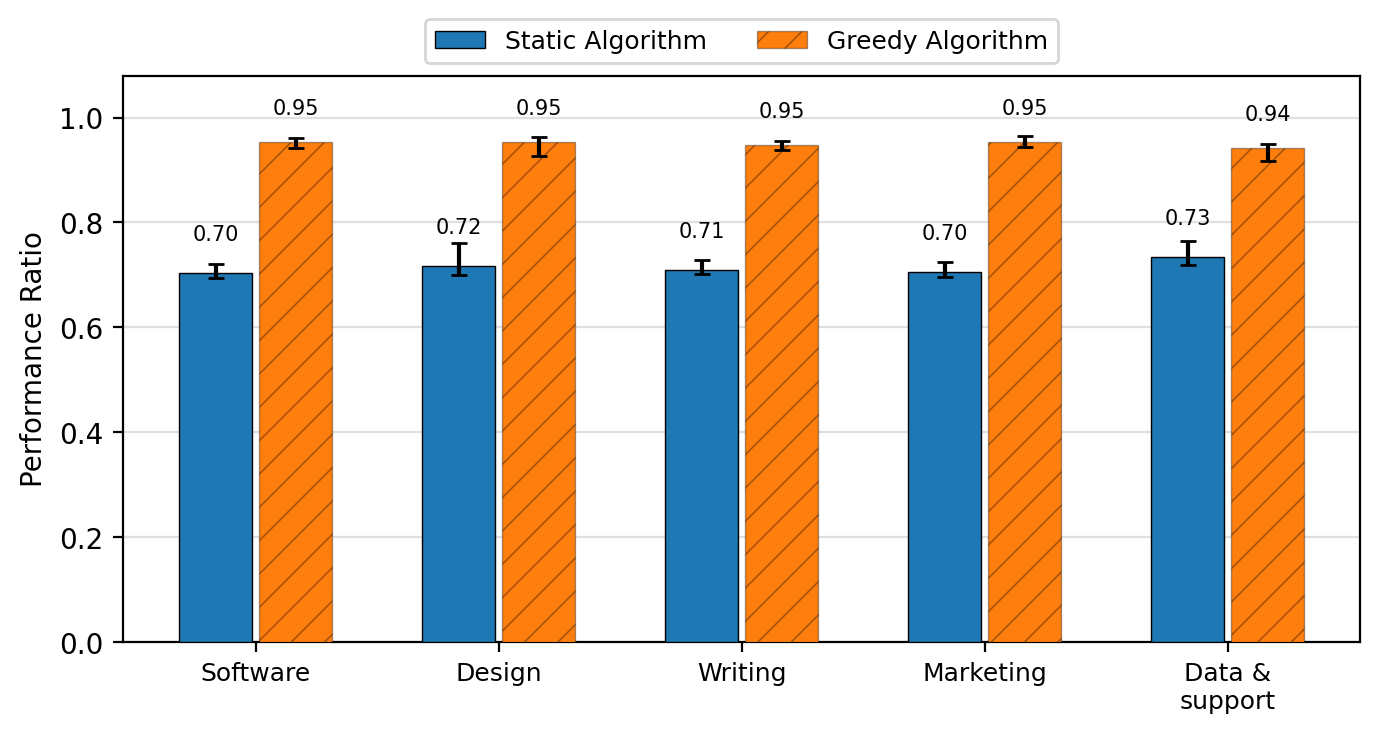}
\caption{{\color{blue}Per-category performance of the static and greedy algorithms against the exact benchmark ${\sf LP}$ on Upwork-calibrated instances, averaged over $100$ independent sub-instances with $n=m=50$ drawn from within each category. Bars show the mean ${\sf ALG}/{\sf LP}$ ratio and the whiskers the minimum and maximum across sub-instances. Algorithm~\ref{alg:Randomized Static 1/2} remains well above its $\tfrac12$ worst-case guarantee, and both algorithms are stable across the five categories.}}
\label{fig:upwork-category}
\end{figure}}

The static algorithm's ratios are somewhat lower than in the
pooled marketplace. This difference is consistent with greater
congestion within individual categories: more freelancers qualify
for the same jobs, increasing competition for those jobs.
The static algorithm fixes its assortments in advance, whereas
the greedy algorithm can respond to previously observed customer
choices. This explanation is consistent with the congestion
mechanism observed in the preference-heterogeneity experiments
of Section~\ref{sec:results}.

These experiments show that the greedy algorithm performs well
at low computational cost on the marketplace-calibrated instances.
The static algorithm retains a worst-case guarantee for arbitrary
pairwise revenues and can also be used when adaptive feedback is
unavailable.
}
\section{Conclusion}\label{sec:conclusion}
We studied adaptive two-sided assortment optimization for maximizing expected revenue in two-sided platforms where agents on both sides make single choices according to their  MNL choice models. Departing from prior work that focuses on number of matches within this context, we proposed approximation algorithms for revenue maximization under heterogeneous pairwise revenues.

{\color{blue}Our randomized static algorithm guarantees a $(\tfrac12-\epsilon)$-approximation for arbitrary pairwise revenues, improving to $(1-\tfrac1e-\epsilon)$ when revenues are uniform across the edges incident to each supplier, while our adaptive greedy algorithm achieves a $\tfrac12$-approximation under the same-order revenue condition. We further showed that, beyond the same-order case, the greedy algorithm admits no constant-factor guarantee: we construct a family of general-revenue instances on which it obtains only a $\Theta(1/n)$ fraction of the optimum, which identifies the static algorithm as the only one of the two that retains a worst-case guarantee for general revenues. Complementing these results, our experiments show that the greedy
algorithm achieves high revenue at low computational cost on the
randomly generated synthetic and Upwork-calibrated instances.
The adversarial experiments illustrate the complementary strength
of the static algorithm: it achieves near-optimal revenue on
instances where the greedy algorithm performs poorly.
}

{\color{blue}Our work opens several directions for future research, some tied to simplifications inherent in the model. First, each agent selects at most one counterpart; allowing multiple selections per agent, as on dating platforms, is a natural extension. Second, we consider a single-shot market and abstract away capacities, deadlines, and repeated interaction over time, which are natural to incorporate in future work. One simple extension is to re-solve the single-round problem over the agents left unmatched after each round; deriving a guarantee for the resulting multi-round policy is left for future work. Third, we assume known MNL preferences on both sides, as is standard in the assortment optimization literature; in practice these are estimated and may vary across agent cohorts, a form of heterogeneity examined in the sensitivity analysis of Section~\ref{sec:numerics}. Beyond these, generalizing to alternative choice models such as mixtures of MNL, to the \emph{fully adaptive} setting in which the platform may alternate between the two sides and choose any agent at each step, and to cardinality-constrained assortments are further natural directions.} {\color{blue}Finally, we show that~\atar~is APX-hard and therefore admits no PTAS; closing the gap between this inapproximability threshold and the $(\tfrac12-\epsilon)$-approximation of our static algorithm remains an interesting direction for future work.}

{
\addcontentsline{toc}{section}{Bibliography}
\bibliographystyle{plainnat}
\bibliography{ref}
}

\newpage

\begin{APPENDICES}
\section*{Appendix}
\section{{\color{blue}Proofs from Section~\ref{sec:preliminaries}}}

\subsection{Properties of Optimal Revenue Function under MNL}\label{Appendix:proof:non-submodularity-g_j}
Let $\mathcal{N}$ be a finite set of alternatives. The MNL-based optimal revenue function is defined as:
$$
R^*(A)=\underset{A'\subseteq A}{\max}\ \sum_{i\in A'} r_i\,\phi(i,A'), \qquad  \text{for all } A \subseteq \mathcal{N},
$$
where $r_i \ge 0$ is the revenue of alternative $i \in \mathcal{N}$, and $\phi(i, A')$ is the MNL choice probability $\phi(i, A') = \frac{u_i}{u_0 + \sum_{ k \in A'} u_k}, \ \text{for } i \in A'.$ We normalize the utility of the outside option to $u_0 = 1$.

\vspace{3mm}
\noindent {\bf Monotonicity of the Optimal Revenue Function.} Let $A \subseteq B \subseteq \mathcal{N}$. Any subset $A' \subseteq A$ is also contained in $B$, so
$$
R^*(B) \ge \sum_{i \in A'} r_i\, \phi(i, A'), \quad \text{for all } A' \subseteq A.
$$
By selecting the maximizer $A^* \subseteq A$ for $R^*(A)$, we get
$$
R^*(B) \ge \sum_{i \in A^*} r_i\, \phi(i, A^*) = R^*(A),
$$
thus $R^*$ is a monotone function.

\vspace{3mm}
\noindent {\bf Non-submodularity of the Optimal Revenue Function.}
Now we show the optimal revenue function $R^*$ is not submodular. Let $\mathcal{N} = \{1,2,3\}$, and define:
$$
\begin{aligned}
&u_1 = 1,\quad u_2 = 1,\quad u_3 = 3,\quad u_0 = 1,\\
&r_1 = 4,\quad r_2 = 3,\quad r_3 = 2.\\
\end{aligned}
$$
We compute value of function $R^*$ on some relevant sets:
$R^*(\{1,3\}) =2$, 
$R^*(\{3\}) = \frac{3}{2}$, 
$R^*(\{1,2,3\}) =\frac{7}{3}$, 
$R^*(\{2,3\}) =\frac{9}{5}$.
Now, let us compare marginal gains:
\begin{align*}
&R^*(\{1,3\}) - R^*(\{3\}) = 2 - \frac{3}{2} = 0.5\\
&R^*(\{1,2,3\}) - R^*(\{2,3\}) =\frac{7}{3}-\frac{9}{5}\approx 0.53
\end{align*}
Thus, we have $R^*(\{1,3\}) - R^*(\{3\}) < R^*(\{1,2,3\}) - R^*(\{2,3\})$. Therefore, the marginal gain of adding item $1$ to $\{2,3\}$ is greater than to $\{3\}$, contradicting submodularity. Hence, the diminishing returns condition is violated, and function $R^*$ is not submodular.

This proof shows that under heterogeneous revenues, the optimal revenue function $R^*$ is monotone but not submodular.

\vspace{3mm}
\subsection{\color{blue}{Proof of Lemma~\ref{lem:atar-apx-hard}: \bf APX-Hardness of~\atar~\label{appendix:apx-hardness}}}
{\color{blue}
We prove that  computing $\mathsf{OPT}$ for \atar\ is APX-hard, so the problem admits no PTAS unless $\mathsf{P}=\mathsf{NP}$.

\begin{proof}
\text{We} reduce from \textsc{Max-3DM-3}, a restricted version of Maximum 3-Dimensional Matching. An instance consists of three pairwise disjoint sets $X,Y,Z$ with $|X|=|Y|=|Z|=\nu$ and a collection of triples $T\subseteq X\times Y\times Z$, where each element of $X\cup Y\cup Z$ belongs to at most three triples. A \emph{matching} is a subset $\mathcal M\subseteq T$ of pairwise disjoint triples, meaning that no two triples in $\mathcal M$ share an element. The objective is to find a matching of maximum cardinality. A matching is \emph{perfect} if it covers every element of $X\cup Y\cup Z$, or equivalently, if $|\mathcal M|=\nu$. Let $m:=|T|$. Since each triple contains three elements and each element belongs to at most three triples, we have $3m\le 9\nu$, and hence $m\le 3\nu$.
The problem \textsc{Max-3DM-3} is MAX SNP-complete \citep{KANN199127}, and therefore APX-hard. Moreover, Theorem~2.2 of \citet{ostrovsky2014s} establishes that there exists an absolute constant $\delta\in(0,1)$ such that it is $\mathsf{NP}$-hard to distinguish between \emph{yes}-instances, which admit a perfect matching of size $\nu$, and \emph{no}-instances, in which every matching has cardinality at most $(1-\delta)\nu$.

We transform an arbitrary instance of \textsc{Max-3DM-3} into an \atar\ instance by the following map, and then show that if \atar\ admits a PTAS, we can distinguish the yes- and no-instances above in polynomial time, which is impossible unless $\mathsf{P}=\mathsf{NP}$.

\textit{Construction.}
Given an instance $(X,Y,Z,T)$ of \textsc{Max-3DM-3}, we create one supplier for each triple in $T$ and one \emph{primary customer} for each element of $X\cup Y\cup Z$. Thus, the supplier set is $\mathcal S=T$, and the primary customer set is $\mathcal C_P=X\cup Y\cup Z$. For each supplier $j\in\mathcal S$, we also introduce a distinct \emph{background customer} $d_j$. Let $\mathcal C_B=\{d_j:j\in\mathcal S\}$ and $\mathcal C=\mathcal C_P\cup\mathcal C_B$. The resulting \atar\ instance has $m$ suppliers and $3\nu+m$ customers.
Let $\delta$ be the constant in the gap-hardness result above, and set
\[
M:=\left\lceil \frac{2040\,\nu}{\delta}\right\rceil.
\]
We define the preference weights and revenues as follows:
\begin{itemize}
    \item \emph{Background customers.} For each supplier $j\in\mathcal S$, its background customer $d_j$ has
    \[
    u_{d_j,j}=M,\qquad
    w_{j,d_j}=100,\qquad
    r_{d_j,j}=L:=84.
    \]

    \item \emph{Primary customers.} For each primary customer $i\in\mathcal C_P$ and each supplier $j\in\mathcal S$ whose triple contains $i$, set
    \[
    u_{ij}=M,\qquad
    w_{ji}=1,\qquad
    r_{ij}=H:=120.
    \]
\end{itemize}
All remaining customer--supplier preference weights and pairwise revenues are zero. Since $\delta$ is an absolute constant, $M$ has bit-length $O(\log\nu)$, and the construction is computable in time polynomial in $\nu$ and $m$.

A supplier with zero weight in a customer's MNL model is never selected by that customer, and removing it from an assortment does not change the customer's other choice probabilities. We may therefore restrict attention to assortments containing only suppliers with positive customer-side preference weights. 
Figure~\ref{fig:apx-reduction} illustrates the construction.

For any realization of the customer choices, let $C_j$ denote the final backlog of supplier $j$, and define
$
W:=\sum_{j\in\mathcal S}R_j^*(C_j).
$
Thus, $W$ is the expected revenue conditional on these backlogs when each supplier is offered an optimal assortment. The optimal expected revenue of the constructed instance is
$
\mathsf{OPT}:=\max_{\pi}\mathbb E_{\pi}[W],
$
where $\pi$ ranges over feasible \atar\ policies. We write $\mathsf{OPT}_{\mathrm{yes}}$ and $\mathsf{OPT}_{\mathrm{no}}$ for this value when the original \textsc{Max-3DM-3} instance is a yes-instance and a no-instance, respectively.

\begin{figure}[t]
\centering
{\color{blue}
\begin{tikzpicture}[>=stealth,
  cust/.style={circle, draw=blue, fill=blue!8, minimum size=6.5mm, inner sep=0pt, font=\small, text=blue},
  bg/.style={circle, draw=blue, densely dashed, fill=blue!3, minimum size=6.5mm, inner sep=0pt, font=\small, text=blue},
  sup/.style={rectangle, rounded corners, draw=blue, fill=blue!8, minimum height=7mm, inner sep=4pt, font=\small, text=blue},
  matched/.style={sup, very thick, fill=blue!20},
  inc/.style={draw=blue},
  bgedge/.style={draw=blue, densely dashed}]
\node[bg]   (d1) at (0.10,4.37) {$d_1$};
\node[cust] (x1) at (-0.10,5.23){$x_1$};
\node[cust] (y1) at (0.90,4.77) {$y_1$};
\node[cust] (z1) at (0.70,5.63) {$z_1$};
\node[bg]   (d2) at (0.10,2.17) {$d_2$};
\node[cust] (x2) at (-0.10,3.03){$x_2$};
\node[cust] (y2) at (0.90,2.57) {$y_2$};
\node[cust] (z2) at (0.70,3.43) {$z_2$};
\node[bg]   (d3) at (0.4,0.6)   {$d_3$};
\node[font=\footnotesize\itshape] at (0.4,6.2) {customers};
\node[matched] (T1) at (5.8,4.9) {$T_1=\{x_1,y_1,z_1\}$};
\node[matched] (T2) at (5.8,2.7) {$T_2=\{x_2,y_2,z_2\}$};
\node[sup]     (T3) at (5.8,0.6) {$T_3=\{x_1,y_2,z_2\}$};
\node[font=\footnotesize\itshape] at (5.8,6.2) {suppliers ($=$ triples)};
\draw[inc] (x1)--(T1); \draw[inc] (y1)--(T1); \draw[inc] (z1)--(T1);
\draw[inc] (x2)--(T2); \draw[inc] (y2)--(T2); \draw[inc] (z2)--(T2);
\draw[inc] (x1)--(T3); \draw[inc] (y2)--(T3); \draw[inc] (z2)--(T3);
\draw[bgedge] (d1)--(T1); \draw[bgedge] (d2)--(T2); \draw[bgedge] (d3)--(T3);
\end{tikzpicture}}
\caption{{\color{blue}The reduction on a small \textsc{Max-3DM-3} instance with $\nu=2$. Each triple is a supplier; the three elements of a triple are its primary customers, joined by solid edges carrying $u_{ij}=M$, $w_{ji}=1$, and $r_{ij}=H=120$. Each supplier $j$ additionally has a private background customer $d_j$, joined by a dashed edge with $u_{d_j,j}=M$, $w_{j,d_j}=100$, and $r_{d_j,j}=L=84$. The bold suppliers $T_1,T_2$ form a perfect matching covering every element once. 
}}
\label{fig:apx-reduction}
\end{figure}

\textit{Step 1: computing the optimal revenue of a supplier.}
Fix a supplier $j\in\mathcal S$, and let $C$ be a set of $k\in\{0,1,2,3\}$ primary customers belonging to its triple. Define
\[
g(k):=R_j^*(\{d_j\}\cup C).
\]
Since all suppliers have the same structure and their primary customers have identical preference weights and revenues, $g(k)$ depends only on $k$, not on $j$ or the choice of $C$.

Since the primary customers have the highest revenue $H=120$, an optimal assortment includes all $k$ of them. The remaining decision is whether to include the background customer $d_j$, so
\[
g(k)=\max\left\{
\frac{kH}{1+k},
\frac{100L+kH}{101+k}
\right\}.
\]
For $k\in\{0,1,2\}$, the primary-only assortment yields revenue
$\frac{kH}{1+k}\le 80<L=84$, so adding $d_j$ increases revenue.
For $k=3$, the primary-only assortment yields revenue
$\frac{3H}{4}=90>L$, so adding $d_j$ decreases revenue.
Thus, the background customer is retained when $k\le 2$ and excluded when $k=3$, giving
\[
g(0)=\frac{8400}{101},\qquad
g(1)=\frac{8520}{102}=\frac{1420}{17},\qquad
g(2)=\frac{8640}{103},\qquad
g(3)=90.
\]
Define the revenue gains relative to the background-only value $g(0)$ by
\begin{align*}
\Delta_1&:=g(1)-g(0)=\frac{620}{1717}\approx0.3611,\\
\Delta_2&:=g(2)-g(0)=\frac{7440}{10403}\approx0.7152,\\
\Delta_3&:=g(3)-g(0)=\frac{690}{101}\approx6.8317.
\end{align*}

\textit{Step 2: obtaining a matching from supplier backlogs.}
Fix any feasible policy and any realization of the customer choices. For each supplier $j\in\mathcal S$, let $k_j\in\{0,1,2,3\}$ denote the number of primary customers in its final backlog, and define
\[
\mathcal M:=\{j\in\mathcal S:k_j=3\}.
\]
Each triple in $\mathcal M$ has all three of its elements in the corresponding supplier's backlog. Since each customer selects at most one supplier, no two triples in $\mathcal M$ can share an element. Thus, $\mathcal M$ is a feasible matching in the original \textsc{Max-3DM-3} instance.

For each $k\in\{0,1,2,3\}$, let $m_k:=|\{j\in\mathcal S:k_j=k\}|$. Then $m_3=|\mathcal M|$ and $\sum_{k=0}^3m_k=m$. Since there are $3\nu$ primary customers and each selects at most one supplier,
\begin{equation}\label{eq:apx-customer-budget}
m_1+2m_2+3m_3\le 3\nu.
\end{equation}
The backlog of supplier $j$ consists of its $k_j$ primary customers and possibly its background customer $d_j$. By monotonicity of $R_j^*$, its contribution to $W$ is therefore at most $g(k_j)$, giving
\begin{equation}\label{eq:apx-revenue-bound}
W\le \sum_{j\in\mathcal S}g(k_j)
= m\,g(0)+m_1\Delta_1+m_2\Delta_2+m_3\Delta_3.
\end{equation}
The values computed in Step~1 satisfy
\begin{equation}\label{eq:apx-gain-order}
\frac{\Delta_3}{3}>\Delta_1>\frac{\Delta_2}{2}.
\end{equation}
Combining \eqref{eq:apx-revenue-bound} with
\eqref{eq:apx-gain-order} and \eqref{eq:apx-customer-budget}, we obtain
\[
W
\le m\,g(0)+\frac{\Delta_3}{3}(m_1+2m_2+3m_3)
\le m\,g(0)+\nu\,\Delta_3.
\]
Since this holds for every realization under every feasible policy, taking expectations and maximizing over policies yields
\begin{equation}\label{eq:apx-universal-bound}
\mathsf{OPT}
\le m\,g(0)+\nu\,\Delta_3
\le 3\nu\,g(0)+\nu\,\Delta_3
=\frac{25890}{101}\,\nu,
\end{equation}
where the second inequality uses $m\le 3\nu$.

\vspace{2mm}
\textit{Step 3: upper bound for no-instances.}
Suppose the original \textsc{Max-3DM-3} instance is a no-instance. Since $\mathcal M$ is a feasible matching, we have
\begin{equation}\label{eq:apx-no-matching}
m_3=|\mathcal M|\le (1-\delta)\nu.
\end{equation}
By \eqref{eq:apx-gain-order}, we have $\Delta_2\le 2\Delta_1$ and $\Delta_3-3\Delta_1>0$. Starting from \eqref{eq:apx-revenue-bound} and using
\eqref{eq:apx-customer-budget} and \eqref{eq:apx-no-matching}, we obtain
\begin{align*}
W
&\le m\,g(0)+\Delta_1(m_1+2m_2)+\Delta_3m_3\\
&\le m\,g(0)+\Delta_1(3\nu-3m_3)+\Delta_3m_3\\
&=m\,g(0)+3\nu\,\Delta_1+(\Delta_3-3\Delta_1)m_3\\
&\le m\,g(0)+3\nu\,\Delta_1
   +(\Delta_3-3\Delta_1)(1-\delta)\nu\\
&=m\,g(0)+\nu\,\Delta_3-(\Delta_3-3\Delta_1)\delta\nu.
\end{align*}
This bound holds for every realization under every feasible policy. Taking expectations and maximizing over policies therefore gives
\begin{equation}\label{eq:apx-no-bound}
\mathsf{OPT}_{\mathrm{no}}
\le m\,g(0)+\nu\,\Delta_3-\frac{9870}{1717}\,\delta\nu
=:\tau_{\mathrm{no}},
\end{equation}
where we used
$\Delta_3-3\Delta_1=\frac{690}{101}-\frac{1860}{1717}
=\frac{9870}{1717}$.

\textit{Step 4: lower bound for yes-instances.}
Suppose the original \textsc{Max-3DM-3} instance is a yes-instance, and fix a perfect matching $\mathcal M^\ast$ of size $\nu$. For each primary customer $i\in\mathcal C_P$, let $j^\ast(i)\in\mathcal M^\ast$ denote the unique triple containing $i$. Consider the static policy $\pi_0$ that offers each primary customer $i$ the singleton $\{j^\ast(i)\}$ and each background customer $d_j$ the singleton $\{j\}$.

Each customer has preference weight $M$ for its offered supplier and therefore selects it with probability $1-q$, where
\begin{equation}\label{eq:apx-outside-probability}
q:=\frac{1}{1+M}\le \frac{\delta}{2040\,\nu}.
\end{equation}
Let $E$ be the event that all customers select their offered suppliers. Since there are $3\nu+m\le 6\nu$ customers, a union bound gives
\[
\mathbb P_{\pi_0}[E]\ge 1-(3\nu+m)q\ge 1-6\nu q.
\]
On $E$, every supplier in $\mathcal M^\ast$ has all three of its primary customers and its background customer in its backlog, yielding value $g(3)$. Every other supplier has only its background customer, yielding value $g(0)$. Thus, on $E$, we have
\[
W=V^\ast:=\nu\,g(3)+(m-\nu)\,g(0)
=m\,g(0)+\nu\,\Delta_3.
\]
Since $W$ is nonnegative,
\[
\mathbb E_{\pi_0}[W]
\ge \mathbb P_{\pi_0}[E]\,V^\ast
\ge (1-6\nu q)V^\ast.
\]
Using \eqref{eq:apx-outside-probability} and the bound
$V^\ast\le \frac{25890}{101}\nu$ from
\eqref{eq:apx-universal-bound}, we obtain
\[
6\nu q\,V^\ast
\le \frac{6}{2040}\cdot\frac{25890}{101}\,\delta\nu
<\delta\nu.
\]
The optimal policy performs at least as well as $\pi_0$, so
\begin{equation}\label{eq:apx-yes-bound}
\mathsf{OPT}_{\mathrm{yes}}
\ge \mathbb E_{\pi_0}[W]
\ge m\,g(0)+\nu\,\Delta_3-\delta\nu
=:\tau_{\mathrm{yes}}.
\end{equation}

\vspace{2mm}
\textit{Step 5: conclusion.}
The thresholds defined in \eqref{eq:apx-no-bound} and
\eqref{eq:apx-yes-bound} satisfy
\begin{equation}\label{eq:apx-threshold-gap}
\tau_{\mathrm{yes}}-\tau_{\mathrm{no}}
=\left(\frac{9870}{1717}-1\right)\delta\nu
=\frac{8153}{1717}\,\delta\nu>0.
\end{equation}
Suppose there exists a polynomial-time algorithm that, on every
\atar\ instance, computes a number $\mathsf{ALG}$ satisfying
$
(1-\epsilon)\mathsf{OPT}\le \mathsf{ALG}\le\mathsf{OPT},
$
where $\epsilon:=0.009\,\delta$.
On a yes-instance, \eqref{eq:apx-universal-bound} and
\eqref{eq:apx-threshold-gap} give
\[
\epsilon\,\mathsf{OPT}_{\mathrm{yes}}
\le 0.009\,\delta\cdot\frac{25890}{101}\,\nu
<\frac{8153}{1717}\,\delta\nu
=\tau_{\mathrm{yes}}-\tau_{\mathrm{no}}.
\]
Consequently, using \eqref{eq:apx-yes-bound}, we get
$
\mathsf{ALG}
\ge (1-\epsilon)\mathsf{OPT}_{\mathrm{yes}}
\ge \tau_{\mathrm{yes}}
    -\epsilon\,\mathsf{OPT}_{\mathrm{yes}}
>\tau_{\mathrm{no}}.
$
On a no-instance, \eqref{eq:apx-no-bound} instead yields
$
\mathsf{ALG}\le\mathsf{OPT}_{\mathrm{no}}\le\tau_{\mathrm{no}}.
$
Therefore, testing whether $\mathsf{ALG}>\tau_{\mathrm{no}}$
distinguishes the yes- and no-instances of \textsc{Max-3DM-3}
in polynomial time, contradicting the gap-hardness result unless
$\mathsf{P}=\mathsf{NP}$. In particular, the reduction establishes
that approximating $\mathsf{OPT}$ within a factor
$1-0.009\,\delta$ is $\mathsf{NP}$-hard, and hence \atar\
admits no PTAS unless $\mathsf{P}=\mathsf{NP}$.

\end{proof}
}

\section{Proofs from Section~\ref{sec:1/2-epsilon-approx}}
\subsection{Proof of Lemma~\ref{lemma:LPrelaxation-1}~\label{proof:sec:1/2-epsilon-approx:LP-relxation}}

    We present {\color{blue}a simplified} version of the proof of Lemma 2 in \citet{housni2024twosided} to prove our Lemma~\ref{lemma:LPrelaxation-1} by constructing a feasible solution $(\lambda,\tau)$ to~\eqref{eq:relaxation_onesided_customers} {\color{blue}from an optimal policy for \atar}. Denote by $\pi^\star$ an optimal policy for~\atar{\color{blue}, and take all probabilities and expectations below under $\pi^\star$}. Recall that in~\atar, customers are processed sequentially, followed by the processing suppliers. In any sample path of~$\pi^\star$, each supplier is assigned an optimal subset of the customers from its backlog, that is, a subset that maximizes the platform’s expected revenue from that supplier, given the corresponding set of customers that choose each of them.

{\color{blue}Since each customer is processed exactly once, the following random variables are well defined under $\pi^\star$, even though the customer processed at each step and the assortment offered to it may depend on the history observed so far. For each customer $i \in \mathcal{C}$, let $S_i$ denote the assortment offered to $i$ and let $\xi_i \in S_i\cup\{0\}$ denote the option it selects, and for each supplier $j \in \mathcal{S}$, let $C_j:=\{i\in\C:\ \xi_i=j\}$ denote its final backlog.

We define our variables as the laws of these random variables: for every $j \in \mathcal{S}$ and $C \subseteq \mathcal{C}$,
\[
\lambda_{j,C}:=\PP\bigl[C_j = C\bigr],
\]
and for all $i\in\C, \ S\subseteq\S$,
\[
\tau_{i,S}:=\PP\bigl[S_i = S\bigr].
\]

Both families are nonnegative. Since $C_j$ takes exactly one value on every sample path, and so does $S_i$, we have $\sum_{C\subseteq\C}\lambda_{j,C}=1$ for all $j\in\S$ and $\sum_{S\subseteq\S}\tau_{i,S}=1$ for all $i\in\C$, which are the first and third families of constraints of~\eqref{eq:relaxation_onesided_customers}. We now establish the second family, which links the two sides of the market. For any $i \in \mathcal{C}$ and $j \in \mathcal{S}$, conditioning on the assortment offered to customer~$i$ gives
\begin{align*}
\sum_{C\subseteq \mathcal{C}:C\ni i}\lambda_{j,C}&=\sum_{C\subseteq\C}\PP\bigl[C_j=C\bigr]\cdot\one_{\{i\in C\}}=\PP\bigl[i\in C_j\bigr]=\PP\bigl[\xi_i=j\bigr]\\
&=\sum_{S\subseteq\S}\PP\bigl[S_i=S\bigr]\cdot\PP\bigl[\xi_i=j\mid S_i=S\bigr]=\sum_{S\subseteq\S: S\ni j}\tau_{i,S}\cdot\phi_i(j,S),
\end{align*}
where the first equality writes the restricted sum as a sum over all subsets weighted by $\one_{\{i\in C\}}$; the second holds since the events $\{C_j=C\}$, over all $C\subseteq\C$, partition the sample space; the third holds by the definition of the backlog $C_j$; the fourth is the law of total probability; and the last holds since the assortment $S_i$ is determined before the choice of customer~$i$ is observed, so that $\PP[\xi_i=j\mid S_i=S]=\phi_i(j,S)$, which is zero whenever $j\notin S$.

Finally, since each supplier $j\in\S$ is assigned an optimal subset of its final backlog $C_j$, the expected revenue that $\pi^\star$ collects from supplier $j$ is $\EE[R_j^*(C_j)]$, so the objective value achieved by $\pi^\star$ is
\[
\sum_{j\in\S}\EE\bigl[R_j^*(C_j)\bigr]=\sum_{j\in\S}\sum_{C\subseteq\C}R_j^*(C)\cdot\lambda_{j,C},
\]
where the equality holds by the definition of $\lambda_{j,C}$. Hence $(\lambda,\tau)$ is feasible for~\eqref{eq:relaxation_onesided_customers} and its objective value equals the expected revenue of the optimal policy, which proves that~\eqref{eq:relaxation_onesided_customers} is a relaxation of~\atar.}

\subsection{Proof of Lemma~\ref{lemma:approx}~\label{appendix:proof:lemma:approx}}
Assume for some \(\delta>0\), we have access to \(\mathcal{A}\), a polynomial‐time \((1-\delta)\)-approximation algorithm for Problem~\eqref{eq:subdual} for each \(j\in\mathcal{S}\). The proof of Lemma~\ref{lemma:approx} proceeds in two parts.
Part 1 establishes that given the approximation algorithm \(\mathcal{A}\), there exists an Ellipsoid‐based algorithm that returns a \((1-\delta)\)-approximate solution to~\eqref{eq:problem:primal}.
Part 2 demonstrates that the overall running time of this Ellipsoid-based algorithm is polynomial in the input size.

\noindent{\bf Part 1.} The first part of the proof consists of two steps.  
First, we show that running the Ellipsoid method with a \((1 - \delta)\)-approximate separation oracle on the dual program~\eqref{eq:problem:dual} yields a dual objective value \(r'\) that is at least a \((1 - \delta)\)-factor of the optimal primal objective~\eqref{eq:problem:primal}; specifically, \(r' \ge (1 - \delta) \cdot \text{\eqref{eq:problem:primal}}\).  
Second, for each supplier \(j \in \mathcal{S}\), let \(\mathcal{V}_j \subseteq 2^\mathcal{C}\) denote the collection of subsets of customers that violate the~\eqref{constraint:dual-AC} constraint at the end of the Ellipsoid method iterations. We then introduce an auxiliary dual program~\eqref{eq:problem:dual-aux} that keeps the~\eqref{constraint:dual-AC} constraints only for those sets \(C \in \mathcal{V}_j\) for every \(j \in \mathcal{S}\), together with its corresponding auxiliary primal program~\eqref{eq:problem:primal-aux}, where variables \(\lambda_{j,C}\)'s are enforced to be zero for all \(C \notin \mathcal{V}_j\). We show that the optimal objective values of both auxiliary programs exceed \(r'\), and consequently is at least \((1 - \delta) \cdot \text{\eqref{eq:problem:primal}}\).

\underline{\textit{Step 1.}} To run the Ellipsoid method and solve~\eqref{eq:problem:dual}, we need an efficient separation oracle. An ideal separation oracle consists of three parts: the first verifies whether the~\eqref{constraint:dual-AC} constraint holds for all $j \in \mathcal{S}$; the second checks whether $\frac{\alpha_{ij}}{u_{ij}} + \sum_{\ell \in \mathcal{S}} \alpha_{i\ell} - \gamma_{ij} \geq 0$ for all $i \in \mathcal{C},\ j \in \mathcal{S}$; and the third ensures that $\alpha_{ij} \geq 0$ for all $i \in \mathcal{C},\ j \in \mathcal{S}$. Since the problem~\eqref{eq:subdual} for any $j\in\S$ is NP-hard (see \citet{kunnumkal2010}), we run the Ellipsoid method on~\eqref{eq:problem:dual} using the $(1 - \delta)$-approximate algorithm $\mathcal{A}$ as part of our separation oracle.

Let $(\boldsymbol\alpha',\boldsymbol\beta',\boldsymbol\gamma')$ be the solution returned at termination, and set \(r' \;=\; \sum_{j\in\mathcal{S}}\beta'_j \;+\; \sum_{(i,j)\in\mathcal{C}\times\mathcal{S}}\alpha'_{ij}\) that is the dual objective. Although $(\boldsymbol\alpha',\boldsymbol\beta',\boldsymbol\gamma')$ does not need to satisfy every dual constraint, it is \emph{approximately} feasible in the following sense. When the oracle evaluates this point, the first part of the oracle solves $\subdual_j(\boldsymbol\gamma')$ and produces a set $C_j^{\mathcal{A}}$ such that \(R_j(C_j^{\mathcal{A}})-\sum_{i\in C_j^{\mathcal{A}}}\gamma'_{ij}
\;\;\ge\;\;
(1 - \delta)\cdot\subdual_j(\boldsymbol\gamma').\) The Ellipsoid method accepts the point only if \(R_j(C_j^{\mathcal{A}})-\sum_{i\in C_j^{\mathcal{A}}}\gamma'_{ij}\le\beta'_j\), even though it may still be the case that 
\(
\subdual_j(\boldsymbol\gamma')\;=\;\max_{C\subseteq\mathcal{C}}
\bigl\{R_j(C)-\sum_{i\in C}\gamma'_{ij}\bigr\} \;>\;\beta'_j.
\)
We begin by showing that the dual objective obtained through this process remains close to the optimal value of the primal objective~\eqref{eq:problem:primal}. To establish this, consider solving~\eqref{eq:problem:dual} under the restriction that \(\boldsymbol\alpha = \boldsymbol\alpha'\) and \(\boldsymbol\gamma = \boldsymbol\gamma'\) are fixed. This restricted problem admits a straightforward solution since the only remaining decision variables are the \(\beta_j\)'s, which must be selected to satisfy the constraint in~\eqref{constraint:dual-AC} with equality for each \(j \in \mathcal{S}\). Specifically, we set
\[
\boldsymbol\beta^*_j = \subdual_j(\boldsymbol\gamma') = \max_{C \subseteq \mathcal{C}} \left\{ R_j(C) - \sum_{i \in C} \gamma'_{ij} \right\},
\]
and the resulting objective value is
\[
r_{\boldsymbol\beta^*} = \sum_{j \in \mathcal{S}} \subdual_j(\boldsymbol\gamma') + \sum_{(i,j) \in \mathcal{C} \times \mathcal{S}} \alpha'_{ij}.
\]
It follows directly that \(r_{\boldsymbol\beta^*} \geq \text{\eqref{eq:problem:dual}}\), as fixing \(\boldsymbol\beta = \boldsymbol\beta^*\) is equivalent to introducing more constraints into~\eqref{eq:problem:dual}.
Note that we also have
$$
r' = \sum_{j \in \mathcal{S}} \beta'_j + \sum_{(i,j) \in \mathcal{C} \times \mathcal{S}} \alpha'_{ij} 
\geq (1 - \delta) \cdot  \sum_{j \in \mathcal{S}} \subdual_j(\boldsymbol\gamma') + (1 - \delta) \cdot \sum_{(i,j) \in \mathcal{C} \times \mathcal{S}} \alpha'_{ij}
= (1 - \delta) \cdot  r_{\boldsymbol\beta^*}\,,
$$
where the inequality follows from $\beta'_j \geq R_j(C_j^\mathcal{A})-\sum_{i\in C_j^\mathcal{A}}\gamma'_{ij}\geq (1 - \delta) \cdot \subdual_j(\boldsymbol\gamma')$ and $\alpha'_{ij} \geq 0$.
Together, the two inequalities above give 
\begin{equation}
\label{eqn:key-ineq}
r' \geq (1 - \delta) \cdot r_{\boldsymbol\beta^*}\ge  (1 - \delta) \cdot \eqref{eq:problem:dual} = (1 - \delta) \cdot \eqref{eq:problem:primal}\,,
\end{equation}
where the last equality holds by strong duality.\footnote{Since both~\eqref{eq:problem:primal} and~\eqref{eq:problem:dual} are feasible and bounded, they each admit an optimal solution. By the strong duality theorem (see Theorem~4.4 in~\citet{bertsimas-LPbook}), their optimal values are equal, i.e., \(\eqref{eq:problem:primal} = \eqref{eq:problem:dual}\).} The inequality in~\eqref{eqn:key-ineq} plays a central role and will serve as a foundation for the arguments developed in Step~2.

\underline{\textit{Step 2.}} We first introduce an auxiliary version of the dual problem, defined as follows:
\begin{align}
\begin{aligned}
\min \quad & \sum_{j \in \mathcal{S}} \beta_j + \sum_{(i,j) \in \mathcal{C} \times \mathcal{S}} \alpha_{ij}\\
  \text{s.t.}\quad&  \beta_j \geq R_j(C)-\sum_{i\in C}\gamma_{ij},  &\forall \  C\in \mathcal{V}_j{\color{blue}\cup\{\emptyset\}}, \ j\in \mathcal{S},\\
     &  \frac{\alpha_{ij}}{u_{ij}} + \sum_{\ell \in \mathcal{S}} \alpha_{i\ell} \geq \gamma_{ij}, &\forall \ i\in \mathcal{C}, \ j\in \mathcal{S},\\
    &   \alpha_{ij} \geq 0,\ \beta_j, \gamma_{ij}\in \mathbb{R}, & \forall \ i\in \mathcal{C}, \ j\in \mathcal{S}.
\end{aligned}
\tag{\textsc{Dual--II--Aux}} \label{eq:problem:dual-aux}
\end{align}
Observe that the auxiliary dual program differs from~\eqref{eq:problem:dual} by keeping constraint~\eqref{constraint:dual-AC} only on sets \(C \in \mathcal{V}_j\) for each \(j \in \mathcal{S}\), where \(\mathcal{V}_j\) denotes the collection of subsets of customers for which constraint~\eqref{constraint:dual-AC} was violated for supplier \(j\) during solving~\eqref{eq:problem:dual} using the Ellipsoid method. As noted in the following remark, the total number of such violated constraints across all $j\in\mathcal{S}$ is polynomial in the input size.
\begin{remark}\label{remark:ellipsoid-runtime}
In~\eqref{eq:problem:dual}, we have $\mathcal{O}(nm)$ variables. As shown by \citet{bertsimas-LPbook}, given a separation oracle, the Ellipsoid method would solve \eqref{eq:problem:dual} in at most $\mathcal{O}((nm)^{6}\log(nmU))$ iterations, where $U$ is the bit complexity of \eqref{eq:problem:dual}. Within each iteration, one violated constraint is identified. Hence, we must have that  $\sum_{j\in \mathcal{S}} |\mathcal{V}_j|$, i.e., the number of sets for which the~\eqref{constraint:dual-AC} constraint in \eqref{eq:problem:dual} is violated for any $j\in\mathcal{S}$ during running the Ellipsoid method, is polynomial in the input size.
\end{remark}
The constraints that were not violated played no role in any of the iterations of the Ellipsoid method when solving~\eqref{eq:problem:dual}, applying the Ellipsoid method, with the same \((1 - \delta)\)-approximate separation oracle, to the auxiliary dual~\eqref{eq:problem:dual-aux} will return the same solution \((\boldsymbol\alpha', \boldsymbol\beta', \boldsymbol\gamma')\), yielding the same objective value \(r'\). {\color{blue} This is because every cut used by the original Ellipsoid run is also a constraint of the auxiliary dual; therefore the auxiliary dual cannot have a feasible solution with objective value below \(r'\).}
As in our earlier analysis, observe that applying the Ellipsoid method with an approximate separation oracle effectively expands the feasible region of the linear program. As a result, the solution \((\boldsymbol\alpha', \boldsymbol\beta', \boldsymbol\gamma')\) obtained at termination may not be feasible to the original LP, and the corresponding objective value \(r'\) may fall below the true optimum. In particular, we have
\(
r' \leq \eqref{eq:problem:dual-aux}.
\)
Combining this with the inequality established in~\eqref{eqn:key-ineq}, we obtain the following bound:
\begin{equation}
\label{eqn:key-ineq-2}
\eqref{eq:problem:dual-aux} \geq r' \geq (1 - \delta) \cdot \eqref{eq:problem:primal}.
\end{equation}
Consider the primal counterpart to the auxiliary dual problem \eqref{eq:problem:dual-aux} defined as:
\begin{align}
\begin{aligned}
 \max \quad & \sum_{j\in \mathcal{S}}\sum_{C \subseteq \mathcal{C}}R_j(C)\cdot \lambda_{j,C}&\\
\text{s.t.} \quad & \sum_{C\subseteq \mathcal{C}}\lambda_{j,C} = 1,  &\forall \ j\in \mathcal{S}, \notag\\
& \sum_{C\subseteq \mathcal{C}: C\ni i} \lambda_{j,C} =  x_{ij}, &\forall \ i\in \mathcal{C}, \ j\in \mathcal{S}, \notag\\
& \frac{x_{ij}}{u_{ij}} + \sum_{\ell \in \mathcal{S}} x_{i\ell} \leq 1, &\forall \ i\in \mathcal{C}, \ j\in \mathcal{S}, \notag\\
& \lambda_{j,C}=0,  &\forall \ C\not\in \mathcal{V}_j{\color{blue}\cup\{\emptyset\}},\ j\in \S, \notag\\
& \lambda_{j,C}, \ x_{ij}\geq 0,  &\forall \ j\in \S, \ C\subseteq \mathcal{C}, \ i\in \mathcal{C}. \notag
    \end{aligned}\label{eq:problem:primal-aux} \tag{\textsc{LP--II--Aux}}
\end{align}

This formulation differs from~\eqref{eq:problem:primal} by imposing \(\lambda_{j,C} = 0\) for all \(C \notin \mathcal{V}_j\) and for every \(j \in \mathcal{S}\). After solving our original dual problem using the Ellipsoid method, we enforce $\lambda_{j,C}=0$ for all $j\in \S$ and $C\not\in \mathcal{V}_j$. As highlighted in Remark~\ref{remark:ellipsoid-runtime}, we will end up with a polynomial number of variables, i.e., $\{ \widehat{x}_{ij}: i\in \mathcal{C},\ j\in \mathcal{S}\}\cup\{ \widehat{\lambda}_{j,C}:j\in \mathcal{S},\ C\in  \mathcal{V}_j\}$, to consider for the auxiliary primal problem. We can then solve the auxiliary primal problem using any polynomial-time algorithm and find an optimal basic feasible solution.

 Let \((\widehat\lambda, \widehat x)\) denote the optimal solution to~\eqref{eq:problem:primal-aux}. By the strong duality\footnote{Both~\eqref{eq:problem:primal-aux} and~\eqref{eq:problem:dual-aux} are feasible and bounded, and therefore each admits an optimal solution. For~\eqref{eq:problem:primal-aux}, {\color{blue}setting \(\lambda_{j,\emptyset} = 1\) for all \(j \in \mathcal{S}\), all other \(\lambda_{j,C}=0\), and \(x = 0\) yields a feasible solution, and every feasible solution has objective value at most \(|\mathcal{S}| \cdot \underset{i\in\mathcal{C},j\in\mathcal{S}}{\max}\: r_{ij}\)}. For~\eqref{eq:problem:dual-aux}, choosing \(\beta_j = \underset{i\in\mathcal{C},j\in\mathcal{S}}{\max}\: r_{ij}\), \(\boldsymbol\alpha = \mathbf{0}\), and \(\boldsymbol\gamma = \mathbf{0}\) yields a feasible solution with nonnegative objective value. These bounds allow us to apply the strong duality theorem (see Theorem~4.4 in~\citet{bertsimas-LPbook}), which ensures that \(\eqref{eq:problem:primal-aux} = \eqref{eq:problem:dual-aux}\).} and Equation~\eqref{eqn:key-ineq-2} we conclude that
\(
\eqref{eq:problem:primal-aux} = \eqref{eq:problem:dual-aux} \geq (1 - \delta) \cdot \eqref{eq:problem:primal}.
\)
This establishes that the optimal solution \((\widehat\lambda, \widehat x)\) obtained by solving~\eqref{eq:problem:primal-aux} is a feasible, and more importantly a \((1 - \delta)\)-approximate solution to~\eqref{eq:problem:primal}.

\noindent{\bf Part 2.} Here we analyze the runtime of the proposed Ellipsoid-based algorithm in Part 1. In Step~1, the Ellipsoid method, when paired with a polynomial-time separation oracle, operates in polynomial-time. Step~2 requires solving a linear program whose number of variables is polynomial in the input size, which can also be solved in polynomial-time. Therefore, assuming access to a polynomial-time separation oracle, the overall computational complexity of the Ellipsoid-based algorithm presented in Step 1 and Step 2 of Part 1 is polynomial in the input size.

Please note that the optimal solution to~\eqref{eq:problem:primal-aux} is a feasible solution to~\eqref{eq:problem:primal} that achieves at least a \((1 - \delta)\) fraction of its optimal value. Therefore, the Ellipsoid-based algorithm presented in this proof guarantees finding a \((1 - \delta)\)-approximate solution to~\eqref{eq:problem:primal} in polynomial time.


\begin{algorithm}[t]
\caption{The Ellipsoid Method for \eqref{eq:problem:dual}}\label{alg:ellipsoid-method}
\small{
{\bf Input:} Starting solution $(\boldsymbol\alpha_0, \boldsymbol\beta_0,\boldsymbol\gamma_0)$, starting matrix $\mathbf{D}_0$, maximum number of iterations $t_{\max}$, a $(1 - \delta)$-approximation algorithm $\mathcal{A}$ for Problem \eqref{eq:subdual} for any $j\in\mathcal{S}$ and some $\delta>0$. Here,  we  rewrite~\eqref{eq:problem:dual} in the form of {\color{blue}$\min\{\mathbf{d}^\top \mathbf{s} : \mathbf{A}\mathbf{s} \geq \mathbf{b}, \boldsymbol\alpha \geq 0\}$, where $\mathbf{s} = (\boldsymbol\alpha, \boldsymbol\beta,\boldsymbol\gamma)$}.
 The matrix $\mathbf{A}$ is chosen such that for its first $ m\times2^n$ rows $\mathbf{a}_{(j,C)}^\top \mathbf{s} = \beta_j +\sum_{i\in C}\gamma_{ij}$, and for the last $ nm$ rows $\mathbf{a}_{(i,j)}^\top \mathbf{s} =\frac{\alpha_{ij}}{u_{ij}} + \sum_{\ell \in \mathcal{S}} \alpha_{i\ell} -\gamma_{ij}$ where $\mathbf{a}_{(j,C)}$'s, and $\mathbf{a}_{(i,j)}$'s are the row vectors in $\mathbf{A}$ indexed by their corresponding pairs. Similarly, the vector $\mathbf{b}$ is chosen such that $b_{(j,C)} = R_j(C)$, and $b_{(i,j)}=0$. \\
{\bf Output:} (i) A collection $\mathcal{V}_j$ of sets that have violated~\eqref{constraint:dual-AC} constraints for each $j\in\mathcal{S}$. (ii) {\color{blue}The best solution $(\boldsymbol\alpha^\star, \boldsymbol\beta^\star,\boldsymbol\gamma^\star)$ accepted by the approximate separation oracle}, (iii) {\color{blue}its} objective $\textsc{obj}$.
\begin{enumerate}
  \item {\bf Initialization.} $(\boldsymbol\alpha, \boldsymbol\beta,\boldsymbol\gamma) = (\boldsymbol\alpha_0, \boldsymbol\beta_0,\boldsymbol\gamma_0), (\boldsymbol\alpha^\star, \boldsymbol\beta^\star,\boldsymbol\gamma^\star) = (\mathbf{0}_{nm}, \mathbf{1}_{m}, \mathbf{0}_{nm}),  \textsc{obj} = m, \mathbf{D} = \mathbf{D}_0, \mathcal{V}_j = \emptyset\ \text{for all } j\in \mathcal{S},$ and  $t = 0$.
  \item \label{step:ellipsoid-while} While $t \leq t_{\max}$:
  \begin{enumerate}
      \item {\bf Find a violated constraint.} 
      \begin{itemize}
          \item Check if we can reduce the objective further. If $\sum_{j \in \mathcal{S}} \beta_j + \sum_{(i,j) \in \mathcal{C} \times \mathcal{S}} \alpha_{ij}  \geq \textsc{obj}$, 
          set $\mathbf{a} = -\mathbf{d}$ and go to Step 2(b).    
          \item        Check if the last $nm$ nontrivial constraints hold. If $\frac{\alpha_{ij}}{u_{ij}} + \sum_{\ell \in \mathcal{S}} \alpha_{i\ell} -\gamma_{ij}  < 0$ for some $(i, j)$, set $\mathbf{a} = \mathbf{a}_{(i,j)}$. Go to Step 2(b).
          \item Check if the non-negativity constraints hold. If $\alpha_{ij} < 0$ for some $(i, j)$, set $\mathbf{a} = \mathbf{e}_{j\cdot n+i}$, here $\mathbf{e}_k$ is the unit vector where its $k^{th}$ coordinate is equal to one. Go to Step 2(b).
          \item Check if the~\eqref{constraint:dual-AC} constraints for any $j\in\mathcal{S}$, holds using the approximate separation oracle. \begin{itemize}
          \item Apply the $(1 - \delta)$-approximation algorithm $\mathcal{A}$ to Problem \eqref{eq:subdual} that returns $C_j^\mathcal{A}$ for any $j\in\mathcal{S}$ such that
          $ \revmod_j(C_j^\mathcal{A}, \boldsymbol\gamma)  \ge (1 - \delta) \cdot \subdual_j(\boldsymbol\gamma)$.
          \item
          If $\revmod_j(C_j^\mathcal{A}, \boldsymbol\gamma) > \beta_j$ for some $j\in\mathcal{S}$, then the set $C_j^\mathcal{A}$ violates the constraint for supplier $j$. Set $\mathbf{a} = \mathbf{a}_{(j,C_j^\mathcal{A})}$ and add $C_j^\mathcal{A}$ to $\mathcal{V}_j$, and go to Step 2(b)
          \end{itemize}
          \item If we have found no violated constraint, update our best {\color{blue}accepted solution} and its objective:
          $$
          (\boldsymbol\alpha^\star, \boldsymbol\beta^\star,\boldsymbol\gamma^\star) \leftarrow (\boldsymbol\alpha, \boldsymbol\beta,\boldsymbol\gamma) \quad \text{and} \quad 
          \textsc{obj} \leftarrow \sum_{j \in \mathcal{S}} \beta_j^\star + \sum_{(i,j) \in \mathcal{C} \times \mathcal{S}} \alpha_{ij}^\star\,.
          $$
          Then, go back to the start of Step~\ref{step:ellipsoid-while} and re-enter the while loop. 
      \end{itemize}
      \item {\bf Use the violated constraint to decrease the volume of the ellipsoid and find a new solution.}
      $$
      \begin{aligned}
      (\boldsymbol\alpha, \boldsymbol\beta,\boldsymbol\gamma) & \leftarrow (\boldsymbol\alpha, \boldsymbol\beta,\boldsymbol\gamma) + \frac{1}{(2nm +m)+1}\frac{\mathbf{D}\mathbf{a}}{\sqrt{\mathbf{a}^\top \mathbf{D}\mathbf{a}}} \,;\\
      \mathbf{D} & \leftarrow \frac{(2nm +m)^2}{(2nm +m)^2 -1} \Big(\mathbf{D} - \frac{2}{(2nm +m)+1}\frac{\mathbf{D}\mathbf{a}\mathbf{a}^\top\mathbf{D}}{\mathbf{a}^\top\mathbf{D}\mathbf{a}}\Big) \,.
      \end{aligned}
      $$
  \item $t \leftarrow t + 1\,.$
  \end{enumerate}
\item {\bf The ellipsoid is sufficiently small.} Return $\mathcal{V}_j$ for all $j\in\mathcal{S}$, $(\boldsymbol\alpha^\star, \boldsymbol\beta^\star,\boldsymbol\gamma^\star)$,  and $\textsc{obj}$. 
\end{enumerate}
}
\end{algorithm}

\noindent {\bf Details of the Ellipsoid Method.~\label{appendix:Ellipsoid-details}}
Here we present full details of the Ellipsoid method for~\eqref{eq:problem:dual}. The analysis follows the same one proposed by \citet{chen2025fairassortment} that leverages approximate oracles to efficiently solve a different linear program with exponentially many variables. We are including the details of the method in our context for completeness.

For simplicity of notation, we can rewrite~\eqref{eq:problem:dual} in the form of $\min\left\{d^\top s : As \ge b,\ \alpha \ge 0 \right\}$. Here, $\mathbf{s} = (\boldsymbol\alpha, \boldsymbol\beta,\boldsymbol\gamma)$. $\mathbf{d}$ is a vector of size $2nm +m$ chosen such that $\mathbf{d}^\top \mathbf{s} =  \sum_{j \in \mathcal{S}} \beta_j + \sum_{(i,j) \in \mathcal{C} \times \mathcal{S}} \alpha_{ij}$. $\mathbf{A}$ is a $N \times (2nm +m)$ matrix and $\mathbf{b}$ is a vector of size $N$, where $N = m\times2^n+nm$, i.e., number of non trivial constraints of the dual. 
Let the first $ m\times2^n$ rows of $\mathbf{A}$ and $\mathbf{b}$ be indexed by a pairs $(j,C)$ with $j\in\mathcal{S}, C\subseteq \mathcal{C}$, and the last $nm$ rows of $\mathbf{A}$ and $\mathbf{b}$ be indexed by a pairs $(i,j)\in \mathcal{C}\times\mathcal{S}$.
The matrix $\mathbf{A}$ is chosen such that for its first $ m\times2^n$ rows $\mathbf{a}_{(j,C)}^\top \mathbf{s} = \beta_j +\sum_{i\in C}\gamma_{ij}$, and for the last $ nm$ rows $\mathbf{a}_{(i,j)}^\top \mathbf{s} =\frac{\alpha_{ij}}{u_{ij}} + \sum_{\ell \in \mathcal{S}} \alpha_{i\ell} -\gamma_{ij}$ where $\mathbf{a}_{(j,C)}$'s, and $\mathbf{a}_{(i,j)}$'s are the row vectors in $\mathbf{A}$ indexed by their corresponding pairs. Similarly, the vector $\mathbf{b}$ is chosen such that $b_{(j,C)} = R_j(C)$, and $b_{(i,j)}=0$. 

Within the Ellipsoid method, we keep track of the following quantities: 
(i) $(\boldsymbol\alpha, \boldsymbol\beta,\boldsymbol\gamma) \in \mathbb{R}^{2nm +m}$, the center of the current ellipsoid, which is also the current solution to the dual problem; note that this solution might not be feasible for the dual problem. 
(ii) $(\boldsymbol\alpha^\star, \boldsymbol\beta^\star,\boldsymbol\gamma^\star)$, the best solution {\color{blue}accepted by the approximate separation oracle} we have found so far. We initialize $(\boldsymbol\alpha^\star, \boldsymbol\beta^\star,\boldsymbol\gamma^\star)$ to be $(\mathbf{0}_{nm}, \mathbf{1}_{m}, \mathbf{0}_{nm})$, where $\mathbf{0}_{nm}$ is a zero vector of length $nm$, and $\mathbf{1}_{m}$ is a vector of all ones with length $m$; this is always a feasible solution to~\eqref{eq:problem:dual}\footnote{Assume we have applied an initial normalization on revenues such that $R_j(C)\leq 1$ for all  $j\in\mathcal{S}, C\subseteq \mathcal{C}$.}.  
(iii) The current best objective $\textsc{obj}$. We initialize it to be $m$, which is the objective of $(\mathbf{0}_{nm}, \mathbf{1}_{m}, \mathbf{0}_{nm})$. 
(iv) A positive-definite matrix $\mathbf{D} \in \mathbb{R}^{(2nm +m) \times (2nm +m)}$, which represents the shape of the ellipsoid.
(v) A collection $\mathcal{V}_j$ of subsets of customers that have violated the~\eqref{constraint:dual-AC} constraint during the execution of the Ellipsoid method for each $j\in\mathcal{S}$. 

Full details of the Ellipsoid method are presented in Algorithm~\ref{alg:ellipsoid-method}. At a high level, the Ellipsoid method for~ \eqref{eq:problem:dual} works as follows. At each iteration, it generates an ellipsoid $E$ centered at the current solution $\mathbf{s} = (\boldsymbol\alpha, \boldsymbol\beta,\boldsymbol\gamma)$, which is defined as:
$$
E = \{\mathbf{x} : (\mathbf{x} - \mathbf{s})^\top \mathbf{D}^{-1} (\mathbf{x} - \mathbf{s}) \leq 1\} \,.
$$
By design of the Ellipsoid method, the ellipsoid $E$ always contains the intersection of the feasibility region of~\eqref{eq:problem:dual} and the half space $\{\mathbf{s} : \mathbf{d}^\top \mathbf{s} < \textsc{obj}\}$.

The Ellipsoid method for~ \eqref{eq:problem:dual} terminates when the ellipsoid $E$ is sufficiently small, by selecting a sufficiently large $t_{\max}$. Recall that the ellipsoid $E$ always contains the intersection of the feasibility region of~\eqref{eq:problem:dual} and the half space $\{\mathbf{s} : \mathbf{d}^\top \mathbf{s} < \textsc{obj}\}$. Intuitively, when the ellipsoid gets reduced to a sufficiently small volume, it is unlikely that there exists a feasible solution that can further reduce our objective. As shown in \citet{bertsimas-LPbook}, given a separation oracle, the Ellipsoid method is guaranteed to terminate in $t_{\max} = \mathcal{O}((nm)^{6}\log(nmU))$ iterations, where $U$ is the bit complexity of~\eqref{eq:problem:dual}. Hence, the Ellipsoid method terminates in some numbers of iterations that is polynomial in input size. Moreover, when the Ellipsoid method terminates, it returns the following: (i) for any $j\in\mathcal{S}$ the collection $\mathcal{V}_j$ of subsets of customers that have violated the~\eqref{constraint:dual-AC} constraint; (ii) {\color{blue}the best solution $(\boldsymbol\alpha^\star, \boldsymbol\beta^\star,\boldsymbol\gamma^\star)$ accepted by the approximate separation oracle} and (iii) {\color{blue}its} objective $\textsc{obj}$.



\section{{\color{blue}Formal Comparison of \textsc{(LP--II)} and the LP of \citet{danny2024}}}
\label{appendix:nissim-comparison}

{\color{blue} In this appendix, we compare \textsc{(LP--II)} with the LP relaxation introduced by \citet{danny2024} and used in their algorithm for their Customized Model, a static one-sided setting in which the platform can personalize the set of selecting customers shown to each supplier. We denote their relaxation by \textsc{(LP-Nissim)}, and formally establish that it is a further relaxation of \textsc{(LP--II)}. We first recall both formulations under the MNL choice model, then construct an explicit feasibility-preserving projection from \textsc{(LP--II)} to \textsc{(LP-Nissim)}, and give an instance on which their optimal values differ.

\vspace{2mm}
\noindent\textbf{\textsc{(LP-Nissim)}.} Let $\hat{w}_{ji} = \min\{1, w_{ji}\}$ for all $(i,j) \in \mathcal{C} \times \mathcal{S}$, and define the marginal MNL polyhedra, whose constraints are again imposed only for the pairs with positive weight:
\begin{align*}
P^\mathcal{S} &= \left\{ y \in \mathbb{R}_+^{|\mathcal{C}|\times|\mathcal{S}|} : \frac{y_{ij}}{w_{ji}} + \sum_{k \in \mathcal{C}} y_{kj} \leq 1 \quad \forall (i,j) \in \mathcal{C} \times \mathcal{S} \right\},\\
P^\mathcal{C} &= \left\{ x \in \mathbb{R}_+^{|\mathcal{C}|\times|\mathcal{S}|} : \frac{x_{ij}}{u_{ij}} + \sum_{\ell \in \mathcal{S}} x_{i\ell} \leq 1 \quad \forall (i,j) \in \mathcal{C} \times \mathcal{S} \right\}.
\end{align*}
The LP of \citet{danny2024} is given by:
\begin{align*}
\max &\quad \sum_{(i, j) \in \mathcal{C}\times\mathcal{S}} r_{ij} \cdot y_{ij}\\
\text{s.t.} &\quad y_{ij} = \hat{w}_{ji}\, x_{ij}, \quad \forall (i,j) \in \mathcal{C} \times \mathcal{S},\\
&\quad x \in P^\mathcal{C}, \quad y \in P^\mathcal{S}.
\end{align*}

\vspace{2mm}
\noindent\textbf{Projection from \textsc{(LP--II)} to \textsc{(LP-Nissim)}.}
Recall that $R_j(C) = \sum_{i \in C} r_{ij} \frac{w_{ji}}{1 + \sum_{k \in C} w_{jk}}$ under MNL. For each $C \subseteq \mathcal{C}$ and $j \in \mathcal{S}$, define the vector $y^C \in \mathbb{R}_+^{|\mathcal{C}|\times|\mathcal{S}|}$ by
\[
y_{ij}^C = \frac{w_{ji}}{1 + \sum_{k \in C} w_{jk}} \cdot \mathds{1}_{\{i \in C\}},
\]
so that $R_j(C) = \sum_{i \in C} r_{ij} y_{ij}^C$. By construction, $y^C \in P^\mathcal{S}$ and $y_{ij}^C \leq \hat{w}_{ji}$ for all $i \in C$.

Given any feasible solution $(\lambda, x)$ to \textsc{(LP--II)}, define the projection:
\[
\bar{y}_{ij} = \sum_{C \subseteq \mathcal{C} :\, i \in C} \lambda_{j,C} \cdot y^C_{ij},
\qquad
\bar{x}_{ij} = \frac{\bar{y}_{ij}}{\hat{w}_{ji}}, \quad \forall (i,j) \in \mathcal{C} \times \mathcal{S},
\]
with $\bar{x}_{ij}=0$ whenever $\hat{w}_{ji}=0$. The equality $\bar{y}_{ij} = \hat{w}_{ji}\bar{x}_{ij}$ holds by the definition of $\bar{x}$, so it remains to show that $\bar{y} \in P^\mathcal{S}$ and $\bar{x} \in P^\mathcal{C}$:
\begin{enumerate}
    \item \textbf{$\bar{y} \in P^\mathcal{S}$:} Since $\sum_{C} \lambda_{j,C} = 1$ and $\lambda_{j,C} \geq 0$, $\bar{y}_{j,\cdot}$ is a convex combination of $\{y^C\}_{C \subseteq \mathcal{C}}$. Each $y^C \in P^\mathcal{S}$ and $P^\mathcal{S}$ is a  polyhedron, so $\bar{y} \in P^\mathcal{S}$.
    \item \textbf{$\bar{x} \in P^\mathcal{C}$:} We have
    \[
    \bar{y}_{ij} = \sum_{C \subseteq \mathcal{C}:\, i \in C} \lambda_{j,C}\, y^C_{ij} \leq \hat{w}_{ji} \sum_{C\subseteq \mathcal{C} :\, i \in C} \lambda_{j,C} = \hat{w}_{ji}\, x_{ij},
    \]
    where the inequality uses $y^C_{ij} \leq \hat{w}_{ji}$ for $i \in C$, and the last equality uses the second family of constraints of \textsc{(LP--II)}. Hence $0 \le \bar{x}_{ij} \le x_{ij}$. The third family of constraints of \textsc{(LP--II)} exactly defines $P^\mathcal{C}$, so $x \in P^\mathcal{C}$. Since the left-hand side of every constraint defining $P^\mathcal{C}$ is nondecreasing in $x$, it follows that $\bar{x} \in P^\mathcal{C}$.
\end{enumerate}
Moreover, the objective values coincide:
\[
\sum_{j\in \S} \sum_{C\subseteq \mathcal{C}} R_j(C)\, \lambda_{j,C} = \sum_{j\in\S} \sum_{C\subseteq \mathcal{C}} \lambda_{j,C} \sum_{i \in C} r_{ij} y^C_{ij} = \sum_{(i,j)\in \C\times \S} r_{ij} \sum_{C \subseteq \mathcal{C}:\, i \in C} \lambda_{j,C}\, y^C_{ij} = \sum_{(i,j)\in \C \times \S} r_{ij}\, \bar{y}_{ij}.
\]
Hence, every feasible solution to \textsc{(LP--II)} maps to a feasible solution to \textsc{(LP-Nissim)} with the same objective value, establishing that \textsc{(LP-Nissim)} is a further relaxation of \textsc{(LP--II)}.

The containment is strict. Consider two customers and a single supplier with $u_{i1}=w_{1i}=r_{i1}=1$ for $i\in \{1,2\}$, so that $R_1(\{1\})=R_1(\{2\})=\tfrac12$ and $R_1(\{1,2\})=\tfrac23$. In \textsc{(LP--II)}, the third family of constraints gives $x_{i1}\le\tfrac12$, and the second family gives $\lambda_{1,\{1\}}+\lambda_{1,\{2\}}=x_{11}+x_{21}-2\lambda_{1,\{1,2\}}$, so the objective equals $\tfrac12(x_{11}+x_{21})-\tfrac13\lambda_{1,\{1,2\}}\le\tfrac12$, attained at $\lambda_{1,\{1\}}=\lambda_{1,\{2\}}=\tfrac12$. In \textsc{(LP-Nissim)}, adding the two constraints of $P^\mathcal{S}$ gives $y_{11}+y_{21}\le\tfrac23$, attained at $x_{11}=x_{21}=y_{11}=y_{21}=\tfrac13$. The two optimal values are $\tfrac12$ and $\tfrac23$, so \textsc{(LP-Nissim)} is a strict relaxation of \textsc{(LP--II)}. Working against this tighter benchmark is one of the sources of our improvement over the guarantee of \citet{danny2024} in the static setting, together with our different techniques and analysis.}


\section{{\color{blue}Tightness of the Correlation-Gap Bound and the Static Algorithm}}
\label{appendix:lp-tightness}

{\color{blue}
We show that the $\tfrac12$ guarantee of Theorem~\ref{thm:general_revenues} is tight with respect to the relaxation~\eqref{eq:problem:primal}: there is a family $\{\mathcal{J}_N\}_{N\ge2}$ of \atar\ instances for which $\lim_{N\to\infty}\mathsf{ALG}_{\sf static}/\eqref{eq:problem:primal}=\tfrac12$. The construction uses submodular-order MNL revenue functions whose correlation gap attains the value $2$ in the limit, the largest value allowed by Lemma~\ref{lemma:MNL 1/2 coerrelation gap}.

\vspace{2mm}
\noindent\textbf{The instance.}
Fix an integer $N\ge2$ and let $\epsilon=1/N^2$. The instance has a single supplier, $\mathcal{S}=\{1\}$, and $N+1$ customers $i_1,\dots,i_{N+1}$. The MNL and revenue parameters are
\begin{align*}
&u_{i_t,1}=\frac{1}{\sqrt N-1},\qquad w_{1,i_t}=\epsilon,\qquad r_{i_t,1}=\frac{1}{\epsilon\sqrt N}\qquad(1 \le t\le N),\\
&u_{i_{N+1},1}=N,\qquad w_{1,i_{N+1}}=N,\qquad r_{i_{N+1},1}=1.
\end{align*}
We call $i_1,\dots,i_N$ the \emph{light} customers and $i_{N+1}$ the \emph{heavy} customer, after their supplier preference weights, and for a backlog $C$, we denote $K(C)$  the number of light customers it contains. Since there is only one supplier, we drop its index from the revenue functions. All light customers share the same parameters, and the revenue of a backlog is increasing in the number of light customers offered, so for a backlog made of the heavy customer together with $K$ light customers the optimal revenue is attained by offering the heavy customer, the $K$ light customers, or both. Dropping the heavy customer leaves only the second of these three options, so for every backlog $C$, writing $K:=K(C)$,
\[
R^*(C)\ \le\ \max\Big\{\frac{N}{1+N},\ \frac{K/\sqrt N}{1+\epsilon K},\ \frac{N+K/\sqrt N}{1+N+\epsilon K}\Big\}\ \le\ \max\Big\{1,\ \frac{K}{\sqrt N}\Big\}.
\]
For the second inequality, the term $\tfrac{N}{1+N}$ is at most $1$ and the term $\tfrac{K/\sqrt N}{1+\epsilon K}$ is at most $\tfrac{K}{\sqrt N}$, so only the term $\tfrac{N+K/\sqrt N}{1+N+\epsilon K}$ needs an argument. When $K\le\sqrt N$, its numerator is at most $N+1$ and its denominator at least $1+N$, so it is at most $1$. When $K>\sqrt N$, the required inequality $\tfrac{N+K/\sqrt N}{1+N+\epsilon K}\le\tfrac{K}{\sqrt N}$ rearranges to $N\le K\sqrt N+\epsilon K^2/\sqrt N$, which holds since $K\sqrt N>N$.

\vspace{2mm}
\noindent\textbf{The LP value.}
Since there is a single supplier, the third family of constraints of~\eqref{eq:problem:primal} reads $\tfrac{x_{i,1}}{u_{i,1}}+x_{i,1}\le1$ for every customer $i$, so $x_{i,1}\le\tfrac{u_{i,1}}{1+u_{i,1}}$. For a light customer this gives $x_{i_t,1}\le\tfrac{1/(\sqrt N-1)}{1+1/(\sqrt N-1)}=\tfrac{1}{\sqrt N}$, and for the heavy customer $x_{i_{N+1},1}\le\tfrac{N}{1+N}$. Consider the solution that places the set $\{i_1,\dots,i_N\}$ with probability $\tfrac{1}{\sqrt N}$ and the singleton $\{i_{N+1}\}$ with probability $1-\tfrac{1}{\sqrt N}$, with $x$ given by the induced marginals $x_{i_t,1}=\tfrac{1}{\sqrt N}$ for $t\le N$ and $x_{i_{N+1},1}=1-\tfrac{1}{\sqrt N}$. The two probabilities sum to $1$, and the marginals meet the two caps above, the light customers' with equality and the heavy customer's since $\sqrt N\le1+N$, so this solution is feasible for~\eqref{eq:problem:primal}. Its objective value is
\[
\Big(1-\tfrac{1}{\sqrt N}\Big)\cdot\frac{N}{1+N}+\tfrac{1}{\sqrt N}\cdot\frac{\sqrt N}{1+1/N}\ =\ 2-O\!\Big(\tfrac{1}{\sqrt N}\Big),
\]
where $R(\{i_{N+1}\})=\tfrac{N}{1+N}$, and where each light customer contributes $r_{i_t,1}w_{1,i_t}=\tfrac{1}{\sqrt N}$ to the numerator and $w_{1,i_t}=\epsilon$ to the denominator, so that $R(\{i_1,\dots,i_N\})=\tfrac{N/\sqrt N}{1+N\epsilon}=\tfrac{\sqrt N}{1+1/N}$. Hence $\eqref{eq:problem:primal}\ge 2-O(1/\sqrt N)$.

\vspace{2mm}
\noindent\textbf{The algorithm's value.}
Algorithm~\ref{alg:Randomized Static 1/2} samples the assortment of each customer independently, so the customers enter the supplier's backlog independently, customer $i$ with probability $x_{i,1}$. These marginals come from an approximate solution to~\eqref{eq:problem:primal} and satisfy the caps above, so $x_{i_t,1}\le\tfrac{1}{\sqrt N}$ for $t\le N$. The number $K$ of light customers in the realized backlog is then a sum of independent indicators, so
\[
\mu:=\mathbb{E}[K]=\sum_{t\le N}x_{i_t,1}\le\sqrt N,
\qquad
\mathrm{Var}(K)=\sum_{t\le N}x_{i_t,1}\bigl(1-x_{i_t,1}\bigr)\le\mu\le\sqrt N.
\]
Applying the bound $R^*(C)\le\max\{1,\tfrac{K(C)}{\sqrt N}\}$ to every realized backlog,
\[
\mathsf{ALG}_{\sf static}\ \le\ \mathbb{E}\Big[\max\Big\{1,\tfrac{K}{\sqrt N}\Big\}\Big]
\ =\ 1+\frac{\mathbb{E}\bigl[(K-\sqrt N)^+\bigr]}{\sqrt N}
\ \le\ 1+\frac{\mathbb{E}\bigl|K-\mu\bigr|}{\sqrt N}
\ \le\ 1+\frac{\sqrt{\mathrm{Var}(K)}}{\sqrt N}
\ \le\ 1+N^{-1/4},
\]
where the equality uses $\max\{1,a\}=1+(a-1)^+$; the first inequality holds since $\mu\le\sqrt N$, so that $(K-\sqrt N)^+\le|K-\mu|$; the second is Cauchy--Schwarz; and the last uses $\mathrm{Var}(K)\le\sqrt N$.

\vspace{2mm}
\noindent\textbf{Conclusion.}
Combining $\mathsf{ALG}_{\sf static}\le1+N^{-1/4}$ with $\eqref{eq:problem:primal}\ge2-O(1/\sqrt N)$ gives $\mathsf{ALG}_{\sf static}/\eqref{eq:problem:primal}\le\tfrac{1+N^{-1/4}}{2-O(1/\sqrt N)}$, while Theorem~\ref{thm:general_revenues} gives the matching lower bound $\mathsf{ALG}_{\sf static}\ge(\tfrac12-o(1))\,\eqref{eq:problem:primal}$. Therefore $\lim_{N\to\infty}\mathsf{ALG}_{\sf static}/\eqref{eq:problem:primal}=\tfrac12$.

\vspace{2mm}
\noindent\textbf{Tightness of the correlation-gap bound.}
The same construction shows that the bound of Lemma~\ref{lemma:MNL 1/2 coerrelation gap} is tight for general revenues. Let $D$ denote the correlated distribution over backlogs used in the LP construction above: it assigns probability $\tfrac{1}{\sqrt N}$ to the set $\{i_1,\dots,i_N\}$ and probability $1-\tfrac{1}{\sqrt N}$ to the singleton $\{i_{N+1}\}$. Let $D^{\mathrm{ind}}$ be the independent distribution with the same marginals as $D$; this is exactly the backlog distribution induced by Algorithm~\ref{alg:Randomized Static 1/2} when run on this solution. Since both sets in the support of $D$ are active, $R$ and $R^*$ coincide on them, and the LP computation above gives
$
\mathbb{E}_{A\sim D}[R^*(A)] = 2-O(1/\sqrt N).
$
On the other hand, the bound on the algorithm's value gives
$
\mathbb{E}_{A\sim D^{\mathrm{ind}}}[R^*(A)] \le 1+N^{-1/4}.
$
Therefore,
\[
\frac{\mathbb{E}_{A\sim D}[R^*(A)]}
{\mathbb{E}_{A\sim D^{\mathrm{ind}}}[R^*(A)]}
\to 2
\qquad\text{as }N\to\infty,
\]
showing that the correlation-gap bound of $2$ is tight.

\vspace{2mm}
\noindent\textbf{Tightness of the analysis of Algorithm~\ref{alg:Randomized Static 1/2} and the correlation-gap bound under uniform revenues.}
The analysis of Algorithm~\ref{alg:Randomized Static 1/2} is likewise tight against~\eqref{eq:problem:primal} under uniform revenues, using the instance from the proof of Proposition~3.1 in \citet{housni2024twosided}. The instance has one supplier and $n$ customers, with all revenues equal to $1$, $u_{i,1}=\tfrac{1}{n-1}$, and $w_{1,i}=w$ for all customers $i$, where $w$ is large. The third family of constraints gives $x_{i,1}\le\tfrac1n$, and assigning probability $\tfrac1n$ to each singleton backlog is feasible with objective value $\tfrac{w}{1+w}$. Hence, the value of~\eqref{eq:problem:primal} approaches $1$ as $w\to\infty$.
On the other hand, each customer selects the supplier with probability at most $\tfrac1n$, and this probability is attained by offering the supplier. Thus, the probability that the supplier's backlog is nonempty is
$
1-\left(1-\frac1n\right)^n.
$
Conditional on a nonempty backlog, the supplier selects some customer with probability approaching $1$ as $w\to\infty$. Therefore,
$
\mathsf{OPT}\to 1-\left(1-\frac1n\right)^n,
$
which converges to $1-\tfrac1e$ as $n\to\infty$. Hence, no analysis of Algorithm~\ref{alg:Randomized Static 1/2} against~\eqref{eq:problem:primal} can yield a guarantee better than $1-\tfrac1e$ under uniform revenues. 

Similar to the general-revenue analysis above, this instance shows that the correlation-gap bound of \(\tfrac{e}{e-1}\) is tight under uniform revenues.
}

\providecommand{\hub}{\mathsf{h}}
\providecommand{\alt}{\mathsf{a}}
\providecommand{\dec}{\mathsf{d}}

\section{{\color{blue}Hard Instances for the Greedy Algorithm}}\label{appendix:greedy-hard}

{\color{blue}We present two instances that show the limitations of the greedy algorithm: the first shows that the $\tfrac12$-approximation is tight under same-order revenues, and the second shows an unbounded optimality gap under general revenues.}

\providecommand{\hub}{\mathsf{h}}
\providecommand{\alt}{\mathsf{a}}
\providecommand{\dec}{\mathsf{d}}

\subsection{{\color{blue}Tightness of the Greedy Algorithm for Same-Order Revenue Structure}}
\label{appendix:greedy-tight}

{\color{blue}We show that the $\tfrac12$-approximation guarantee of Algorithm~\ref{alg:Submodular_order_greedy} for same-order revenues (Theorem~\ref{thm:separable_revenues}) is tight: there is a family of same-order instances on which the greedy algorithm obtains only half of the optimal expected revenue.

\noindent\textbf{Example~1 (Tightness of the $1/2$-Approximation for Same-Order Revenues).} Consider an instance with $n=m=2$, revenues $r_{11}=r_{12}=1$, $r_{21}=0$, and $r_{22}=1-\varepsilon$, where $0<\varepsilon<1$. Let the customer preference weights be $u_{i1}=L$ and $u_{i2}=L^2$, and let the supplier preference weights be $w_{ji}=L^2$ for all $i,j\in\{1,2\}$. Both suppliers rank $i_1$ ahead of $i_2$ in revenue, so the instance satisfies the same-order condition and the greedy algorithm processes $i_1$ first. Initially, adding $i_1$ to either supplier's empty backlog increases its optimal expected revenue by $H_L:=L^2/(1+L^2)$. Hence, the greedy objective for an assortment $S$ is $H_L\sum_{j\in S}u_{1j}/(1+\sum_{j\in S}u_{1j})$, which is uniquely maximized by offering both suppliers. Customer $i_1$ therefore selects $j_2$ with probability $p_L:=L^2/(1+L+L^2)$, which tends to $1$ as $L\to\infty$. Conditional on this event, customer $i_2$ has zero marginal value at $j_1$ because $r_{21}=0$. At $j_2$, offering both customers yields expected revenue $(2-\varepsilon)L^2/(1+2L^2)$, which is smaller than $H_L$ whenever $L^2>(1-\varepsilon)/\varepsilon$; offering only $i_2$ yields $(1-\varepsilon)H_L<H_L$. Thus, for all sufficiently large $L$, the optimal assortment for supplier $j_2$ remains $\{i_1\}$ even if $i_2$ joins its backlog, implying $R^*_{j_2}(\{i_1,i_2\})=R^*_{j_2}(\{i_1\})=H_L$. Consequently, $i_2$ has zero marginal value at both suppliers, and the conditional expected total revenue is $H_L$, regardless of the assortment offered to $i_2$. Since the complementary event has probability $1-p_L=O(1/L)$ and total revenue is always between $0$ and $2$, the greedy algorithm's expected revenue satisfies $p_LH_L\leq \mathrm{ALG}_{\mathrm{greedy}}\leq p_LH_L+2(1-p_L)$ and therefore converges to $1$. In contrast, the feasible policy that offers $\{j_1\}$ to $i_1$ and $\{j_2\}$ to $i_2$ earns expected revenue $\frac{L}{1+L}H_L+(1-\varepsilon)H_L^2$, which converges to $2-\varepsilon$. No policy can earn more than $2-\varepsilon$, because each customer can be matched at most once and the maximum revenue obtainable from $i_1$ and $i_2$ is $1$ and $1-\varepsilon$, respectively. Hence, $\mathrm{OPT}\to2-\varepsilon$ as $L\to\infty$, and $\mathrm{ALG}_{\mathrm{greedy}}/\mathrm{OPT}\to1/(2-\varepsilon)$ for each fixed $\varepsilon$. Letting $\varepsilon\to 0$ makes this ratio arbitrarily close to $1/2$, proving that the greedy algorithm's $1/2$-approximation guarantee under same-order revenues is tight. This instance is illustrated in Figure~\ref{fig:hard-instances}.

}

\begin{figure}[H]
\centering
\begin{tikzpicture}[
  font=\small,
  cust/.style={draw, rounded corners, fill=black!5, minimum width=1.6cm, minimum height=1.05cm, align=center},
  supp/.style={draw, circle, fill=black!5, minimum size=0.95cm},
  elab/.style={fill=white, inner sep=1.6pt, font=\scriptsize},
  ar/.style={semithick, black!55}
]
\node[cust] (c1a) at (0, 1.2)  {$i_1$};
\node[cust] (c2a) at (0, -1.2) {$i_2$};
\node[supp] (s1a) at (6, 1.2)  {$j_1$};
\node[supp] (s2a) at (6, -1.2) {$j_2$};
\draw[ar] (c1a) -- node[elab, above]
  {$r_{11}{=}1,\ u_{11}{=}L,\ w_{11}{=}L^2$} (s1a);
\draw[ar] (c2a) -- node[elab, below]
  {$r_{22}{=}1{-}\varepsilon,\ u_{22}{=}L^2,\ w_{22}{=}L^2$} (s2a);
\draw[ar] (c2a) -- node[elab, pos=0.30, below right]
  {$r_{21}{=}0,\ u_{21}{=}L,\ w_{12}{=}L^2$} (s1a);
\draw[ar] (c1a) -- node[elab, pos=0.30, above right]
  {$r_{12}{=}1,\ u_{12}{=}L^2,\ w_{21}{=}L^2$} (s2a);
\end{tikzpicture}
\caption{{\color{blue}A tight instance for the greedy algorithm under same-order revenues. Edge labels give revenues $r_{ij}$ and preference weights $u_{ij}$ and $w_{ji}$.}}
\label{fig:hard-instances}
\end{figure}

\subsection{{\color{blue}Unbounded Gap of the Greedy Algorithm for General Revenues}}
\label{appendix:hard-instance}

{\color{blue}
Without the same-order revenue condition, Algorithm~\ref{alg:Submodular_order_greedy} loses its
$\tfrac12$-approximation guarantee. In this section we show that the loss can be made \emph{arbitrarily
large}: we present a family of instances $\{\mathcal{I}_n\}_{n\ge2}$, indexed by the number of customers $n$,
on which the greedy algorithm, run under an adversarially chosen processing order, obtains only a $\Theta(1/n)$
fraction of the optimum. For general revenues the greedy algorithm therefore admits no constant-factor
guarantee, in contrast to Algorithm~\ref{alg:Randomized Static 1/2}, which retains its
$(\tfrac12-\epsilon)$-guarantee for arbitrary pairwise revenues. The same construction can also be adapted so
that the heuristic order used by our greedy algorithm in the numerical experiments of
Section~\ref{sec:numerics} (ranking customers by descending average pairwise revenue) coincides with this
adversarial order. Indeed, augment $\mathcal{I}_n$ with an inactive supplier whose MNL preference weights are zero
($w=u=0$) and whose pair revenues are chosen so that the congesting customer has the largest average pairwise
revenue. Since its preference weights are zero, this supplier is never selected and contributes nothing to the
greedy algorithm or to the optimum, so ${\sf ALG}_{\sf greedy}$ and ${\sf OPT}$ are unchanged; but the
heuristic now processes the congesting customer first, which is all the argument below uses, the light
customers being exchangeable. The failure of the greedy algorithm is therefore not specific to the worst-case
order.}

\vspace{2mm}
\noindent{\color{blue}\textbf{Construction.}
Fix an integer $n\ge2$. Instance $\mathcal{I}_n$ has $n$ customers $\mathcal{C}=\{i_1,\dots,i_n\}$
and three suppliers $\mathcal{S}=\{\hub,\alt,\dec\}$, which we call the \emph{hub}, the \emph{alternate}, and
the \emph{decoy}. We refer to $i_1$ as the \emph{congesting customer} and to the remaining $n-1$ customers
$i_2,\dots,i_n$ as the \emph{light customers}, writing $i_t$ with $t\in\{2,\dots,n\}$ for a generic light customer, and we process them in the adversarial order
$i_1,i_2,\dots,i_n$, so that the congesting customer is processed first.
The instance is governed by fixed positive constants $\rho,a,b,V,p,\kappa$ and a scale parameter $\theta>1$; every
parameter not shown in Figure~\ref{fig:appendix-hard-instance} is assumed to be zero. The figure also shows the
inactive supplier $\mathsf{x}$ of the heuristic-order variant above, which is not part of $\mathcal{I}_n$. For the heuristic-order variant, set $\Lambda=\theta+p+1$; all
parameters on other edges incident to $\mathsf{x}$ remain zero.
Then $\rho+a+\Lambda>\theta+p$, so the congesting customer has
strictly larger average pairwise revenue than every light customer.
We assume throughout that $\rho>V$ and $\rho>\tfrac{ab}{1+b}$.
We also denote the three single-customer optimal revenues used below:
\[
H_\theta:=R^*_\hub(\{i_1\})=\frac{\rho \theta}{1+\theta},\qquad
A:=R^*_\alt(\{i_1\})=\frac{ab}{1+b},\qquad
D:=R^*_\dec(\{i_t\})=\frac{p\kappa}{1+\kappa},
\]
and note that $H_\theta$ approaches $\rho$ as $\theta$ grows.}

\begin{figure}[H]
\centering
\begin{tikzpicture}[
  font=\small,
  cust/.style={draw, rounded corners, fill=black!5, minimum width=2.3cm, minimum height=1.05cm, align=center},
  supp/.style={draw, circle, minimum size=0.95cm},
  role/.style={font=\scriptsize\itshape, anchor=west},
  elab/.style={fill=white, inner sep=1.6pt, font=\scriptsize},
  ar/.style={semithick, black!55}
]
\node[cust] (cong) at (0,1.8) {congesting\\ customer $i_1$};
\node[cust, fill=black!5] at (0.18,-1.78) {};
\node[cust, fill=black!5] at (0.09,-1.69) {};
\node[cust] (light) at (0,-1.6) {light customers\\ $i_2,\dots,i_n$};
\node[supp, fill=green!12]  (a) at (6,1.9)  {$\alt$};
\node[role] at (a.east) {\,alternate};
\node[supp, fill=red!15]    (h) at (6,0.1)  {$\hub$};
\node[role] at (h.east) {\,hub};
\node[supp, fill=orange!22] (d) at (6,-1.7) {$\dec$};
\node[role] at (d.east) {\,decoy};
\node[supp, fill=blue!10]   (x) at (6,3.7) {$\mathsf{x}$};
\node[role] at (x.east) {\,inactive};
\draw[ar] (cong)  -- node[elab, pos=0.55]{$r{=}a,\ u{=}\theta,\ w{=}b$}        (a);
\draw[ar] (cong)  -- node[elab, pos=0.42]{$r{=}\rho,\ u{=}\theta,\ w{=}\theta$}     (h);
\draw[ar] (light) -- node[elab, pos=0.42]{$r{=}\theta,\ u{=}\theta,\ w{=}V/\theta$}      (h);
\draw[ar] (light) -- node[elab, pos=0.55]{$r{=}p,\ u{=}\theta^{-1/2},\ w{=}\kappa$} (d);
\draw[ar] (cong)  -- node[elab, pos=0.5]{$r{=}\Lambda,\ u{=}0,\ w{=}0$} (x);
\end{tikzpicture}
\caption{{\color{blue}Hard instance for general revenues, with
$\mathsf{x}$ included only in the heuristic-order variant.
Edge labels give $r=r_{ij}$, $u=u_{ij}$, and $w=w_{ji}$.}}
\label{fig:appendix-hard-instance}
\end{figure}

{\color{blue}
The mechanism is as follows. Under the adversarial processing order, the greedy algorithm first processes the congesting customer. Since its hub value $H_\theta$ exceeds its alternate value $A=\tfrac{ab}{1+b}$ for large $\theta$, the greedy algorithm offers the hub to the congesting customer, which it selects with probability $\tfrac{\theta}{1+\theta}$. This selection incorporates the customer's large weight $\theta$ into the denominator of the hub's revenue function, $R_\hub(C)=\frac{\sum_{i\in C}r_{i\hub}w_{\hub i}}{1+\sum_{i\in C}w_{\hub i}}$, for all subsequent marginal computations. From that point on, every light customer's
\emph{marginal} value to the hub has collapsed to nearly zero. Each light customer is then offered its
decoy, which presents a high marginal value to the algorithm but a negligible choice probability for the customer, so the customer takes the outside option with probability close to one. A feasible comparison policy offers $\{\hub\}$ to every light
customer and $\{\alt\}$ to the congesting customer. For each fixed
$n$, its expected revenue tends to $(n-1)V+A$ as $\theta\to\infty$,
whereas the greedy algorithm's expected revenue tends to $\rho$.}

\begin{lemma}[{\color{blue}Marginal collapse on a congested hub}]\label{lem:marginal-collapse}
{\color{blue}For every light customer $i_t$ and all sufficiently large $\theta$,
\[
R^*_\hub\bigl(i_t\mid\{i_1\}\bigr)=R^*_\hub(\{i_1,i_t\})-R^*_\hub(\{i_1\})=\frac{V}{\theta}+O(\theta^{-2}),
\]
which vanishes as $\theta$ grows.}
\end{lemma}
\begin{myproof}
{\color{blue}The marginal value involves only the congesting customer $i_1$ and the light customer $i_t$. A direct computation of the relevant revenue terms gives
\[
R_\hub(\{i_1\})=\frac{\rho \theta}{1+\theta},\qquad R_\hub(\{i_t\})=\frac{V}{1+V/\theta},\qquad
R_\hub(\{i_1,i_t\})=\frac{\rho \theta+V}{1+\theta+V/\theta}.
\]
Since $\rho>V$, $R_\hub(\{i_1,i_t\})>\max\{R_\hub(\{i_1\}),R_\hub(\{i_t\})\}$ for all sufficiently large $\theta$, hence
$R^*_\hub(\{i_1,i_t\})=R_\hub(\{i_1,i_t\})$. A first-order expansion gives
$R^*_\hub(\{i_1,i_t\})=\rho-\frac{\rho-V}{\theta}+O(\theta^{-2})$ and $R^*_\hub(\{i_1\})=\rho-\frac{\rho}{\theta}+O(\theta^{-2})$, and
subtracting yields $R^*_\hub(i_t\mid\{i_1\})=\frac{V}{\theta}+O(\theta^{-2})$. Intuitively, the congesting customer's
large weight $\theta$ inflates the hub's denominator, causing the marginal revenue from customer $i_t$ to become negligible.}
\end{myproof}

\begin{proposition}[{\color{blue}Unbounded gap}]\label{prop:hard-unbounded}
{\color{blue}Under the assumptions $\rho>V$ and $\rho>\tfrac{ab}{1+b}$, and writing $K_n:=\rho+(n-1)V+A+p$, the
family $\{\mathcal{I}_n\}_{n\ge2}$ satisfies
\[
\frac{\rho}{K_n}\;\le\;\liminf_{\theta\to\infty}\frac{{\sf ALG}_{\sf greedy}(\mathcal{I}_n)}{{\sf OPT}(\mathcal{I}_n)}
\;\le\;\limsup_{\theta\to\infty}\frac{{\sf ALG}_{\sf greedy}(\mathcal{I}_n)}{{\sf OPT}(\mathcal{I}_n)}
\;\le\;\frac{\rho}{(n-1)V+A}.
\]
Both ends are $\Theta(1/n)$, so the ratio decays at that rate; in particular
$\inf_n\limsup_{\theta\to\infty}{\sf ALG}_{\sf greedy}(\mathcal{I}_n)/{\sf OPT}(\mathcal{I}_n)=0$, so the greedy
algorithm has no constant-factor guarantee for general revenues.}
\end{proposition}
\begin{myproof}
{\color{blue}\emph{Greedy side: $\lim_{\theta\to\infty}{\sf ALG}_{\sf greedy}=\rho$.}
When the congesting customer $i_1$ is processed, all backlogs are empty, so its marginal at supplier $j$ is
$R^*_j(\{i_1\})$: namely $H_\theta$ at the hub, $A$ at the alternate, and $0$ at the decoy. The greedy objective
$\sum_{j\in S}R^*_j(i_1\mid\emptyset)\,\phi_{i_1}(j,S)$ is maximized by $S=\{\hub\}$: the decoy marginal is
zero, and since $u_{i_1\hub}=u_{i_1\alt}=\theta$, offering $\{\hub\}$ yields $H_\theta\tfrac{\theta}{1+\theta}$ while offering
$\{\hub,\alt\}$ yields $(H_\theta+A)\tfrac{\theta}{1+2\theta}$, the former being larger for all sufficiently large $\theta$
because $H_\theta$ approaches $\rho>A$. The greedy algorithm therefore offers $\{\hub\}$, and $i_1$ selects it with probability
$\tfrac{\theta}{1+\theta}$ and the outside option otherwise. We condition on this choice.
Let $Z$ denote the realized total revenue of the greedy algorithm,
so that ${\sf ALG}_{\sf greedy}=\mathbb{E}[Z]$.
Conditional on the history before processing customer $i_t$, the
greedy objective
\[
\sum_{j\in S_{i_t}}
R^*_j(i_t\mid C^{t-1}_j)\,\phi_{i_t}(j,S_{i_t})
\]
is the expected increase in the total optimal supplier revenue
$\sum_{j\in\mathcal{S}}R^*_j(C_j)$.
Here $C^{t-1}_j$ is supplier $j$'s backlog before processing $i_t$.
The marginal increases telescope over the customer-processing
steps, and their expected sum equals $\mathbb{E}[Z]$.

\emph{Case 1 ($i_1$ joins the hub; probability $\tfrac{\theta}{1+\theta}$).} The hub then contains $i_1$, and we show
by induction that no light customer ever joins it, so that the hub backlog stays $\{i_1\}$ and contributes
$R^*_\hub(\{i_1\})=H_\theta$ throughout. Suppose that when light customer $i_t$ is processed the hub backlog is
$\{i_1\}$ and the decoy holds $k\le n-2$ light customers, which is the case for $i_2$ with $k=0$. The marginal of
$i_t$ at the alternate is $0$. At the decoy, $R_\dec$ is increasing
in the number of light customers offered, so $R^*_\dec$ agrees with $R_\dec$ and the marginal of $i_t$ is
\[
\Delta_\dec(k)\;=\;\frac{(k+1)p\kappa}{1+(k+1)\kappa}-\frac{kp\kappa}{1+k\kappa}\;=\;\frac{p\kappa}{(1+k\kappa)(1+(k+1)\kappa)},
\]
which decreases in $k$, so that $\Delta_\dec(n-2)\le\Delta_\dec(k)\le\Delta_\dec(0)=D$; since $n$, $p$ and $\kappa$
are fixed, both ends are positive constants. At the hub the marginal is $\tfrac{V}{\theta}+O(\theta^{-2})$ by
Lemma~\ref{lem:marginal-collapse}. With
$u_{i_t\hub}=\theta$ and $u_{i_t\dec}=\theta^{-1/2}$, the three relevant assortments give
\begin{align*}
\{\dec\}&:\ \Delta_\dec(k)\,\tfrac{\theta^{-1/2}}{1+\theta^{-1/2}}=\Theta(\theta^{-1/2}),
\qquad
\{\hub\}:\ \Bigl(\tfrac{V}{\theta}+O(\theta^{-2})\Bigr)\tfrac{\theta}{1+\theta}=\Theta(\theta^{-1}),\\
\{\dec,\hub\}&:\ \Delta_\dec(k)\,\tfrac{\theta^{-1/2}}{1+\theta^{-1/2}+\theta}
+\Bigl(\tfrac{V}{\theta}+O(\theta^{-2})\Bigr)\tfrac{\theta}{1+\theta^{-1/2}+\theta}=\Theta(\theta^{-1}),
\end{align*}
so for all sufficiently large $\theta$ the greedy algorithm offers $i_t$ the decoy alone. Customer $i_t$ therefore
cannot join the hub, which closes the induction, and its expected contribution is
$\Delta_\dec(k)\,\tfrac{\theta^{-1/2}}{1+\theta^{-1/2}}\le D\,\theta^{-1/2}$. The $n-1$ light customers contribute at
most $(n-1)\,D\,\theta^{-1/2}$ in total, and therefore
$\mathbb{E}\!\left[Z\mid\text{Case 1}\right]\le H_\theta+(n-1)\,D\,\theta^{-1/2}$.

\emph{Case 2 ($i_1$ takes the outside option; probability
$\tfrac{1}{1+\theta}$).} The hub starts empty, and the greedy
algorithm may route light customers to it. We bound the expected
revenue conditional on the final backlogs. For any set $C$,
\[
R_\hub(C)\le \rho+(n-1)V,\qquad
R_\alt(C)\le A,\qquad
R_\dec(C)\le p.
\]
The first inequality follows because the congesting customer's
contribution is at most $\rho\theta/(1+\theta)\le\rho$, and each
light customer's contribution is at most $V$. The other two
inequalities follow from the specified revenue and preference
parameters. Since $R^*_j(C)=\max_{C'\subseteq C}R_j(C')$, the same
bounds hold for $R^*_j$. Thus, for every collection of final
backlogs, the optimal expected revenue from the suppliers is at
most $K_n$. Consequently,
$\mathbb{E}\!\left[Z\mid\text{Case 2}\right]\le K_n$.
The same argument bounds the expected revenue of any policy
by $K_n$.

Combining the two cases,
\[
\mathbb{E}\!\left[Z\right]\;\le\;\tfrac{\theta}{1+\theta}\Bigl(H_\theta+(n-1)\,D\,\theta^{-1/2}\Bigr)\;+\;\tfrac{1}{1+\theta}\,K_n,
\]
and since $n$ is fixed and $H_\theta$ approaches $\rho$, the right-hand side tends to $\rho$; thus
$\limsup_{\theta\to\infty}{\sf ALG}_{\sf greedy}\le\rho$. For the reverse inequality, in Case 1 the hub alone
contributes $H_\theta$ and every other contribution is nonnegative, so
${\sf ALG}_{\sf greedy}\ge\tfrac{\theta}{1+\theta}H_\theta$ and
$\liminf_{\theta\to\infty}{\sf ALG}_{\sf greedy}\ge\rho$. The two bounds give
$\lim_{\theta\to\infty}{\sf ALG}_{\sf greedy}=\rho$.

\emph{Optimal side: $\liminf_{\theta\to\infty}{\sf OPT}\ge(n-1)V+A$ and ${\sf OPT}\le K_n$.} Consider the static policy that offers
$\{\hub\}$ to every light customer and $\{\alt\}$ to the congesting customer. For large $\theta$, each light
customer chooses the hub and the congesting customer chooses the alternate, each with probability
$\tfrac{\theta}{1+\theta}$, which approaches $1$; the backlogs thus become $C_\hub=\{i_2,\dots,i_n\}$ and
$C_\alt=\{i_1\}$. Since every light hub weight $w_{\hub i_t}=V/\theta$ is negligible for large $\theta$, the hub
denominator stays near $1$ and $R^*_\hub(\{i_2,\dots,i_n\})=\frac{(n-1)V}{1+(n-1)V/\theta}$, which approaches
$(n-1)V$, while $R^*_\alt(\{i_1\})=A$. As the optimum is at least the value of this policy,
$\liminf_{\theta\to\infty}{\sf OPT}\ge(n-1)V+A$. For the upper bound, the bound established in Case 2 holds for
any policy, so ${\sf OPT}\le K_n$ for every $\theta$.

Combining these bounds gives the two inequalities of the proposition; with $\rho,V,A,p$ constant both ends are
$\Theta(1/n)$ and vanish as $n$ grows.}
\end{myproof}

{\color{blue}\noindent
Numerical results for the greedy and static algorithms on this
family are reported in Table~\ref{table:hard-instance-numeric}
of Section~\ref{sec:numerics}.}

\section{{\color{blue}Solving~\eqref{eq:problem:primal} via Column Generation}}
\label{appendix:exact-lp}

{\color{blue}In this appendix, we present two column generation algorithms for~\eqref{eq:problem:primal}. Since~\eqref{eq:problem:primal} has exponentially many variables $\lambda_{j,C}$, solving it exactly in polynomial time is challenging. The first algorithm uses the FPTAS of~\citet{chen2025fairassortment} as a pricing oracle and underlies the experiments in Table~\ref{table:revenue-comparison-approx}. The second replaces the FPTAS with a mixed-integer program (MIP) to solve the sub-dual exactly, and underlies Table~\ref{table:revenue-comparison-exact}.}

\vspace{2mm}
\noindent{\color{blue}\textbf{Column Generation Framework.} Column generation maintains, for each supplier $j \in \mathcal{S}$, a restricted collection of columns $\mathcal{V}_j \subseteq 2^{\mathcal{C}}$. At each iteration, the \emph{master problem} solves~\eqref{eq:problem:primal} restricted to those columns, returning a primal solution $(\hat\lambda, \hat x)$ and a dual solution $(\boldsymbol\alpha, \boldsymbol\beta, \boldsymbol\gamma)$. A column $C \notin \mathcal{V}_j$ has positive reduced cost if and only if
\[R_j(C) - \sum_{i \in C} \gamma_{ij} > \beta_j,\]
where $\beta_j$ is the dual variable associated with the
normalization constraint
$\sum_{C\in\mathcal{V}_j}\lambda_{j,C}=1$
for supplier $j$ in the restricted master problem.
Finding the column with the highest reduced cost for supplier $j$ reduces to solving the sub-dual~\eqref{eq:subdual}, which is assortment optimization with fixed costs under MNL. When no improving column exists for any supplier, the current restricted master LP solution is optimal for~\eqref{eq:problem:primal}.}

\subsection{{\color{blue}Approximate Algorithm: Column Generation with FPTAS}}\label{appendix:colgen-fptas}

{\color{blue}Problem~\eqref{eq:subdual} is an instance of assortment optimization with fixed costs under the MNL model, which is NP-hard in general~\citep{kunnumkal2010}. To handle this, we use the FPTAS of~\citet{chen2025fairassortment} as the pricing oracle; for any $\delta > 0$, it returns a $(1-\delta)$-approximate solution to~\eqref{eq:subdual} in polynomial time. Consequently, the stopping criterion in Algorithm~\ref{alg:col-gen-FPTAS} may be triggered before all positive-reduced-cost columns are identified, but the returned solution to~\eqref{eq:problem:primal} is within a $(1-\delta)$ factor of the LP optimum. Using the same argument as in the proof of Lemma~\ref{lemma:approx}, we show that Algorithm~\ref{alg:col-gen-FPTAS}, which uses column generation instead of the Ellipsoid method, also returns a $(1-\delta)$-approximate solution to~\eqref{eq:problem:primal}. In the experiments of Table~\ref{table:revenue-comparison-approx}, we set $\delta = 0.01$ for small instances and $\delta = 0.1$ for larger ones.}

\begin{algorithm}[ht]
\color{blue}
\caption{{\color{blue}$(1-\delta)$-Approximate column generation for~\eqref{eq:problem:primal} using FPTAS for the sub-dual}}
\label{alg:col-gen-FPTAS}
\footnotesize{
\noindent\textbf{Input:} Instance $\mathbf{\Theta} = (\mathbf{r}, \mathbf{u}, \mathbf{w})$, tolerance $\delta > 0$, FPTAS $\mathcal{A}$ for~\eqref{eq:subdual}.\\
\textbf{Output:} $(1-\delta)$-approximate solution $(\hat\lambda, \hat x)$ to~\eqref{eq:problem:primal}.
\begin{enumerate}
  \item \textbf{Initialization.} For each $j\in\mathcal{S}$, set $\mathcal{V}_j = \{\emptyset\}\cup\bigl\{\{i\}: i\in\mathcal{C}\bigr\}$; the empty column keeps the restricted master feasible.
  \item Repeat until the stopping rule is met:
  \begin{enumerate}
      \item \textbf{Master problem.} Solve~\eqref{eq:problem:primal} with $\lambda_{j,C} = 0$ for all $C \notin \mathcal{V}_j$; let $(\hat\lambda, \hat x)$ be the primal solution and $(\boldsymbol{\alpha}, \boldsymbol{\beta}, \boldsymbol{\gamma})$ the dual solution.
      \item \textbf{Sub-dual.} For each $j\in\mathcal{S}$, apply $\mathcal{A}$ to~\eqref{eq:subdual} to obtain $C_j^{\mathcal{A}}$ with $R_j(C_j^{\mathcal{A}}) - \sum_{i \in C_j^{\mathcal{A}}} \gamma_{ij} \geq (1-\delta)\cdot \max_{C \subseteq \mathcal{C}}\bigl\{R_j(C) - \sum_{i \in C}\gamma_{ij}\bigr\}$.
      \item \textbf{Update or stop.} For each $j \in \mathcal{S}$: if $R_j(C_j^{\mathcal{A}}) - \sum_{i \in C_j^{\mathcal{A}}} \gamma_{ij} > \beta_j$, add $C_j^{\mathcal{A}}$ to $\mathcal{V}_j$ and return to Step~2(a). If no such $j$ exists, return $(\hat\lambda, \hat x)$.
  \end{enumerate}
\end{enumerate}
}
\end{algorithm}

\subsection{{\color{blue}Exact Algorithm: Column Generation with MIP}}\label{appendix:colgen-mip}

{\color{blue}Algorithm~\ref{alg:col-gen-FPTAS} returns a $(1-\delta)$-approximate solution rather than an exact one, and the FPTAS of~\citet{chen2025fairassortment} that it calls at each iteration becomes computationally expensive for large instances. We now reformulate~\eqref{eq:subdual} as a mixed-integer program (MIP) that solves it exactly, and which modern branch-and-cut solvers such as Gurobi solve efficiently in practice. In our experiments, this MIP-based oracle is considerably faster than the FPTAS, allowing column generation to scale to much larger instances.}

\vspace{2mm}
\noindent{\color{blue}\textbf{MIP Reformulation.} For supplier $j \in \mathcal{S}$ and dual variables $\boldsymbol{\gamma}$, the sub-dual is:
\[\max_{C \subseteq \mathcal{C}} \left( \frac{\sum_{i \in C} r_{ij} w_{ji}}{1 + \sum_{i \in C} w_{ji}} - \sum_{i \in C} \gamma_{ij} \right).\]
This is a mixed-integer nonlinear program due to the fractional revenue term. We linearize it exactly by introducing three sets of variables: a binary variable $x_i \in \{0,1\}$ indicating whether customer $i$ is included in $C$, a continuous variable $y_0 \geq 0$ representing the outside-option probability mass ($y_0 = \frac{1}{1 + \sum_{i \in \mathcal{C}} w_{ji} x_i}$), and a continuous variable $y_i \geq 0$ for each $i \in \mathcal{C}$ representing the probability that customer $i$ is selected ($y_i = y_0 w_{ji} x_i$). Together, $(y_0, y_1, \ldots, y_n)$ form a probability distribution over the outside option and the $n$ customers.}

{\color{blue}The normalization identity $y_0 \bigl(1 + \sum_{i \in \mathcal{C}} w_{ji} x_i\bigr) = 1$ translates directly to the simplex constraint $y_0 + \sum_{i \in \mathcal{C}} y_i = 1$. Since $x_i \in \{0,1\}$ and $y_0 \in [0,1]$, the nonlinear product $y_i = y_0 w_{ji} x_i$ is exactly linearized by:
\[y_i \leq w_{ji} x_i, \quad y_i \leq w_{ji} y_0, \quad y_i \geq w_{ji}\bigl(y_0 - (1 - x_i)\bigr), \quad \forall i \in \mathcal{C}.\]
These constraints enforce $y_i = 0$ when $x_i = 0$ and $y_i = w_{ji} y_0$ when $x_i = 1$. The expected revenue $R_j(C) = \sum_{i \in C} r_{ij} \phi_j(i,C)$ equals $\sum_{i \in C} r_{ij} y_i$ under this parameterization, so the objective simplifies to $\sum_{i \in \mathcal{C}} r_{ij} y_i - \sum_{i \in \mathcal{C}} \gamma_{ij} x_i$. The resulting exact MIP reformulation is:}

{\color{blue}\fontsize{11pt}{5.5pt}\selectfont
\begin{align}
\begin{aligned}
    \max \quad & \sum_{i\in\mathcal{C}} r_{ij}\, y_i \;-\; \sum_{i\in\mathcal{C}} \gamma_{ij}\, x_i & \\
    \text{s.t.}\quad
        & y_0 + \sum_{i\in\mathcal{C}} y_i = 1, &\notag\\
        & y_i \le w_{ji}\, x_i,  & \forall i \in \mathcal{C},\notag \\
        & y_i \le w_{ji}\, y_0, & \forall i \in \mathcal{C},\notag \\
        & y_i \ge w_{ji}\,(y_0 - (1 - x_i)), & \forall i \in \mathcal{C},\notag \\
        & x_i \in \{0,1\},\quad y_0 \ge 0,\quad y_i \ge 0, & \forall i \in \mathcal{C}. \notag
\end{aligned}
\label{eq:MIP-subdual}\tag{$\textsc{MIP-Sub}_j$}
\end{align}}
{\color{blue}\noindent The optimal value of~\eqref{eq:MIP-subdual} equals $\max_{C \subseteq \mathcal{C}}\{R_j(C) - \sum_{i \in C}\gamma_{ij}\}$, and can be computed in practice using any standard branch-and-cut solver such as Gurobi or CPLEX. The full column generation procedure is given in Algorithm~\ref{alg:col-gen-MIP}. It follows the same structure as Algorithm~\ref{alg:col-gen-FPTAS}, with the only modification being that~\eqref{eq:MIP-subdual} is solved at each iteration in place of the FPTAS. Since the number of subsets of $\mathcal{C}$ is finite, the algorithm terminates in a finite number of iterations and returns the optimal solution $(\lambda^*, x^*)$ to~\eqref{eq:problem:primal}.}
\begin{algorithm}[ht]
\color{blue}
\caption{{\color{blue}Exact column generation for~\eqref{eq:problem:primal} using MIP for the sub-dual}}
\label{alg:col-gen-MIP}
\footnotesize{
\noindent\textbf{Input:} Instance $\mathbf{\Theta} = (\mathbf{r}, \mathbf{u}, \mathbf{w})$, MIP solver for~\eqref{eq:MIP-subdual}.\\
\textbf{Output:} Exact optimal solution $(\lambda^*, x^*)$ to~\eqref{eq:problem:primal}.
\begin{enumerate}
  \item \textbf{Initialization.} For each $j\in\mathcal{S}$, set $\mathcal{V}_j = \{\emptyset\}\cup\bigl\{\{i\}: i\in\mathcal{C}\bigr\}$; the empty column keeps the restricted master feasible.
  \item Repeat until the stopping rule is met:
  \begin{enumerate}
      \item \textbf{Master problem.} Solve~\eqref{eq:problem:primal} with $\lambda_{j,C} = 0$ for all $C \notin \mathcal{V}_j$; let $(\lambda^*, x^*)$ be the primal solution and $(\boldsymbol{\alpha}, \boldsymbol{\beta}, \boldsymbol{\gamma})$ the dual solution.
      \item \textbf{Sub-dual.} For each $j\in\mathcal{S}$, solve~\eqref{eq:MIP-subdual} to obtain $C_j^* \in \arg\max_{C \subseteq \mathcal{C}}\bigl\{R_j(C) - \sum_{i \in C}\gamma_{ij}\bigr\}$.
      \item \textbf{Update or stop.} For each $j \in \mathcal{S}$: if $R_j(C_j^*) - \sum_{i \in C_j^*}\gamma_{ij} > \beta_j$, add $C_j^*$ to $\mathcal{V}_j$ and return to Step 2(a). If no such $j$ exists, return $(\lambda^*, x^*)$.
  \end{enumerate}
\end{enumerate}
}
\end{algorithm}

\section{{\color{blue}Marketplace-Calibrated Instance Generation from Upwork Data}}
\label{appendix:upwork-generation}

{\color{blue}
In this appendix, we describe how we construct the Upwork-calibrated instances used in Section~\ref{sec:numerics-upwork}. The construction uses two public datasets. We derive the revenue matrix $\mathbf{R}$ from job and freelancer attributes, including the imputed job lengths defined below. We construct the preference matrices $\mathbf{U}$ and $\mathbf{W}$ from skill profiles generated from each agent's record by a large language model (LLM).

\vspace{2mm}
\noindent\textbf{Two-sided interpretation of Upwork.}
We model Upwork as a two-sided market. Freelancers are the customers and form the initiating side. Each posted job represents the client who posted it and forms a supplier on the responding side. A freelancer observes an assortment of jobs and applies to at most one. Each job then observes its set of applicants, which forms its backlog, and selects at most one freelancer. A match between freelancer $i$ and job $j$ yields platform revenue $r_{ij}$, equal to the service fee on the resulting contract.

\vspace{2mm}
\noindent\textbf{Data sources.}
We use two publicly available datasets. The \emph{job postings} dataset\footnote{\url{https://www.kaggle.com/datasets/asaniczka/upwork-job-postings-dataset-2024-50k-records}} contains approximately $50{,}000$ Upwork job postings. For each posting, it reports a title, a free-text description, a posted budget, and, for hourly postings, an hourly range. The \emph{freelancer} dataset\footnote{\url{https://www.kaggle.com/datasets/shohinurpervezshohan/freelancer-earnings-and-job-trends}} reports each freelancer's hourly rate, job success rate, client rating, rehire rate, experience level, job category, and typical job duration. This dataset combines records from several freelancing platforms. We therefore use it only as a source for the empirical distribution of freelancer attributes and do not assume that its records come exclusively from Upwork.

\vspace{2mm}
\noindent\textbf{Preprocessing.}
For each job posting $j$, we obtain its category and budget $B_j$ from the corresponding structured fields when available. If either field is absent, we extract its value from the description. We discard postings with a missing or non-positive budget; the remaining postings form the pool of candidate suppliers. For each freelancer $i$, we convert the hourly rate $p_i$ and job success rate to numeric values and record the freelancer's category. We discard records with a missing hourly rate or job success rate; the remaining records form the pool of candidate customers. Because the datasets use different category vocabularies, we map both vocabularies to five categories: software, design, writing, marketing, and data support. A pair $(i,j)$ is \emph{category-compatible} if freelancer $i$ and job $j$ are assigned to the same category.

\vspace{2mm}
\noindent\textbf{Job length.}
The job-postings dataset does not report engagement lengths. For each job $j$, we sample a \texttt{job\_duration\_days} value from the freelancer records assigned to $j$'s category and multiply it by $8$ to define $L_j$ in hours. If no freelancer record is assigned to $j$'s category, we sample the value from the full pool of freelancer records. This procedure assigns each job a separate draw from the observed job-duration values.

\vspace{2mm}
\noindent\textbf{Feasibility.}
We determine feasibility using category compatibility and affordability. A pair $(i,j)$ is feasible if freelancer $i$ and job $j$ are category-compatible and
$p_i \le \frac{B_j}{L_j}$,
where $B_j/L_j$ is the implied hourly cap for job $j$. Because the two datasets are unlinked and record monetary variables on different scales, we apply a single global rescaling to the freelancer rates before evaluating this condition. For each infeasible pair, we set $r_{ij}=u_{ij}=w_{ji}=0$, which excludes the pair from every assortment and backlog. Category incompatibility sets all cross-category entries of $\mathbf{R}$, $\mathbf{U}$, and $\mathbf{W}$ to zero, while the affordability condition may produce additional zero entries within each category.

\vspace{2mm}
\noindent\textbf{Revenue.}
We set the platform's service fee to $\mathrm{fee}=10\%$ of the contract value. For an engagement of $L_j$ hours at the freelancer's rate $p_i$, the platform revenue is
\[
r_{ij} =
\begin{cases}
\mathrm{fee}\cdot p_i \cdot L_j & \text{if } (i,j) \text{ is feasible},\\
0 & \text{otherwise}.
\end{cases}
\]
Before setting the revenues of infeasible pairs to zero, the revenues factorize as
$r_{ij}=(\mathrm{fee}\cdot p_i)L_j
=r_i^{\mathcal C}r_j^{\mathcal S}$ and satisfy the
same-order condition. Setting these revenues to zero can violate this common order. We therefore evaluate the greedy algorithm using the heuristic processing order described in Section~\ref{sec:exp-setup}.

\vspace{2mm}
\noindent\textbf{LLM skill profiles.}
We derive the two-sided preference weights from language-model skill profiles of the agents. For each agent, we send one query to Llama~3.1 (8B), a large language model (LLM) served locally. The model returns an integer skill vector over five categories: software, design, writing, marketing, and data support. The vector $\bm{v}\in\{0,1,\dots,5\}^{5}$ has one coordinate for each category, and each coordinate is an integer from $0$ to $5$. For a freelancer $i$, the query evaluates the freelancer's expertise in each category and returns $\bm{v}^{\mathcal{C}}_i$. For a job $j$, the query evaluates how strongly the posting requires each category and returns $\bm{v}^{\mathcal{S}}_j$. The model receives each agent's full record. A freelancer's record contains its category, experience level, region, completed jobs, earnings, rate, success and rehire rates, and rating. A job's record contains its title, parsed category and skill list, and free-text description. Because we query each agent once, the procedure uses $n+m$ queries rather than one query for each pair. We query each agent separately and do not carry information across queries, so the profiles are generated independently. We set the sampling temperature to $0$, so the responses are deterministic and the same record always yields the same profile.

\vspace{2mm}
\noindent\textbf{Two-sided preferences.}
We construct the preference weights from two quantities: skill match and price share. We measure the skill match between freelancer $i$ and job $j$ by the cosine similarity of their skill vectors:
\[
\mathrm{fit}_{ij}=\frac{\langle \bm{v}^{\mathcal{C}}_i,\,\bm{v}^{\mathcal{S}}_j\rangle}{\lVert \bm{v}^{\mathcal{C}}_i\rVert\,\lVert \bm{v}^{\mathcal{S}}_j\rVert}\in[0,1].
\]
For each feasible pair, the price share is $\pi_{ij}=p_i L_j/B_j\in[0,1]$, which equals the freelancer's earnings divided by the job's budget. The amount $p_iL_j$ is the freelancer's earnings and the client's cost. A freelancer favors jobs that match its skills and offer a larger price share, whereas a client favors freelancers that match the required skills and have a smaller price share. Accordingly, we define
\[
\tilde{u}_{ij}=\tfrac12\bigl(\mathrm{fit}_{ij}+\pi_{ij}\bigr),
\qquad
\tilde{w}_{ij}=\tfrac12\bigl(\mathrm{fit}_{ij}+(1-\pi_{ij})\bigr).
\]
Both valuations lie in $[0,1]$, and the price share enters with opposite signs on the two sides. We map these valuations to the MNL weights used by the solver:
\[
u_{ij}=\mathds{1}^{\mathrm{feas}}_{ij}\cdot e^{\beta_1\tilde{u}_{ij}-\beta_2},
\qquad
w_{ji}=\mathds{1}^{\mathrm{feas}}_{ij}\cdot e^{\gamma_1\tilde{w}_{ij}-\gamma_2},
\]
where $\mathds{1}^{\mathrm{feas}}_{ij}$ is the feasibility indicator. We set $\beta_1=\gamma_1=2$ and $\beta_2=\gamma_2=\tfrac12$. The constants $\beta_2$ and $\gamma_2$ shift the weights relative to the no-purchase option, whose weight is normalized to one. The feasible-pair weights lie in $[e^{-1/2},e^{3/2}]$ and equal zero on infeasible pairs, including every cross-category or unaffordable pair.

\vspace{2mm}
\noindent\textbf{Instance sampling.}
We compute the LLM skill profiles once for the base pool and reuse them across all sampled sub-instances. Each instance used in Section~\ref{sec:numerics-upwork} consists of $n$ freelancers and $m$ jobs sampled uniformly at random from the relevant pools.
}

\end{APPENDICES}
\end{document}